\documentclass{LMCS}
 
\def\doi{8 (1:09) 2012}
\lmcsheading%
{\doi}
{1--47}
{}
{}
{Dec.~\phantom01, 2009}
{Feb.~23, 2012}
{}

\usepackage{color}
\usepackage{amsmath,amssymb,stmaryrd}
\usepackage{tikz}
\usetikzlibrary{arrows,automata,plotmarks}
\usepackage{enumerate,hyperref}

\newtheorem{definition}{Definition}
\newtheorem{proposition}{Proposition}
\newtheorem{theorem}{Theorem}
\newtheorem{corollary}{Corollary}
\newtheorem{notation}{Notation}

\newenvironment{example}{\vspace*{0ex}\noindent\textbf{Example.}}{\hfill$\square$\vspace*{1ex}}
\definecolor{Grey}{rgb}{0.9,0.9,0.9}

\renewcommand{\i}{\iota}

\newcommand{\cal}{\mathcal}

\newcommand{\always}{\mbox{{\scriptsize\mbox{$\square$}}}}

\newcommand{\eventuallyPast}{\mbox{$\diamond\!\!\!\!\;\cdot\,$}}
\newcommand{\previously}{\mbox{$\circ\!\!\!\:\:\!\!\cdot\,$}}
\newcommand{\since}{\,{\cal S}\,}

\newcommand{\parametric}[2]{\Lambda{#1}\,.\,{#2}}

\newcommand{\partialf}[2]{[{#1}\!\rightharpoondown\!{#2}]}
\newcommand{\totalf}[2]{[{#1}\stackrel{}{\!\rightarrow\!}{#2}]}
\newcommand{\Dom}{\textsf{\footnotesize Dom}}

\newcommand{\paramevents}{{\cal E}\langle X \rangle}

\newcommand{\A}{\mathbb{A}\langle X \rangle}
\newcommand{\T}{\mathbb{T}}
\newcommand{\B}{\mathbb{B}\langle X \rangle}
\newcommand{\C}{\mathbb{C}\langle X \rangle}

\newcommand{\XV}{\partialf{X}{V}}

\newcommand{\property}[1]{\text{\sc #1}}

\newcommand{\unsafemap}{\property{UnsafeMapIter}}
\newcommand{\unsafeiter}{\property{UnsafeIter}}

\newcommand{\hasnext}{\property{HasNext}}

\newcommand{\unsafesynccoll}{\property{UnsafeSyncColl}}
\newcommand{\unsafesyncmap}{\property{UnsafeSyncMap}}

\newcommand{\idea}[1]{
\vspace*{1ex}

\noindent
\rule{0.48\textwidth}{1pt}
\textit{#1}
\\
\noindent
\rule{0.48\textwidth}{1pt}

\vspace*{1ex}
}
\newcommand{\comment}[1]{}

\begin{document}

\title[Parametric Monitoring]{Semantics and Algorithms for Parametric
  Monitoring\rsuper*} 
\author[G.~Ro{\c s}u]{Grigore Ro{\c s}u\rsuper a}
\address{{\lsuper{a,b}}Department of Computer Science, University of Illinois at Urbana-Champaign}
\email{grosu@illinois.edu}
\thanks{{\lsuper a}Supported in part by NSF grants CCF-0448501, CNS-0509321,
  CNS-0720512 and CCF-0916893, and by NASA contract NNL08AA23C.}
\author[F.~Chen]{Feng Chen\rsuper b}

\titlecomment{{\lsuper*}A preliminary version of this paper was published 
as a 2008 technical report \cite{rosu-chen-2008-tr-a}, and an extended
abstract was presented at TACAS 2009 and published in its conference proceedings
\cite{chen-rosu-2009-tacas}.
}

\begin{abstract}
Analysis of execution traces plays a fundamental role in many program
analysis approaches, such as runtime verification, testing, monitoring,
and specification mining.  Execution traces are frequently parametric,
i.e., they contain events with parameter bindings.  Each parametric
trace usually consists of many meaningful {\em trace slices} merged
together, each slice corresponding to one parameter binding.  For
example, a Java program creating iterator objects $i_1$ and $i_2$ over
collection object $c_1$ may yield a trace 
$\textsf{\scriptsize createIter}\langle c_1\,i_1\! \rangle\ 
\textsf{\scriptsize next}\langle i_1 \!\rangle\ 
\textsf{\scriptsize createIter}\langle c_1\,i_2 \!\rangle\ 
\textsf{\scriptsize updateColl}\langle c_1 \!\rangle\ 
\textsf{\scriptsize next}\langle i_1 \!\rangle
$
parametric in collection $c$ and iterator $i$, whose slices
corresponding to instances ``$c,i\mapsto c_1,i_1$''
and ``$c,i\mapsto c_1,i_2$'' are
$\textsf{\scriptsize createIter}\langle c_1\,i_1 \rangle\ 
\textsf{\scriptsize next}\langle i_1 \rangle\ 
\textsf{\scriptsize updateColl}\langle c_1 \rangle\ 
\textsf{\scriptsize next}\langle i_1 \rangle
$
and, respectively,
$\textsf{\scriptsize createIter}\langle c_1\,i_2 \rangle\ 
\textsf{\scriptsize updateColl}\langle c_1 \rangle$.
Several approaches have been proposed to specify and dynamically
analyze parametric properties, but they have limitations: some in the
specification formalism, others in the type of trace they support.
Not unexpectedly, the existing approaches share common notions, intuitions,
and even techniques and algorithms, suggesting that a fundamental study
and understanding of parametric trace analysis is necessary.

This foundational paper aims at giving a semantics-based solution to
parametric trace analysis that is unrestricted by the type of
parametric property or trace that can be analyzed.  Our approach is
based on a rigorous understanding of {\em what} a parametric
trace/property/monitor is and {\em how} it relates to its
non-parametric counter-part.  A general-purpose parametric trace
slicing technique is introduced, which takes each event in the
parametric trace and dispatches it to its corresponding trace
slices.  This parametric trace slicing technique can be used in
combination with any conventional, non-parametric trace analysis
technique, by applying the later on each trace slice.  As an instance,
a parametric property monitoring technique is then presented, which
processes each trace slice online.  Thanks to the generality of
parametric trace slicing, the parametric property monitoring technique
reduces to encapsulating and indexing unrestricted and well-understood
non-parametric property monitors (e.g., finite or push-down
automata).

The presented parametric trace slicing and monitoring techniques 
have been implemented and extensively evaluated.  Measurements of
runtime overhead confirm that the generality of the discussed
techniques does not come at a performance expense when compared with
existing parametric trace monitoring systems.
\end{abstract}

\maketitle

\section{Introduction and Motivation}
\label{sec:intro}

\comment{
Computing systems in general and programs in particular can be
regarded as ``generators of execution traces'', that is, as running
devices that yield relevant events; the ``outside environment'' that
perceives the events can be, depending upon the application,
almost anything: from memory modules concerned with fine-grained
location read and write events, to an API concerned with function
call and return events, to soups of concurrent processes concerned
with message send and receive events.
Ultimately, a computing system or program can be identified with the
execution traces that it can produce, since those traces precisely
capture the actual system behavior, as a user or an external observer
of the system sees it.
}

Parametric traces, i.e., traces containing events with parameter
bindings, abound in program executions, because 
they naturally appear whenever abstract parameters (e.g., variable
names) are bound to concrete data (e.g., heap objects) at runtime.
In this section we first discuss some motivating examples and describe
the problem addressed in this paper, then we recall related work and put
the work in this paper in context, and then we explain our contributions
and finally the structure of the paper.

\subsection{Motivating examples and highlights}
\label{sec:motivating-examples}
We here describe three examples of parametric properties, in increasing
difficulty order, and use them to highlight and motivate the semantic results
and the algorithms presented in the rest of the paper.

Typestates~\cite{strom-yemeni-1986-tse} refine the notion of type by stating
not only what operations are allowed by a particular object, but also what
operations are allowed in what contexts.
Typestates are particular parametric properties with only one parameter.
Figure \ref{fig:typestate} shows the typestate description for a property saying that it
is invalid to call the \textsf{\footnotesize next()} method on an iterator object when there
are no more elements in the underlying collection, i.e., when \textsf{\footnotesize hasnext()} returns false,
or when it is unknown if there are more elements in the collection, i.e., \textsf{\footnotesize hasnext()}
is not called.  From the {\it unknown} state, it is always an error to call the \textsf{\footnotesize next()} method
because such an operation could be unsafe.  If \textsf{\footnotesize hasnext()} is called and returns true,
it is safe to call \textsf{\footnotesize next()}, so the typestate enters the {\it more} state.
If, however, the \textsf{\footnotesize hasnext()} method returns false, there are no more elements, and
the typestate enters the {\it none} state.  In the {\it more}
and {\it none} states, calling the \textsf{\footnotesize hasnext()} method provides no new information.
It is safe to call \textsf{\footnotesize next()} from the {\it more} state, but it becomes unknown if
more elements exist, so the typestate reenters the initial {\it unknown} state.
Finally, calling \textsf{\footnotesize next()} from the {\it none} state results in an error.
For simplicity, we here assume that \textsf{\footnotesize next()} is the only means to modify the state
of an iterator; concurrent modifications are discussed in other examples shortly.

\newcommand{\ang}{55}
\newcommand{\loopang}{0}
\newcommand{\shift}{\vspace{-1ex}\hspace{4ex}}
\newcommand{\lastshift}{\vspace{-1ex}\hspace{2ex}}
\newcommand{\loopshift}{\vspace{40ex}\hspace{0ex}}

\begin{figure}
\begin{center}
\scalebox{1}{
\thicklines
\hspace{1pt}
\begin{tikzpicture}[->,>=stealth',shorten >=1pt,node distance=1.2in, auto,
                    semithick]
  \definecolor{grey}{RGB}{192,192,192}
  \tikzstyle{every state}=[fill=grey,draw,text=black, minimum size=15mm]

  \node[state,initial]      (A)              {\footnotesize{\it unknown}};
  \node[state]              (B) [below left of=A] {\footnotesize{\it more}};
  \node[state]              (C) [below right of=A] {\footnotesize{\it none}};
  \node[state,accepting]    (D) [below left of=C]  {\footnotesize{\it error}};

  \path (A) edge [bend right] node {\hspace{-1.5in}\textsf{\footnotesize hasnext()} == true} (B)
        (A) edge node {\hspace{-.05in}\textsf{\footnotesize hasnext()} == false} (C)
        (A) edge  node {\hspace{0in}\textsf{\footnotesize next()}} (D)
        (B) edge [loop left] node {\rotatebox{90}{\textsf{\footnotesize hasnext()}}} (B) 
        (B) edge [bend right] node {\textsf{\footnotesize next()}} (A)
        (C) edge [loop right] node {\rotatebox{-90}{\textsf{\footnotesize hasnext()}}} (C)
        (C) edge node {{\hspace{-.05in}\textsf{\footnotesize next()}}} (D)
;
\end{tikzpicture}
}
\end{center}
\caption{Typestate property describing the correct use of the \textsf{\footnotesize next()} and
\textsf{\footnotesize hasnext()} methods.}
\label{fig:typestate}
\end{figure}

It is straightforward to represent the typestate property in Figure~\ref{fig:typestate},
and all typestate properties, as particular (one-parameter) \emph{parametric} properties.  
Indeed, the behaviors described by the typestate in Figure~\ref{fig:typestate} are
intended to be obeyed by all iterator object instances; that is, we have a property parametric in the iterator.
To make this more precise, let us look at the problem from the perspective of observable program
{\em execution traces}.  A trace can be
regarded as a sequence of {\em events} relevant to the property of interest, in our case calls to
\textsf{\footnotesize next()} or to \textsf{\footnotesize hasnext()}; the latter can be further split into
two categories, one when \textsf{\footnotesize hasnext()} returns true and the other when it returns false.
Since the individuality of each iterator matters, we must regard each event as being {\em parametric}
in the iterator yielding it.  Formally, traces relevant to our typestate property are formed with three
parametric events, namely
\textsf{\footnotesize next$\langle i \rangle$},
\textsf{\footnotesize hasnexttrue$\langle i \rangle$},
and
\textsf{\footnotesize hasnextfalse$\langle i \rangle$}.
A possible trace can be
\textsf{\footnotesize hasnexttrue$\langle i_1 \rangle$}
\textsf{\footnotesize hasnextfalse$\langle i_2 \rangle$}
\textsf{\footnotesize next$\langle i_1 \rangle$}
\textsf{\footnotesize next$\langle i_2 \rangle$}...,
which violates the typestate property for iterator instance $i_2$.
How to obtain execution traces is not our concern here
(several runtime monitoring systems use AspectJ instrumentation).
Our results in this paper are concerned with how to specify properties over parametric traces,
what is their meaning, and how to monitor them.

Let us first briefly discuss our approach to specifying properties over parametric execution traces, that is,
{\em parametric properties}.  To keep it as easy as possible for the user and to leverage our knowledge
on specifying ordinary, non-parametric properties, we build our specification approach on top of {\em any}
formalism for specifying non-parametric properties.  More precisely, all one has to do is to first specify the
property using any conventional formalism as if there were only one possible instance of its parameters,
and then use a special $\Lambda$ quantifier to make it parametric.  For our typestate example,  suppose
that \textsf{\footnotesize typestate} is the finite state machine in Figure~\ref{fig:typestate}, modified by replacing
the method calls on edges with actual events as described above.  Then the desired parametric property
is $\parametric{i}{\textsf{\footnotesize typestate}}$.  The meaning of this parametric property is that whatever was intended for
its non-parametric counterpart, \textsf{\footnotesize typestate}, must hold {\em for each} parameter instance;
that is, the {\em error} state must not be reached for any iterator instance, which is precisely the desired meaning
of this typestate.  Another way to specify the same property is using a regular expression matching all the good behaviors,
each bad prefix for any instance thus signaling a violation:
$$
\parametric{i}{
(
\textsf{\footnotesize hasnexttrue$\langle i \rangle$}^+
\textsf{\footnotesize next$\langle i \rangle$}
\mid
\textsf{\footnotesize hasnextfalse$\langle i \rangle$}^*
)^*
}
$$
Yet another way to specify its non-parametric part is with a linear-temporal logic formula:
$$
\parametric{i}{
\always
(
\textsf{\footnotesize next$\langle i \rangle$}
\implies
\previously
\textsf{\footnotesize hasnexttrue$\langle i \rangle$}
)
}
$$
The above LTL formula says ``it is always ($\always$) the case that each
\textsf{\footnotesize next$\langle i \rangle$} event is preceded ($\previously$) by
a \textsf{\footnotesize hasnexttrue$\langle i \rangle$}.
The LTL formula must hold for any iterator instance.
In general, if $\parametric{X}{P}$ is a parametric property, where $X$ is a set of one or more
parameters, we may call $P$ its corresponding {\em base} or {\em root} or
{\em non-parametric property}.
In this paper we develop a mathematical foundation for specifying such parametric properties
independently of the formalism used for specifying their non-parametric part, define their
precise semantics, provide algorithms for online monitoring of parametric properties, and
finally bring empirical evidence showing that monitoring parametric properties is in fact feasible.

Parametric properties properly \emph{generalize} typestates in two different directions.
First, parametric properties allow more than one parameter, allowing us to specify not only
properties about a given object such as the typestate example above, but also properties
that capture \emph{relationships~between~objects}.  Second, they allow us to specify 
infinite-state root properties using formalisms like context-free grammars (see Sections~\ref{sec:ex-cfg},
\ref{sec:safe-resource-safe-client} and \ref{sec:success-ratio}).

Let us now consider a two-parameter property.
Suppose that one is interested in analyzing collections and
iterators in Java.  Then execution traces of interest may contain
events $\textsf{\footnotesize createIter}\langle c\,i\rangle$ (iterator $i$ is
created for collections $c$), $\textsf{\footnotesize updateColl}\langle c\rangle$
($c$ is modified), and $\textsf{\footnotesize next}\langle i\rangle$ ($i$ is
accessed using its next element method), instantiated for particular
collection and iterator instances.  Most properties of parametric traces
are also parametric; for our example, a property may be ``collections are not
allowed to change while accessed through iterators'', which is
parametric in a collection {\em and} an iterator.  The
parametric property above expressed as a regular expression
(here matches mean violations) can be
$$
\parametric{c,i}{\ \textsf{\footnotesize createIter}\langle c\,i\rangle\ 
                   \textsf{\footnotesize next}\langle i\rangle^*\ 
                   \textsf{\footnotesize updateColl}\langle c\rangle^+\ 
                   \textsf{\footnotesize next}\langle i\rangle}
$$
From here on, when we know the number and types of parameters of each event, we omit writing
them in parametric properties, because they are redundant; for example, we write
$$
\parametric{c,i}{\ \textsf{\footnotesize createIter}\ 
                   \textsf{\footnotesize next}^*\ 
                   \textsf{\footnotesize updateColl}^+\ 
                   \textsf{\footnotesize next}}
$$

Parametric properties, unfortunately, are
very hard to formally verify and validate against real systems, mainly
because of their dynamic nature and potentially huge or even unlimited
number of parameter bindings.  
Let us extend the above example:
in Java, one may create a collection from a map and use the
collection's iterator to operate on the map's elements.
A similar safety property is: ``maps are not allowed to change while
accessed indirectly through iterators''.
Its violation pattern is:
$$
\parametric{m,c,i}{ 
\begin{array}[t]{l}
\textsf{\footnotesize createColl}\ 
(\textsf{\footnotesize updateMap} \mid \textsf{\footnotesize updateColl})^*\
\textsf{\footnotesize createIter}\
\textsf{\footnotesize next}^*
(\textsf{\footnotesize updateMap} \mid \textsf{\footnotesize updateColl})^+\ 
\textsf{\footnotesize next}
\end{array}
}
$$
with two new parametric events $\textsf{\footnotesize createColl}\langle m\, c\rangle$
(collection $c$ is created from map $m$) and $\textsf{\footnotesize updateMap}\langle m \rangle$ ($m$ is updated).
All the events used in this property provide only partial parameter bindings ($\textsf{\footnotesize createColl}$ binds
only $m$ and $c$, etc.), and parameter bindings carried by different events may be combined into
larger bindings; e.g., $\textsf{\footnotesize createColl}\langle m_1\, c_1\rangle$
can be combined with $\textsf{\footnotesize createIter}\langle c_1\, i_1\rangle$ into a full binding
$\langle m_1\, c_1\, i_1\rangle$, and also with $\textsf{\footnotesize createIter}\langle c_1\, i_2\rangle$
into $\langle m_1\, c_1\, i_2\rangle$.  It is highly challenging for a trace analysis technique
to correctly and efficiently maintain, locate and combine trace slices for different parameter
bindings, especially when the trace is long and the number of parameter bindings is large.

This paper addresses the problem of parametric trace analysis from a
foundational, semantic perspective: 
\begin{quote}\it
Given a parametric trace $\tau$ and a
parametric property $\parametric{X}{P}$, what does it mean for $\tau$
to be a good or a bad trace for $\parametric{X}{P}$?  How can we show
it?  How can we leverage, to the parametric case, our knowledge and
techniques to analyze conventional, non-parametric traces against 
conventional, non-parametric properties?
\end{quote}
In this paper we first formulate and then rigorously answer and empirically
validate our answer to these questions, in the context of runtime verification.
In doing so, a technique for trace slicing is also presented and shown
correct, which we regard as one of the central results in parametric
trace analysis.  In short, our overall approach to monitor a parametric property 
$\parametric{X}{P}$ is to observe the parametric trace as it is being generated
by the running system, slice it online with respect to the used parameter instances,
and then send each slice piece-wise to a non-parametric monitor corresponding
to the base property $P$; this way, multiple monitor instances for $P$ can and
typically do coexist, one for each trace slice.

The main conceptual limitation of our approach is that the parametrization of
properties is only allowed at the top-level, that is, the base property $P$ in
the parametric property $\parametric{X}{P}$ cannot have any $\Lambda$
binders.  In other words, we do not consider nested parameters.  To allow
nested parameters one needs a syntax for
properties, so that one can incorporate the syntax for parameters within the
syntax for properties.  However, one of our major goals is to be
formalism-independent, which means that, by the nature of the problem that
we are attempting to solve, we can only parameterize properties at the top.
Many runtime verification approaches deliberately accept the same limitation,
as discussed below, because arbitrarily nested parameters are harder to
understand and turn out to generate higher runtime overhead in the systems
supporting them.

Our concrete contributions are explained after the related work.

\subsection{Related Work}

We here discuss several major approaches that have been proposed so
far to specify and monitor parametric properties, and relate them to
our subsequent results in this paper.  It is worth mentioning upfront that,
except for the MOP approach \cite{meredith-jin-griffith-chen-rosu-2010-jsttt}
which motivated and inspired the work in
this paper, the existing approaches do {\em not} follow the general
methodology proposed by our approach in this paper.  More precisely, they
employ a monolithic monitoring approach, where one monitor is associated
to each parametric property; the monitor receives and handles each
parametric event in a formalism-specific way.  In contrast, our approach
is to generate multiple local monitors, each keeping track of one
parameter instance.  Our approach leads not only to a lower runtime
overhead as empirically shown in Section~\ref{sec:implementation}, but
it also allows us to separate concerns (i.e., the parameter handling from
the specification formalism and monitor synthesis for the basic property)
and thus potentially enabling a broader spectrum of optimizations
that work for various different property specification formalisms and
corresponding monitors.

Tracematches
\cite{tracematches-oopsla,oopsla07abc} is an
extension of AspectJ \cite{aspectj} supporting parametric regular
patterns; when patterns are matched during the execution, user-defined
advice can be triggered.  J-LO \cite{jlo} is a variation of
Tracematches that supports a first-order extension variant of linear
temporal logic (LTL) that supports data parametrization by means of quantifiers
\cite{stolz-2006-rv}; the user-provided actions are executed when the
LTL properties are violated.  Also based on AspectJ, \cite{lsc-monitor}
proposes Live Sequence Charts (LSC) \cite{damm01lscs} as an
inter-object scenario-based specification formalism; LSC is implicitly
parametric, requiring parameter bindings at runtime.
Tracematches, J-LO and LSC \cite{lsc-monitor} support a limited number
of parameters, and each has its own approach to handle
parameters, specific to its particular specification
formalism.  Our semantics-based approach in this paper is generic in 
the specification formalism and admits, in theory, a potentially unlimited
number of parameters.  In spite of the generality of our theoretical results,
we chose in our current implementations (see Section \ref{sec:implementation})
to also support only a bounded number of parameters, like in the aforementioned
approaches.

JavaMOP~\cite{meredith-jin-griffith-chen-rosu-2010-jsttt,chen-rosu-2007-oopsla}
({\tt http://javamop.org}) is a parametric specification and monitoring system
that is generic in the specification formalism for base properties, each formalism
being included as a logic plugin.  Monitoring code is generated from parametric
specifications and woven within the original Java program, also using AspectJ,
but using a different approach that allows it to encapsulate monitors for
non-parametric properties as blackboxes.  Until recently, JavaMOP's genericity
came at a price: it could only monitor execution traces in which the first event
in each slice instantiated all the property parameters.  This limitation
prevented the JavaMOP system presented in \cite{chen-rosu-2007-oopsla}
from monitoring some basic parametric properties, including ones discussed
in this paper.  Our novel approach to parametric trace slicing and monitoring
discussed in this paper does not have that limitation anymore.  The parametric
slicing and monitoring technique discussed in this paper has been incorporated
both in JavaMOP \cite{meredith-jin-griffith-chen-rosu-2010-jsttt} and in its
commercial-grade successor RV \cite{meredith-rosu-2010-rv}, together with several optimizations that we
do not discuss here; Section~\ref{sec:implementation} discusses
experiments done with both these systems, as well as with Tracematches,
for comparison, because Tracematches has proven to be the most efficient
runtime verification system besides JavaMOP.

Program Query Language (PQL) \cite{pql-oopsla} allows the
specification and monitoring of parametric context-free grammar (CFG)
patterns.  Unlike the approaches above that only allow a bounded
number of property parameters, PQL can associate parameters 
with sub-patterns that can be recursively matched at runtime, yielding
a potentially unbounded number of parameters.  PQL's approach to
parametric monitoring is specific to its particular CFG-based
specification formalism.  Also, PQL's design does not support
arbitrary execution traces.  For example, field updates and method
begins are not observable; to circumvent the latter, PQL allows
for observing traces local to method calls.  Like PQL, our technique
also allows an unlimited number of parameters (but as mentioned above,
our current implementation supports only a bounded number of
parameters).  Unlike PQL, our semantics and techniques are not limited
to particular events, and are generic in the property specification
formalism; CFGs are just one such possible formalism.

Eagle \cite{DBLP:conf/vmcai/BarringerGHS04}, RuleR
\cite{DBLP:conf/rv/BarringerRH07},  and Program Trace Query Language
(PTQL) \cite{ptql-oopsla} are very general trace specification and
monitoring systems, whose specification formalisms allow complex
properties with parameter bindings anywhere in the specification (not
only at the beginning, like we do).  Eagle and RuleR are based on
fixed-point logics and rewrite rules, while PTQL is based on SQL
relational queries.  These systems tackle a different aspect of
generality than we do: they attempt to define general specification
formalisms supporting data binding among many other features, while we
attempt to define a general parameterization approach that is
logic-independent.
As discussed in
\cite{oopsla07abc,meredith-jin-chen-rosu-2008-ase,chen-rosu-2007-oopsla}
(Eagle and PQL cases), the very general specification formalisms
tend to be slower; this is not surprising, as the more general the
formalism the less the potential for optimizations.  Our techniques can
be used as an optimization for certain common types
of properties expressible in these systems: use any of these to
specify the base property $P$, then use our generic techniques to
analyze $\parametric{X}{P}$.


\subsection{Contributions}
Besides proposing a formal semantics to parametric traces,
properties, and monitoring, we make two theoretical contributions
and discuss implementations that validate them empirically:
\begin{enumerate}[(1)]
\item Our first result is a general-purpose online parametric trace
slicing algorithm (algorithm $\A$ in Section \ref{sec:trace-slicing})
together with its proof of correctness (Theorem \ref{thm:trace-slicing}),
positively answering the following question: {\em given a parametric
execution trace, can one effectively find the slices corresponding to
each parameter instance without having to traverse the trace for each
instance?}
\item Our second result, building upon the slicing algorithm,
is an online monitoring technique (algorithms $\B$ and $\C$ in Section
\ref{sec:trace-monitoring}) together with its proof of correctness
(Theorems \ref{thm:trace-monitoring} and \ref{thm:implementation}),
which positively answers the following question: {\em is it possible to
monitor arbitrary parametric properties $\parametric{X}{P}$ against
parametric traces, provided that the root property $P$ is monitorable
using conventional monitors?}
\item Finally, our implementation of these techniques in the
JavaMOP and RV systems positively answers the following question:
{\em can we implement general purpose and unrestricted parametric
property monitoring tools which are comparable in performance with
or even outperform existing parametric property monitoring tools on
the restricted types of properties and/or traces that the latter support?}
\end{enumerate}

\noindent Preliminary results reported in this paper have been published in
a less polished form as a technical report in summer 2008
\cite{rosu-chen-2008-tr-a}.  Then a shorter, conference paper was
presented at TACAS 2009 in York, U.K. \cite{chen-rosu-2009-tacas}.
This extended paper differs from \cite{chen-rosu-2009-tacas} as follows:
\begin{enumerate}[(1)]
\item It defines all the mathematical infrastructure needed to prove
the results claimed in \cite{chen-rosu-2009-tacas}.  For example,
Section \ref{sec:math} is new.
\item It expands the results in \cite{chen-rosu-2009-tacas} and
includes all their proofs, as well as additional results needed for
those proofs.  For example, Section \ref{sec:parametric-trace} is new.
\item It discusses more examples of parametric properties.  For
example, Section \ref{sec:examples} is new.
\item The implementation section in \cite{chen-rosu-2009-tacas}
presented an incipient implementation of our technique in a prototype
system called PMon there (from \underline{P}arametric
\underline{Mon}itoring).  In the meanwhile, we have implemented the
technique described in this paper as an integral part of the runtime
verification systems JavaMOP ({\tt http://javamop.org}) and RV \cite{meredith-rosu-2010-rv}.  The
implementation section (Section\ref{sec:implementation}) now refers
to these systems.
\end{enumerate}

\subsection{Paper Structure}
 Section \ref{sec:examples} discusses examples of parametric properties.
Section \ref{sec:math} provides the mathematical
background needed to formalize the concepts introduced later in the
paper.
Section \ref{sec:parametric-trace-property} formalizes parametric
events, traces and properties, and defines trace slicing. 
Section \ref{sec:parametric-trace} establishes a tight connection
between the parameter instances in a trace and the parameter instance
used for slicing.  Sections \ref{sec:trace-slicing}, \ref{sec:monitors} and
\ref{sec:trace-monitoring} discuss our main techniques for parametric
trace slicing and monitoring, and prove them correct.  Section
\ref{sec:implementation} discusses implementations of these
techniques in two related systems, JavaMOP and RV.
Section~\ref{sec:conclusion} concludes and proposes future work.

\section{Examples of Parametric Properties}
\label{sec:examples}

In this section we discuss several examples of parametric properties.
Our purpose here is twofold.  On the one hand we give the
reader additional intuition and motivation for the subsequent
semantics and algorithms, and, on the other hand, we justify the
generality of our approach with respect to the employed specification
formalism for trace properties.  The discussed examples of parametric
properties are defined using various trace specification formalisms, some
with more than one parameter and some with more than validating and/or
violating categories of behaviors.  For each of the examples, we give hints
on how our subsequent techniques in Sections~\ref{sec:trace-slicing} and
\ref{sec:trace-monitoring} work.  In order to explain the examples in this
section we also informally introduce necessary notions, such as events and
traces (both parametric and non-parametric); all these notions will be
formally defined in Section~\ref{sec:parametric-trace-property}.

For each example, we also discuss which of the existing runtime
verification systems can support it.  Note that
JavaMOP~\cite{meredith-jin-griffith-chen-rosu-2010-jsttt}
and its commercial-grade successor RV \cite{meredith-rosu-2010-rv}, which build upon the
trace slicing and monitoring techniques presented in this paper,
are the only runtime verification systems that support all the
parametric properties discussed below.

\subsection{Releasing acquired resources}
\label{sec:acquire-release}

Consider a certain type of resource (e.g., synchronization objects) that can
be acquired and released by a given procedure, and suppose that we want
the resources of this type to always be explicitly released by the procedure
whenever acquired and only then.
This example will be broken in subparts and used as a running example in
Section~\ref{sec:parametric-trace-property} to introduce our main notions
and notations.

Let us first consider the non-parametric case in which
we have only one resource.  Supposing that the four events of interest, i.e.,
the begin/end of the procedure and the acquire/release of the resource, are
${\cal E} = \{\textsf{\footnotesize begin},\textsf{\footnotesize end},\textsf{\footnotesize acquire},\textsf{\footnotesize release}\}$,
then the following regular pattern $P$ captures the desired behavior requirements:
$$P = (\textsf{\footnotesize begin} (\epsilon\ |\ 
     (\textsf{\footnotesize acquire} 
       (\textsf{\footnotesize acquire} \ |\  \textsf{\footnotesize release})^* 
     \textsf{\footnotesize release})) \textsf{\footnotesize end})^*$$
The above regular pattern states that the procedure can take place
multiple times and, if the resource is acquired then it is released by the end
of the procedure ($\epsilon$ is the empty word).  For simplicity, we here assume
that the procedure is not recursive and that the resource can be acquired and
released multiple times, with the effect of acquiring and respectively releasing it
precisely once; Section \ref{sec:ex-cfg} shows how to use a context-free pattern
to specify possibly recursive procedures with matched acquire/release events within
each procedure invocation.  One matching execution trace for this property is, e.g.,
$\textsf{\footnotesize begin}\ \textsf{\footnotesize acquire}\ \textsf{\footnotesize acquire}\
\textsf{\footnotesize release}\ \textsf{\footnotesize end}\ \textsf{\footnotesize begin}\ \textsf{\footnotesize end}$.

Let us now consider the parametric case in which we may have more than one
resource and we want each of them to obey the requirements specified above.
Now the events \textsf{\footnotesize acquire} and \textsf{\footnotesize release} are parametric in the
resource being acquired or released, that is, they have the form
$\textsf{\footnotesize acquire}\langle r_1\rangle$,
$\textsf{\footnotesize release}\langle r_2\rangle$, etc.
The begin/end events take no parameters, so we write them
$\textsf{\footnotesize begin}\langle \rangle$ and
$\textsf{\footnotesize end}\langle \rangle$.
A parametric trace $\tau$ for our running example can be the following:
$$
\begin{array}{@{}r@{\ }c@{\ }l@{}}
\tau & = &
\textsf{\footnotesize begin}\langle\rangle\,
\textsf{\footnotesize acquire}\langle r_1\rangle\,
\textsf{\footnotesize acquire}\langle r_2\rangle\, 
\textsf{\footnotesize acquire}\langle r_1\rangle\, 
\textsf{\footnotesize release}\langle r_1\rangle\,
\textsf{\footnotesize end}\langle\rangle\,
\textsf{\footnotesize begin}\langle\rangle\, 
\textsf{\footnotesize acquire}\langle r_2\rangle\,
\textsf{\footnotesize release}\langle r_2\rangle\, 
\textsf{\footnotesize end}\langle\rangle
\end{array}
$$
This trace involves two resources, $r_1$ and $r_2$, and it really
consists of \textit{two trace slices}  merged together, one for each resource:
$$
\begin{array}{l}
\langle r_1\rangle: \ \ 
\textsf{\footnotesize begin}\ 
\textsf{\footnotesize acquire}\
\textsf{\footnotesize acquire}\ 
\textsf{\footnotesize release}\
\textsf{\footnotesize end}\
\textsf{\footnotesize begin}\ 
\textsf{\footnotesize end}
\\
\langle r_2\rangle: \ \ 
\textsf{\footnotesize begin}\ 
\textsf{\footnotesize acquire}\ 
\textsf{\footnotesize end}\
\textsf{\footnotesize begin}\
\textsf{\footnotesize acquire}\
\textsf{\footnotesize release}\ 
\textsf{\footnotesize end}
\end{array}
$$
The begin and end events belong to both trace
slices.  Since we know the parameter instance for
each trace slice and we know the types of parameters for each event,
to avoid clutter we do not mention the redundant parameter bindings
of events in trace slices.

Our trace slicing algorithm discussed in Section~\ref{sec:trace-slicing}
processes the parametric trace only once, traversing it from the
first parametric event to the last, incrementally calculating
a collection of meaningful trace slices so that it can quickly identify and
report the slice corresponding to any parameter instance when requested.

Note that the $\langle r_1\rangle$ trace slice matches the
specification $P$ above, while the $\langle r_2\rangle$ trace slice does not.
To distinguish parametric properties referring to multiple trace slices from 
ordinary properties, we explicitly list the parameters using a special $\Lambda$
binder.  For example, our property above parametric in the resource $r$ is
$\parametric{r}{P}$, or
$$\parametric{r}{(\textsf{\footnotesize begin}\ (\epsilon \ |\  
     (\textsf{\footnotesize acquire}\ 
       (\textsf{\footnotesize acquire} \ |\ \textsf{\footnotesize release})^*\ 
     \textsf{\footnotesize release}))\ \textsf{\footnotesize end})^*}$$
Both Tracematches~\cite{tracematches-oopsla,oopsla07abc}
and JavaMOP~\cite{meredith-jin-griffith-chen-rosu-2010-jsttt}
can specify/monitor such parametric regular properties,
the latter using its extended-regular expression (ERE) plugin.

For the sake of a terminology, $P$ is called a non-parametric, or a root,
or a basic property, in contrast to $\parametric{r}{P}$, which is called a
parametric property.
As detailed in Section~\ref{sec:parametric-trace-property},
parametric properties are functions taking a parametric trace (e.g., $\tau$)
and a parameter instance (e.g., $r\mapsto r_1$ or $r \mapsto r_2$) into a verdict
category for the basic property $P$ (e.g., \textsf{\footnotesize match} or \textsf{\footnotesize fail}).
In our case, the semantics of our parametric property $\parametric{r}{P}$
takes parametric trace $\tau$ and parameter instance $r\mapsto r_1$ to
\textsf{\footnotesize match}, and takes $\tau$ and $r\mapsto r_2$ to \textsf{\footnotesize fail}, that is,
$$
\begin{array}{l}
(\parametric{r}{P})(\tau)(r\mapsto r_1) = \textsf{\footnotesize match} \\
(\parametric{r}{P})(\tau)(r\mapsto r_2) = \textsf{\footnotesize fail}
\end{array}
$$
Our parametric monitoring algorithm in Section~\ref{sec:trace-monitoring}
reports a \textsf{\footnotesize fail} for instance $r \mapsto r_2$ precisely when the first
$\textsf{\footnotesize end}$ event is encountered.

We would like to make two observations at this stage.  First, as we already mentioned,
we only parameterize a property at the top, that is, the $\Lambda$ binder cannot be used
inside the basic property.  Indeed, since we do not enforce any particular syntax for
basic properties, it is not clear how to mix the $\Lambda$ binder with
the inexistent property constructs.  Second, one should not confuse our parameters
with universally quantified variables.  While in our example above $\Lambda$ may
feel like a universal quantifier, note that one may prefer to specify the same
parametric property in a more negative fashion, for example to specify the bad
behaviors instead of the positive ones.  Relying on the fact that the
\textsf{\footnotesize begin} and \textsf{\footnotesize end} events must be correctly matched, one can only
state the bad patterns, which are a \textsf{\footnotesize begin} followed by a \textsf{\footnotesize release} and
an \textsf{\footnotesize acquire} followed by an \textsf{\footnotesize end}:
$$
\parametric{r}{
({\cal E}^*(\textsf{\footnotesize begin}\,\textsf{\footnotesize release}\mid\textsf{\footnotesize acquire}\,\textsf{\footnotesize end})\,{\cal E}^*)}
$$
The right way to regard a parametric property is as one indexed by all possible
instances of the parameters, each instance having its own interpretation of the trace
(only caring of the events relevant to it), which is orthogonal to the other instances'
interpretations.

\subsection{Authenticate before use}
\label{sec:authenticate}
Consider a server authenticating and using keys, say
$k_1$, $k_2$, $k_3$, etc., whose execution traces contain events
$\textsf{\footnotesize authenticate}\langle k_1 \rangle$,
$\textsf{\footnotesize use}\langle k_2 \rangle$, etc.
A possible trace of such a system can be
$$
\tau = \textsf{\footnotesize authenticate}\langle k_1 \rangle\ 
\textsf{\footnotesize authenticate}\langle k_3 \rangle\ 
\textsf{\footnotesize use}\langle k_3 \rangle\  
\textsf{\footnotesize use}\langle k_2 \rangle\ 
\textsf{\footnotesize authenticate}\langle k_2 \rangle\
\textsf{\footnotesize use}\langle k_1 \rangle\
\textsf{\footnotesize use}\langle k_2 \rangle\ 
\textsf{\footnotesize use}\langle k_3 \rangle
$$
A parametric property for such a system can be ``each key
must be authenticated before use'', which, using linear temporal logic
(LTL) as a specification formalism for the corresponding base property, can
be expressed as
$$\parametric{k}{\always(\textsf{\footnotesize use} \rightarrow \eventuallyPast
\textsf{\footnotesize authenticate})}$$
Such parametric LTL properties can be expressed in both J-LO
\cite{jlo} and JavaMOP \cite{meredith-jin-griffith-chen-rosu-2010-jsttt,chen-rosu-2007-oopsla}
(the later using its LTL logic plugin).  For the trace above, the trace slice corresponding
to $k_3$ is $\textsf{\footnotesize authenticate}\ \textsf{\footnotesize use}\ \textsf{\footnotesize use}$
corresponding to the parametric subtrace
$\textsf{\footnotesize authenticate}\langle k_3 \rangle\ 
\textsf{\footnotesize use}\langle k_3 \rangle\
\textsf{\footnotesize use}\langle k_3 \rangle$ of events relevant to $k_3$ in
$\tau$, but keeping only the base events; also, the trace slice
corresponding to $k_2$ is $\textsf{\footnotesize use}\ \textsf{\footnotesize authenticate}\ \textsf{\footnotesize use}$.
Our trace slicing algorithmin Section \ref{sec:trace-slicing}
can detect these slices.  Moreover, with the finite trace LTL semantics
in \cite{rosu-havelund-2005-jase},
$$
\begin{array}{l}
{(\parametric{k}{\always(\textsf{\footnotesize use} \rightarrow \eventuallyPast
\textsf{\footnotesize authenticate})})(\tau)(k \mapsto k_3)}=\textsf{\footnotesize true} \\
{(\parametric{k}{\always(\textsf{\footnotesize use} \rightarrow \eventuallyPast
\textsf{\footnotesize authenticate})})(\tau)(k \mapsto k_2)=\textsf{\footnotesize false}}
\end{array}
$$
Our parametric monitoring algorithm in Section
\ref{sec:trace-monitoring} reports a violation for
instance $k \mapsto k_2$ precisely when the first
$\textsf{\footnotesize use}\langle k_2\rangle$ is encountered.

\subsection{Safe iterators}
\label{ex:safe-iterators}
Consider the following property for iterators created over vectors:
when an iterator is created for a vector, one is not allowed to modify
the vector while its elements are traversed using the iterator.  The JVM
usually throws a runtime exception when this occurs, but the exception
is not guaranteed in a multi-threaded environment.
Supposing that parametric event
$\textsf{\footnotesize create}\langle v\,i\rangle$ is generated when iterator $i$ is
created for vector $v$, $\textsf{\footnotesize update}\langle v\rangle$ is generated
when $v$ is modified, and $\textsf{\footnotesize next}\langle i\rangle$ is
generated when $i$ is accessed using its ``next element'' interface,
then one can write it as the parametric regular property
$$
\parametric{v,i}{\textsf{\footnotesize create}\ \textsf{\footnotesize next}^*\ \textsf{\footnotesize update}^+\ \textsf{\footnotesize next}}.
$$
Such parametric regular expression properties can be expressed in both
Tracematches \cite{tracematches-oopsla} and JavaMOP
\cite{meredith-jin-griffith-chen-rosu-2010-jsttt,chen-rosu-2007-oopsla}
(the latter using its ERE plugin).
We here assumed that the matching of the regular expression corresponds to
violation of the base property.  Thus, the parametric
property is violated by a given trace and a given parameter instance
whenever the regular pattern above is matched by the corresponding
trace slice.  For example,
if $\tau = \textsf{\footnotesize create}\langle v_1\,i_1 \rangle\ 
\textsf{\footnotesize next}\langle i_1 \rangle\ 
\textsf{\footnotesize create}\langle v_1\,i_2 \rangle\ 
\textsf{\footnotesize update}\langle v_1 \rangle\ 
\textsf{\footnotesize next}\langle i_1 \rangle
$ is a parametric trace where two iterators are created for a
vector, then the slice corresponding to $\langle v_1\,i_1\rangle$
is
$\textsf{\footnotesize create}\
\textsf{\footnotesize next}\
\textsf{\footnotesize update}\
\textsf{\footnotesize next}\
$
and the one corresponding to $\langle v_1\,i_2\rangle$ is
$\textsf{\footnotesize create}\
\textsf{\footnotesize update}$,
so $\tau$ violates the parametric property (i.e., matches the
regular pattern above) on instance $\langle{v_1\,i_1}\rangle$, but not on
instance $\langle{v_1\,i_2}\rangle$.  Note that in this example there are more
than one parameters in events, traces and property, namely a vector
and an iterator.  Indeed, the main difficulty of our techniques in
Sections \ref{sec:trace-slicing} and \ref{sec:trace-monitoring} was
precisely to handle general purpose parametric properties with an
arbitrary number of parameters.  The slicing algorithm in
Section \ref{sec:trace-slicing} processes parametric traces and
maintains enough slicing information so that, when asked to produce
slices corresponding to particular parameter instances, e.g.,
to $\langle v_1\,i_2 \rangle$, it can do so without any further analysis of the
trace.  Also, in this case, the monitoring algorithm in Section~\ref{sec:trace-monitoring}
reports a match each time a parameter instance yields a matching trace slice.

\subsection{Correct locking}
\label{sec:ex-cfg}
Consider a custom implementation of synchronization in
which one can acquire and release locks manually (like in Java 5 and later versions).  A
basic property is that each function
releases each lock as many times as it acquires it.  Assuming that
the executing code is always inside some function (like in Java, C,
etc.), that $\textsf{\footnotesize begin}\langle\rangle$ and
$\textsf{\footnotesize end}\langle\rangle$ events are generated
whenever function executions are started and terminated, respectively,
and that $\textsf{\footnotesize acquire}\langle l\rangle$ and
$\textsf{\footnotesize release}\langle l\rangle$ events are generated
whenever lock $l$ is acquired or released, respectively, then one
can specify this safety property using the following parametric
context-free grammar (CFG) pattern:
$$
\parametric{l}{S \rightarrow S\ \textsf{\footnotesize begin}\ S\ \textsf{\footnotesize end} \mid
  S\ \textsf{\footnotesize acquire}\ S\ \textsf{\footnotesize release} \mid \epsilon}
$$
Such parametric CFG properties can be expressed in both PQL~\cite{pql-oopsla} and
JavaMOP~\cite{meredith-jin-griffith-chen-rosu-2010-jsttt,meredith-jin-chen-rosu-2008-ase}
(the later using its CFG plugin).
We here borrow the CFG property semantics of the CFG plugin of
JavaMOP (and also RV \cite{meredith-rosu-2010-rv}) in \cite{meredith-jin-chen-rosu-2008-ase}, that is, this
parametric property is violated by a parametric execution with a
given parameter instance (i.e., concrete lock) whenever the
corresponding trace slice cannot be completed into one accepted by the
grammar language.  While this property can be expressed in JavaMOP and
even monitored in its non-parametric form, the previous implementation
of JavaMOP in \cite{meredith-jin-chen-rosu-2008-ase} cannot monitor it
as a parametric property because its violating traces most likely start
with a property-relevant $\textsf{\footnotesize begin}\langle\rangle$ event, which
does not contain a lock parameter; therefore, the previous limitation
of JavaMOP (allowing only events that instantiate all property's
parameters to create a monitor instance) did not allow us to monitor
this natural CFG property.  To circumvent this limitation,
\cite{meredith-jin-chen-rosu-2008-ase} proposed a different way to
specify this property, in which the violating traces
started with an $\textsf{\footnotesize acquire}\langle l\rangle$ event.
We do not need such artificial encodings anymore in the new version
of JavaMOP, and they were never needed in RV (RV improves the new JavaMOP).

For profiling reasons, one may also want to take notice of
validations, or matches of the property, as well as matches followed
by violation, etc.; one can therefore have different interpretations
of CFG patterns as base properties, classifying traces into various
categories.  What is different in this example, compared to the
previous ones, is that the non-parametric property cannot be
implemented as a finite state machine.  With the CFG monitoring
algorithm proposed in \cite{meredith-jin-chen-rosu-2008-ase} used to
monitor the base property, our parametric monitoring algorithm in
Section \ref{sec:trace-monitoring} reports a violation of this
parametric CFG property as soon as a parameter instance is detected
for which the corresponding trace slice has no future, that is, it
admits no continuation into a trace in the language of the grammar.

\subsection{Safe resource use by safe client}
\label{sec:safe-resource-safe-client}
A client can use a resource only within a given procedure and, when
that happens, both the client and the resource must have been
previously authenticated as part of that procedure.  Assuming the
procedure fixed and given, this is a property over traces with five
types of events: begin and end of the procedure
($\textsf{\footnotesize begin}\langle\rangle$ and
$\textsf{\footnotesize end}\langle\rangle$), authenticate of client
($\textsf{\footnotesize auth-client}\langle c \rangle$) or of resource
($\textsf{\footnotesize auth-resource}\langle r \rangle$), and use of
resource by client ($\textsf{\footnotesize use}\langle r\,c \rangle$).  Using the
past time linear temporal logic with calls and returns (ptCaRet) in
\cite{rosu-chen-ball-2008-rv}, one would write it as follows:
$$
\parametric{r,c}
{
\begin{array}[t]{l}
\textsf{\footnotesize use} \rightarrow ((\overline{\eventuallyPast}\,\textsf{\footnotesize begin})
\ \wedge \, \neg (
  (\neg \textsf{\footnotesize auth-client})\ \overline{\since}\ \textsf{\footnotesize begin}) 
\
\wedge \, \neg (
  (\neg \textsf{\footnotesize auth-resource}) \ \overline{\since}\ \textsf{\footnotesize begin}
))
\end{array}
}
$$
The overlined operators are abstract variants of temporal operators,
in the sense that they are defined on traces that collapse terminated procedure
calls (erase subtraces bounded by matched begin/end events).  For
example, ``$\overline{\eventuallyPast}\,\textsf{\footnotesize begin}$'' holds only
within a procedure call, because all the nested and terminated
procedure calls are abstracted away.  In words, the above says: if one
sees the use of the resource (\textsf{\footnotesize use}) then that must take place
within the procedure and it is not the case that one did not see,
within the main procedure since its latest invocation, the
authentication of the client or the authentication of the resource.

JavaMOP can express this property using its ptCaRet logic plugin
\cite{rosu-chen-ball-2008-rv}.  However, until recently \cite{meredith-jin-griffith-chen-rosu-2010-jsttt},
JavaMOP could again only monitor it in its non-parametric form, because of its
previous limitation allowing only completely parameterized events to
create monitors.  Even though it may appear that this property can
only be violated when a completely parameterized
$\textsf{\footnotesize use}\langle r\,c \rangle$ event is observed, in fact, the
monitor must already exist at that point in the execution and ``know''
whether the client and the resource have authenticated since the begin
of the current procedure; all the other events involved in the
property are incompletely parameterized, so, unfortunately, this
parametric property could not be monitored using the previous JavaMOP
system, but it can be monitored with the new one.

\subsection{Success ratio}
\label{sec:success-ratio}
Consider now parametric traces with events
$\textsf{\footnotesize success}\langle a \rangle$ and
$\textsf{\footnotesize fail}\langle a \rangle$, saying whether a certain action $a$
was successful or not.  For a given action, a meaningful property can
classify its (non-parametric) traces into an infinite number of
categories, each representing a success ratio of the given action,
which is a (rational) number $s/t$ between 0 and 1, where $s$
is the number of success events in the trace and $t$ is the total
number of events in the trace.  Then the corresponding parametric
property over such parametric traces gives a success ratio for each
action.  We can specify such a property in JavaMOP and RV by making use of
monitor variables and event actions
\cite{meredith-jin-griffith-chen-rosu-2010-jsttt}.  Indeed, one can add
two monitor variables, $s$ and $t$, and then increment $t$ in each event
action and increment $s$ only in the event action of the \textsf{\footnotesize success}
event.  The underlying parametric monitoring algorithm keeps separation
between the various $s$ and $t$ monitor variables, one such pair for each distinct action $a$,
guaranteeing that the correct ones will be accessed.

\section{Mathematical Background: Partial Functions, Least Upper Bounds (lubs) and lub Closures}
\label{sec:math}

In this section we first discuss some basic notions of partial
functions and least upper bounds of them, then we introduce least
upper bounds of sets of partial functions and least upper bound
closures of sets of partial functions.  This section is rather
mathematical.  We need these mathematical notions because it turns out
that parameter instances are partial maps from the domain of
parameters to the domain of parameter values.  As shown later,
whenever a new parametric event is observed, it needs to be dispatched
to the interested parts (trace slices or monitors), and those parts
updated accordingly: these informal operations can be rigorously
formalized as existence of least upper bounds and least upper bound
closures over parameter instances, i.e., partial functions.

We recommend the reader who is only interested in our algorithms but
not in the details of our subsequent proofs, to read the first two
definitions and then jump to Section
\ref{sec:parametric-trace-property}, returning to this section for
more mathematical background only when needed.

\subsection{Partial Functions}
\label{sec:partial-functions}

This section discusses partial functions and (least) upper bounds
over sets of partial functions.  The notions and the results discussed
in this section are broadly known, and many of their properties are
folklore.  They can be found in one shape or another in virtually any
book on denotational semantics or domain theory.  Since we need
only a small subset of notions and results on partial functions and
(least) upper bounds in this paper, and since we need to fix a uniform
notation anyway, we prefer to define and prove everything we need
and hereby also make our paper self-contained.

We think of partial functions as ``information carriers'': if a
partial function $\theta$ is defined on an element $x$ of its domain,
then ``$\theta$ carries the information $\theta(x)$ about $x\in X$''.
Some partial functions can carry more information than others; two or
more partial functions can, together, carry compatible information, but
can also carry incompatible information (when two or more of them
disagree on the information they carry for a particular $x\in X$).

\begin{definition}
\label{dfn:partial-functions}
We let $\totalf{X}{V}$ and $\XV$ denote the sets of
\textbf{total} and of \textbf{partial functions} from $X$ to $V$,
respectively.  The domain of $\theta\in\XV$ is the set
$\Dom(\theta)=\{x \in X \mid \theta(x) \textit{ is defined}\}$.
Let $\bot\in\XV$ be the map undefined everywhere, that
is, $\Dom(\bot)=\emptyset$.  If $\theta,\theta'\in\XV$
then:
\begin{enumerate}[(1)]
\item $\theta$ and $\theta'$ are \textbf{are compatible} if and only
if $\theta(x)=\theta'(x)$ for any $x\in\Dom(\theta)\cap\Dom(\theta')$;
\item \textbf{$\theta$ is less informative than $\theta'$},
written $\theta\sqsubseteq\theta'$, if for any $x\in X$, $\theta(x)$
defined implies $\theta'(x)$ also defined and $\theta'(x)=\theta(x)$;
\item $\theta$ is \textbf{strictly less informative} than $\theta'$, written
$\theta\sqsubset\theta'$, when $\theta\sqsubseteq\theta'$ and
$\theta\neq\theta'$.
\end{enumerate}
\end{definition}

The relation of compatibility is reflexive and symmetric, but not
transitive.  When $\theta,\theta'\in\XV$ are compatible, we let
$\theta\sqcup\theta'\in\XV$ denote the partial function whose domain
is $\Dom(\theta)\cup\Dom(\theta')$ and which is defined as $\theta$ or
$\theta'$ in each element in its domain.  The partial function
$\theta\sqcup\theta'$ is called the least upper bound of $\theta$ and
$\theta'$.  We define least upper bounds more generally below.  Also,
note that $\theta\sqsubseteq\theta'$ and, respectively,
$\theta\sqsubset\theta'$, iff $\theta,\theta'$ compatible and
$\Dom(\theta)\subseteq\Dom(\theta')$ and, respectively,
$\Dom(\theta)\subset\Dom(\theta')$.

\begin{definition}
Given $\Theta\subseteq\XV$ and $\theta'\in\XV$,
\begin{enumerate}[(1)]
\item  $\theta'$ is an \textbf{upper bound} of $\Theta$
  iff $\theta\sqsubseteq\theta'$ for any $\theta\in\Theta$; $\Theta$
  \textbf{has upper bounds} iff there is a $\theta'$ which is an
  upper bound of $\Theta$;
\item  $\theta'$ is the \textbf{least upper bound (\textbf{lub})} of
  $\Theta$ iff $\theta'$ is an upper bound of 
  $\Theta$ and $\theta'\sqsubseteq\theta''$ for any other upper bound
  $\theta''$ of $\Theta$;
\item $\theta'$ is the \textbf{maximum (max)} of $\Theta$ iff
  $\theta'\in\Theta$ and $\theta'$ is a lub of $\Theta$.
\end{enumerate}
\end{definition}

A set of partial functions has an upper bound iff the
partial functions in the set are pairwise {compatible}, that is, no two
of them disagree on the value of any particular element in their
domain.
Least upper bounds and maximums may not always exist for any
$\Theta\subseteq\XV$; if a lub or a maximum
for  $\Theta$ exists, then it is, of course, unique, because $\sqsubseteq$ is
a partial order, so antisymmetric.

It is known that $(\XV,\sqsubseteq,\bot$) is a
complete (i.e., any $\sqsubseteq$-chain has a least upper bound)
partial order with bottom (i.e., $\bot$).

\begin{definition}
Given $\Theta\subseteq\XV$, let $\sqcup\Theta$ and
$\max\Theta$ be the lub and the max of $\Theta$,
respectively, when they exist.  When $\Theta$ is finite, one may write
$\theta_1\sqcup\theta_2\sqcup\cdots\sqcup\theta_n$ instead of
$\sqcup\{\theta_1,\theta_2,\dots,\theta_n\}$.
\end{definition}

If $\Theta$ has a maximum, then it also has a lub and
$\sqcup\Theta=\max\Theta$.  Here are several common properties that
we use frequently in Sections~\ref{sec:lubs-of-families} and \ref{sec:lub-closure}
(these sections will present less known results with specific to our particular approach to
parametric slicing and monitoring):

\begin{proposition}
\label{prop:basic-lub}
The following hold ($\theta,\theta_1,\theta_2,\theta_3\in\XV$):
$\bot\sqcup\theta$ exists and $\bot\sqcup\theta=\theta$;
$\theta_1\sqcup\theta_2$ exists iff
  $\theta_2\sqcup\theta_1$ exists, and, if they exist then
$\theta_1\sqcup\theta_2=\theta_2\sqcup\theta_1$;
$\theta_1\sqcup(\theta_2\sqcup\theta_3)$ exists iff
  $(\theta_1\sqcup\theta_2)\sqcup\theta_3$ exists, and if they exist
  then $\theta_1\sqcup(\theta_2\sqcup\theta_3) =
       (\theta_1\sqcup\theta_2)\sqcup\theta_3$.
\end{proposition}

\begin{proposition}
\label{prop:basic-ub}
Let $\Theta\subseteq\XV$.  Then
\begin{enumerate}[(1)]
\item $\Theta$ has an upper bound iff for any
$\theta_1,\theta_2\in\Theta$ and $x\in X$, if $\theta_1(x)$ and
$\theta_2(x)$ are defined then $\theta_1(x)=\theta_2(x)$;
\item If $\Theta$ has an upper bound then $\sqcup\Theta$ exists and,
for any $x\in X$,\\
$$
(\sqcup\Theta)(x) = \left\{\begin{array}{ll}
\mbox{undefined} & \mbox{if $\theta(x)$ is undefined for any
  $\theta\in\Theta$}\\
\theta(x) & \mbox{if there is a $\theta\in\Theta$ with $\theta(x)$ defined.}
\end{array}
\right.
$$
\end{enumerate}
\end{proposition}
\begin{proof}
Since $\Theta$ has an upper bound $\theta'\in\XV$ 
iff $\theta\sqsubseteq\theta'$ for any $\theta\in\Theta$, if
$\theta_1,\theta_2\in\Theta$ and $x\in X$ are such that $\theta_1(x)$ and
$\theta_2(x)$ are defined then $\theta'(x)$ is also defined and
$\theta_1(x)=\theta_2(x)=\theta'(x)$.  Suppose now that
for any $\theta_1,\theta_2\in\Theta$ and $x\in X$, if $\theta_1(x)$
and $\theta_2(x)$ are defined then $\theta_1(x)=\theta_2(x)$.  All we need
to show in order to prove both results is that we can find a lub for
$\Theta$.  Let $\theta'\in\XV$ be defined as follows:
for any $x\in X$, let
$$\theta'(x) = \left\{\begin{array}{ll}
\mbox{undefined} & \mbox{if $\theta(x)$ is undefined for any
  $\theta\in\Theta$}\\
\theta(x) & \mbox{if there is a $\theta\in\Theta$ with $\theta(x)$ defined}
\end{array}
\right.
$$
First, $\theta'$ above is indeed well-defined, because we
assumed that for any $\theta_1,\theta_2\in\Theta$ and $x\in X$, if
$\theta_1(x)$ and $\theta_2(x)$ are defined then
$\theta_1(x)=\theta_2(x)$.  Second, $\theta'$ is an upper
bound for $\Theta$: indeed, if $\theta\in\Theta$ and $x\in X$ such
that $\theta(x)$ is defined, then $\theta'(x)$ is also defined and
$\theta'(x)=\theta(x)$, that is, $\theta\sqsubseteq\theta'$ for any
$\theta\in\Theta$.  Finally, $\theta'$ is a lub for
$\Theta$: if $\theta''$ is another upper bound for $\Theta$ and
$\theta'(x)$ is defined for some $x\in X$, that is, $\theta(x)$ is defined
for some $\theta\in\Theta$ and $\theta'(x)=\theta(x)$, then
$\theta''(x)$ is also defined and $\theta'(x)=\theta(x)$ (as
$\theta\sqsubseteq\theta''$), so $\theta'\sqsubseteq\theta''$. 
\end{proof}


\begin{proposition}
\label{prop:ub}
The following hold:
\begin{enumerate}[(1)]
\item The empty set of partial functions
$\emptyset\subseteq\XV$ has upper bounds and
$\sqcup\emptyset=\bot$;
\item The one-element sets have upper bounds and
$\sqcup\{\theta\}=\theta$ for any $\theta\in\XV$;
\item The bottom ``$\bot$'' does not influence the least upper bounds:
$\sqcup(\{\bot\}\cup\Theta) = \sqcup\Theta$ for any
$\Theta\subseteq\XV$;
\item If $\Theta,\Theta'\subseteq\XV$ such that
$\sqcup\Theta'$ exists and for any $\theta\in\Theta$ there is a
$\theta'\in\Theta'$ with $\theta\sqsubseteq\theta'$, then
$\sqcup\Theta$ exists and $\sqcup\Theta\sqsubseteq\sqcup\Theta'$;
for example, if $\sqcup\Theta'$ exists and $\Theta\subseteq\Theta'$
then $\sqcup\Theta$ exists and $\sqcup\Theta\sqsubseteq\sqcup\Theta'$; 
\item Let $\{\Theta_i\}_{i\in I}$ be a family of sets of partial
functions with $\Theta_i\subseteq\XV$.  Then
$\sqcup\cup\{\Theta_i\mid i\in I\}$ exists
iff
$\sqcup\{\sqcup\Theta_i \mid i\in I\}$ exists, and, if both exist,
$$\sqcup\cup\{\Theta_i\mid i\in I\}=\sqcup\{\sqcup\Theta_i \mid i\in I\}.$$
\end{enumerate}
\end{proposition}
\begin{proof}
\textit{1.}, \textit{2.} and \textit{3.} are straightforward.  For
\textit{4.}, since for each $\theta\in\Theta$ there is some
$\theta'\in\Theta'$ with $\theta\sqsubseteq\theta'$, and since
$\theta'\sqsubseteq\sqcup\Theta'$ for any $\theta'\in\Theta'$, it follows
that $\theta\sqsubseteq\sqcup\Theta'$ for any $\theta\in\Theta$, that
is, that $\sqcup\Theta'$ is an upper bound for $\Theta$.  Therefore,
by Proposition \ref{prop:basic-ub} it follows that $\sqcup\Theta$
exists and $\sqcup\Theta\sqsubseteq\sqcup\Theta'$ (the latter because
$\sqcup\Theta$ is the \textit{least} upper bound of $\Theta$).
We prove \textit{5.} by double implication, each
implication stating that if one of the lub's exist then
the other one also exists and one of the inclusions holds; that indeed
implies that one of the lub's exists if and only if the
other one exists and, if both exist, then they are equal.  Suppose
first that $\sqcup\cup\{\Theta_i\mid i\in I\}$ exists, that is, that
$\cup\{\Theta_i\mid i\in I\}$ has an upper bound, say $u$.  Since 
$\Theta_i\subseteq\cup\{\Theta_i\mid i\in I\}$ for each $i\in I$, it
follows first that each $\Theta_i$ also has $u$ as an upper bound, so
all $\sqcup\Theta_i$ for all $i \in I$ exist, and second by
\textit{4.} above that
$\sqcup\Theta_i\sqsubseteq\sqcup\cup\{\Theta_i\mid i\in I\}$
for each $i\in I$.  Item \textit{4.} above then further implies that
$\sqcup\{\sqcup\Theta_i \mid i\in I\}$ exists and
$\sqcup\{\sqcup\Theta_i \mid i\in I\}\sqsubseteq
 \sqcup\{\sqcup\cup\{\Theta_i\mid i\in I\}\}
=\sqcup\cup\{\Theta_i\mid i\in I\}$
(the last equality follows by \textit{2.} above).
Conversely, suppose now that $\sqcup\{\sqcup\Theta_i \mid i\in I\}$
exists.  Since for each $\theta\in\cup\{\Theta_i\mid i\in I\}$
there is some $i\in I$ such that $\theta\sqsubseteq\sqcup\Theta_i$
(an $i\in I$ such that $\theta\in\Theta_i$), item \textit{4.} above
implies that $\sqcup\cup\{\Theta_i\mid i\in I\}$ also exists and
$\sqcup\cup\{\Theta_i\mid i\in I\} \sqsubseteq
 \sqcup\{\sqcup\Theta_i \mid i\in I\}$.
\end{proof}

\subsection{Least Upper Bounds of Families of Sets of Partial Maps}
\label{sec:lubs-of-families}

Motivated by requirements and optimizations of our trace slicing and
monitoring algorithms in Sections \ref{sec:trace-slicing} and
\ref{sec:trace-monitoring}, in this section and the next we define
several less known notions and results.  We are actually not aware
of other places where these notions are defined, so they could be
novel and specific to our approach to parametric trace slicing and
monitoring.

We first extend the notion of least upper bound (lub) from one
associating a partial function to a set of partial functions to one
associating a set of partial functions to a family (or set) of sets
of partial functions:


\begin{definition}
\label{dfn:sqcup-Theta}
If $\{\Theta_i\}_{i\in I}$ is a family of sets in $\XV$, then
we let the \textbf{least upper bound} (also \textbf{lub}) of
$\{\Theta_i\}_{i\in I}$ be defined as: \\
$$
\begin{array}{l}
\sqcup\{\Theta_i \mid i\in I\} \stackrel{\rm def}{=}
\{\sqcup\{\theta_i\mid i\in I\} \mid \theta_i \in\Theta_i
\mbox{ for each } i\in I \
 \mbox{ such that }
\sqcup\{\theta_i\mid i\in I\} \mbox{ exists} \}.
\end{array}
$$
\\
As before, we use the infix notation when $I$ is
finite, e.g., we may write
$\Theta_1 \sqcup \Theta_2 \sqcup \dots \sqcup \Theta_n$
instead of
$\sqcup\{\Theta_i \mid i \in \{1,2,\dots,n\}\}$.
\end{definition}

Therefore, $\sqcup\{\Theta_i \mid i\in I\}$ is the set containing all the
lub's corresponding to sets formed by picking for each $i\in I$
precisely one element from $\Theta_i$.
Unlike for sets of partial functions, the lub's of families
of sets of partial functions always exist;
$\sqcup\{\Theta_i\mid i\in I\}$ is the empty set when no collection
of $\theta_i\in\Theta_i$ can be found (one
$\theta_i\in\Theta_i$ for each $i\in I$) such that
$\{\theta_i\mid i\in I\}$ has an upper bound.

There is an admitted slight notational ambiguity between the two least
upper bound notations introduced so far.
We prefer to purposely allow this ambiguity instead of inventing a new
notation for the lub's of families of sets, hoping that
the reader is able to quickly disambiguate the two by checking the
types of the objects involved in the lub: if partial
functions then the first lub is meant, if sets of
partial functions then the second.  Note that such notational
ambiguities are actually common practice elsewhere; e.g.,
in a monoid 
$(M,\_\!\!*\!\!\_:\!M\!\times\! M\!\rightarrow\! M,1)$ with binary
operation $*$ and unit $1$, the $*$ is commonly extended to sets of
elements $M_1,M_2$ in $M$ as expected: 
$M_1*M_2=\{m_1*m_2 \mid m_1\in M_1,m_2\in M_2\}$.

\begin{proposition}
\label{prop:sqcup-simple}
The following facts hold, where
$\Theta,\Theta_1,\Theta_2,\Theta_3 \subseteq \XV$:
\begin{enumerate}[(1)]
\item $\sqcup\emptyset = \{\bot\}$, where, in this case,
the empty set $\emptyset\subseteq{\cal P}(\XV)$ is meant;
\item $\sqcup\{\Theta\} = \Theta$; in particular
  $\sqcup\{\emptyset\} = \emptyset$ when the empty set
  $\emptyset\subseteq\XV$ is meant;
\item $\sqcup\{\{\theta\}\mid\theta\in\Theta\}=
\left\{
\begin{array}{ll}
\{\sqcup\Theta\} & \mbox{ if $\Theta$ has a lub, and} \\
\emptyset & \mbox{ if $\Theta$ does not have a lub;}
\end{array}
\right.
$
\item $\emptyset \sqcup \Theta = \emptyset$, where the empty set
  $\emptyset\subseteq\XV$ is meant;
\item $\{\bot\} \sqcup \Theta = \Theta$;
\item If $\Theta_1\subseteq\Theta_2$ then
  $\Theta_1\sqcup\Theta_3\subseteq\Theta_2\sqcup\Theta_3$;
in particular, if $\bot\in\Theta_2$ then
$\Theta_3\subseteq\Theta_2\sqcup\Theta_3$;
\item $(\Theta_1\cup\Theta_2)\sqcup\Theta_3
      =(\Theta_1\sqcup\Theta_3)\cup(\Theta_2\sqcup\Theta_3)$.
\end{enumerate}
\end{proposition}
\begin{proof}
Recall that the least upper bound $\sqcup\{\Theta_i\mid i\in I\}$ of
sets of sets of partial functions is built by collecting all the lubs
of sets $\{\theta_i\mid i\in I\}$ containing one element
$\theta_i$ from each of the sets $\Theta_i$.  When $|I|=0$,
that is when $I$ is empty, there is precisely one set
$\{\theta_i\mid i\in I\}$, the empty set of partial functions.  Then
\textit{1.} follows by \textit{1.} in Proposition \ref{prop:ub}.  When
$|I|=1$, that is when $\{\Theta_i\mid i\in I\} = \{\Theta\}$ for some 
$\Theta\subseteq\XV$ like in \textit{2.}, then the sets
$\{\theta_i\mid i\in I\}$ are precisely the singleton sets
corresponding to the elements of $\Theta$, so \textit{2.} follows by
\textit{2.} in Proposition \ref{prop:ub}.  \textit{3.} holds because
there is only one way to pick an element from each singleton set
$\{\theta\}$, namely to pick the $\theta$ itself; this also shows how
the notion of a lub of a family of sets generalizes the conventional
notion of lub.  When $|I|\geq 2$ and at 
least one of the involved sets of partial functions is empty, like in
\textit{4.}, then there is no set $\{\theta_i\mid i\in I\}$, so the
least upper bound of the set of sets is empty (regarded, again, as the
empty set of sets of partial functions).  \textit{5.} follows by
\textit{1.} in Proposition \ref{prop:basic-lub}.  The first part of
\textit{6.} is immediate and the second part follows from the first
using \textit{5.}.  Finally, \textit{7.} follows by double
implication:
$(\Theta_1\sqcup\Theta_3)\cup(\Theta_2\sqcup\Theta_3) \subseteq
(\Theta_1\cup\Theta_2)\sqcup\Theta_3$ follows by \textit{6.} because
$\Theta_1$ and $\Theta_2$ are included in $\Theta_1\cup\Theta_2$, and
$(\Theta_1\cup\Theta_2)\sqcup\Theta_3 \subseteq
(\Theta_1\sqcup\Theta_3)\cup(\Theta_2\sqcup\Theta_3)$ because for any
$\theta_1\in\Theta_1\cup\Theta_2$, say $\theta_1\in\Theta_1$, and any
$\theta_3\in\Theta_3$, if $\theta_1\sqcup\theta_3$ exists then it also
belongs to $\Theta_1\sqcup\Theta_3$.
\end{proof}

\begin{proposition}
\label{prop:sqcup-partition}
Let $\{\Theta_i\}_{i\in I}$ be a family of sets of partial maps
in $\XV$ and let ${\cal I}=\{I_j\}_{j\in J}$ be a
partition of $I$: $I=\cup\{I_j \mid j\in J\}$ and
$I_{j_1} \cap I_{j_2}=\emptyset$ for any different $j_1,j_2\in J$.  Then
$
\sqcup\{\Theta_i \mid i\in I\} =
\sqcup\{\sqcup\{\Theta_{i_j}\mid i_j\in I_j\} \mid j\in J\}.
$
\end{proposition}
\begin{proof}
For each $j\in J$, let $Q_j$ denote the set
$\sqcup\{\Theta_{i_j} \mid i_j\in I_j\}$.
Definition \ref{dfn:sqcup-Theta} then implies the
following:
$
Q_j \stackrel{\rm def}{=}
\{\sqcup\{\theta_{i_j}\mid i_j\in I_j\} \mid \theta_{i_j} \in\Theta_{i_j}
\mbox{ for each } i_j\in I_j, \mbox{ such that }
\sqcup\{\theta_{i_j}\mid i_j\in I_j\} \mbox{ exists} \}.
$
Definition \ref{dfn:sqcup-Theta} also implies the following:
$
\sqcup\{Q_j \mid j\in J\} \stackrel{\rm def}{=}
\{\sqcup\{q_{j}\mid j\in J\} \mid q_j\in Q_j
\mbox{ for each } j\in J, \mbox{ such that }
\sqcup\{q_j \mid j\in J\} \mbox{ exists} \}.
$
Putting the two equalities above together, we get that
$\sqcup\{\sqcup\{\Theta_{i_j}\mid i_j\in I_j\} \mid j\in J\}$ equals
the following:
$$
\begin{array}{l}
\{\ \ \sqcup\{\sqcup\{\theta_{i_j}\mid i_j\in I_j\}\mid j\in J\}

\ \ \mid \ \ \theta_{i_j} \in\Theta_{i_j} \mbox{ for each } j\in J\mbox{ and }\ i_j\in I_j,
  \mbox{ such that } \\
\hspace*{34ex}
\sqcup\{\theta_{i_j}\mid i_j\in I_j\} \mbox{ exists for each } j\in J
  \mbox{ and } \\
\hspace*{34ex}
\sqcup\{\sqcup\{\theta_{i_j}\mid i_j\in I_j\} \mid j\in J\} \mbox{ exists} \ \ \}.
\end{array}
$$
Since $\{I_j\}_{j\in J}$ is a partition of $I$, the indexes $i_j$
generated by ``for each $j\in J$ and $i_j\in I_j$'' cover precisely
all the indexes $i\in I$.  Moreover, picking partial functions
$\theta_{i_j}\in\Theta_{i_j}$ for each $j\in J$ and $i_j\in I_j$ is
equivalent to picking partial functions $\theta_{i}\in\Theta_{i}$ for
each $i\in I$, and, in this case,
$\{\theta_i\mid i\in I\} =
 \cup\{\{\theta_{i_j}\mid i_j \in I_j\}\mid j \in J\}$.  By
\textit{5.} in Proposition \ref{prop:ub} we then infer that
$\sqcup\{\theta_i\mid i\in I\}$ exists if and only if
$\sqcup \{\sqcup\{\theta_{i_j}\mid i_j \in I_j\}\mid j \in J\}$
exists, and if both exist then
$\sqcup\{\theta_i\mid i\in I\} = 
 \sqcup \{\sqcup\{\theta_{i_j}\mid i_j \in I_j\}\mid j \in J\}$; if
both exist then $\sqcup\{\theta_{i_j}\mid i_j\in I_j\}$ also exists
for each $j\in J$ (because
$
\{\theta_{i_j}\mid i_j\in I_j\} \subseteq
\{\theta_i\mid i\in I\} =
\cup\{\{\theta_{i_j}\mid i_j \in I_j\}\mid j \in J\}$).
Therefore, we can conclude that
$\sqcup\{\sqcup\{\Theta_{i_j}\mid i_j\in I_j\} \mid j\in J\}$ equals
$\{\sqcup\{\theta_{i}\mid i\in I\} \mid \theta_{i} \in\Theta_{i}
\mbox{ for each } i\in I, \mbox{ such that }
\sqcup\{\theta_{i}\mid i\in I\} \mbox{ exists} \}$, which is nothing
but $\sqcup\{\Theta_i \mid i\in I\}$.
\end{proof}

\begin{corollary}
\label{cor:lub-family}
The following hold:
\begin{enumerate}[(1)]
\item $\{\bot\}\sqcup\Theta=\Theta$ (already proved as \textit{5.} in
  Proposition \ref{prop:sqcup-simple});
\item $\Theta_1\sqcup\Theta_2 = \Theta_2\sqcup\Theta_1$;
\item $\Theta_1\sqcup(\Theta_2\sqcup\Theta_3) = 
       (\Theta_1\sqcup\Theta_2)\sqcup\Theta_3$;
\end{enumerate}
\end{corollary}
\begin{proof}
These follow from Proposition \ref{prop:sqcup-partition} for various
index sets $I$ and partitions of it: for \textit{1.} take $I=\{1\}$
and its partition $I=\emptyset\cup I$, take $\Theta_1=\Theta$, and
then use \textit{1.} in Proposition \ref{prop:sqcup-simple} saying
that $\sqcup\emptyset=\{\bot\}$; for \textit{2.} take partitions
$\{1\}\cup\{2\}$ and $\{2\}\cup\{1\}$ of $I=\{1,2\}$, getting
$\Theta_1\sqcup\Theta_2 = \Theta_2\sqcup\Theta_1=\sqcup\{\Theta_i\mid
i\in\{1,2\}\}$; finally, for \textit{3.} take partitions
$\{1\}\cup\{2,3\}$ and $\{1,2\}\cup\{3\}$ of $I=\{1,2,3\}$, getting
$\Theta_1\sqcup(\Theta_2\sqcup\Theta_3) =
(\Theta_1\sqcup\Theta_2)\sqcup\Theta_3=\sqcup\{\Theta_i\mid
i\in\{1,2,3\}\}$.
\end{proof}

\subsection{Least Upper Bound Closures}
\label{sec:lub-closure}

We next define lub closures of sets of partial maps, a crucial
operation for the algorithms discussed next in the paper.

\begin{definition}
\label{dfn:lub-closure}
$\Theta\subseteq\XV$ is \textbf{lub closed} if and only if $\sqcup\Theta'\in\Theta$
for any $\Theta'\subseteq\Theta$ admitting upper bounds.
\end{definition}

\begin{proposition}
\label{prop:bot-closed}
$\{\bot\}$ and $\{\bot,\theta\}$ are lub closed ($\theta\in\XV$).
\end{proposition}
\begin{proof}
It follows easily from Definition \ref{dfn:lub-closure}, using the
facts that $\sqcup\emptyset = \bot$ (\textit{1.} in Proposition~\ref{prop:ub}),
$\sqcup\{\theta\}=\theta$ (\textit{2.} in Proposition~\ref{prop:ub}),
and $\sqcup\{\bot,\theta\}=\theta$ (\textit{3.} in
Proposition~\ref{prop:ub} for $\Theta=\{\theta\}$).
\end{proof}

\begin{proposition}
\label{prop:lub-closure}
If $\Theta\subseteq\XV$ and $\{\Theta_i\}_{i\in I}$ is a
family of sets of partial functions in $\XV$, then:
\begin{enumerate}[(1)]
\item If $\Theta$ is lub closed then $\bot\in\Theta$; in particular,
$\emptyset$ is not lub closed;
\item If $\Theta$ has upper bounds and is lub closed then it has a maximum;
\item $\Theta$ is lub closed iff $\sqcup\{\Theta\mid i\in I\}=\Theta$
for any $I$;
\item If $\Theta$ is lub closed and $\Theta_i\subseteq\Theta$ for each
  $i\in I$ then $\sqcup\{\Theta_i\mid i\in I\}\subseteq\Theta$;
\item If $\Theta_i$ is lub closed for each $i\in I$ then
$\sqcup\{\Theta_i\mid i\in I\}$ is lub closed and
$$\cup\{\Theta_i\mid i\in I\}\subseteq\sqcup\{\Theta_i\mid i\in I\};$$
\item If $I$ finite and $\Theta_i$ finite for all $i\in I$, then
$\sqcup\{\Theta_i\mid i\in I\}$ finite;
\item If $\Theta_i$ lub closed for all $i\in I$ then
$\cap\{\Theta_i \mid i \in I\}$ is lub closed;
\item $\cap\{\Theta' \mid \Theta'\subseteq\XV
\mbox{ with } \Theta\subseteq\Theta'
\mbox{ and } \Theta' \mbox{ is lub closed}\}$ is the smallest lub
closed set including $\Theta$.
\end{enumerate}
\end{proposition}
\begin{proof}
\textit{1.} follows taking $\Theta'=\emptyset$ in Definition
\ref{dfn:lub-closure} and using $\sqcup\emptyset=\bot$
(\textit{1.} in Proposition \ref{prop:ub}).

\textit{2.} follows taking $\Theta'=\Theta$ in Definition
\ref{dfn:lub-closure}: $\sqcup\Theta\in\Theta$, so $\max\,\Theta$
exists (and equals $\sqcup\Theta$).

\textit{3.} Definition \ref{dfn:sqcup-Theta} implies that
$\sqcup\{\Theta \mid i\in I\}$ equals
$\{\sqcup\{\theta_i\mid i\in I\} \mid \theta_i \in\Theta
\mbox{ for each } i\in I, \mbox{ such that }
\sqcup\{\theta_i\mid i\in I\} \mbox{ exists} \}$, which is nothing but
$\{\sqcup\Theta'\mid\Theta'\subseteq\Theta
\mbox{ such that }\sqcup\Theta'\mbox{ exists}\}$; the later can be now
shown equal to $\Theta$ by double inclusion:
$\{\sqcup\Theta'\mid\Theta'\subseteq\Theta
\mbox{ such that }\sqcup\Theta'\mbox{ exists}\}
\subseteq\Theta$ because $\Theta$ is lub closed, and
$\Theta\subseteq\{\sqcup\Theta'\mid\Theta'\subseteq\Theta
\mbox{ such that }\sqcup\Theta'\mbox{ exists}\}$ because one can pick
$\Theta'=\{\theta\}$ for each $\theta\in\Theta$ and use the fact that
$\sqcup\{\theta\}=\theta$ (\textit{2.} in Proposition \ref{prop:ub}).

\textit{4.} Let $\theta$ be an arbitrary partial function in
$\sqcup\{\Theta_i\mid i\in I\}$, that is,
$\theta=\sqcup\{\theta_i\mid i\in I\}$ for some $\theta_i \in\Theta_i$,
one for each $i\in I$, such that $\{\theta_i\mid i\in I\}$ has upper
bounds.  Since $\Theta$ is lub closed and $\Theta_i\subseteq\Theta$
for each $i\in I$, it follows that $\theta\in\Theta$.  Therefore,
$\sqcup\{\Theta_i\mid i\in I\}\subseteq\Theta$.

\textit{5.} Let $\Theta'$ be a set of partial functions included in
$\sqcup\{\Theta_i\mid i\in I\}$ which admits an upper bound; moreover,
for each $\theta'\in\Theta'$, let us fix a set
$\{\theta_i^{\theta'}\mid i\in I\}$ such that
$\theta_i^{\theta'}\in\Theta_i$ for each $i\in I$ and
$\theta'=\sqcup\{\theta_i^{\theta'}\mid i\in I\}$ (such sets exist
because $\theta'\in\Theta'\subseteq\sqcup\{\Theta_i\mid i\in I\}$).
Let $\Theta^{\theta'}$ be the set $\{\theta_i^{\theta'}\mid i\in I\}$
for each $\theta'\in\Theta'$, let $\Theta'_i$ be the set
$\{\theta_i^{\theta'}\mid \theta'\in\Theta'\}$ for each $i\in I$, and
let $\Theta$ be the set
$\{\theta_i^{\theta'}\mid\theta'\in\Theta',i\in
I^{\theta'}\}$.  It is easy to see that $\Theta=
\cup\{\Theta^{\theta'}\mid\theta'\in\Theta'\} =
\cup\{\Theta'_i\mid i\in I\}$ and that $\Theta'_i\subseteq\Theta_i$
for each $i\in I$.  Since $\sqcup\Theta'$ exists (because
$\Theta'$ has upper bounds) and
$\sqcup\Theta'= \sqcup\{\theta'\mid\theta'\in\Theta'\} = 
\sqcup\{\sqcup\Theta^{\theta'}\mid\theta'\in\Theta'\}$,
by \textit{5.} in Proposition \ref{prop:ub} it follows that
$\sqcup\Theta$ exists and $\sqcup\Theta'=\sqcup\Theta$.
Since $\Theta=\cup\{\Theta'_i\mid i\in I\}$ and $\sqcup\Theta$ exists,
by \textit{5.} in Proposition \ref{prop:ub} again we get that
$\sqcup\{\sqcup\Theta'_i\mid i\in I\}$ exists and is equal to
$\sqcup\Theta$, which is equal to $\sqcup\Theta'$.  Since
$\Theta'_i\subseteq\Theta_i$ and $\Theta_i$ is lub closed, we get that
$\sqcup\Theta'_i\in\Theta_i$.  That means that
$\sqcup\{\sqcup\Theta'_i\mid i\in I\}\in\sqcup\{\Theta_i\mid i\in
I\}$, that is, that $\sqcup\Theta'\in\sqcup\{\Theta_i\mid i\in I\}$.
Since $\Theta'\subseteq\sqcup\{\Theta_i\mid i\in I\}$ was chosen
arbitrarily, we conclude that $\sqcup\{\Theta_i\mid i\in I\}$ is lub
closed.  To show that
$\cup\{\Theta_i\mid i\in I\}\subseteq\sqcup\{\Theta_i\mid i\in I\}$,
let us pick an $i\in I$ and let us partition $I$ as
$\{i\}\cup (I\backslash\{i\})$.  By Proposition
\ref{prop:sqcup-partition}, 
$\sqcup\{\Theta_i\mid i\in I\} = \Theta_i \sqcup (\sqcup\{\Theta_j\mid
j\in I\backslash\{i\}\})$.  The proof above also implies that
$\sqcup\{\Theta_j\mid j\in I\backslash\{i\}\}$ is lub closed, so by
\textit{1.} we get that $\bot\in\sqcup\{\Theta_j\mid
j\in I\backslash\{i\}\}$.  Finally, \textit{6.} in Proposition
\ref{prop:sqcup-simple} implies
$\Theta_i \sqsubseteq \Theta_i \sqcup (\sqcup\{\Theta_j\mid
j\in I\backslash\{i\}\})$, so
$\Theta_i\subseteq\sqcup\{\Theta_i\mid i\in I\}$ for each $i\in I$,
that is, $\cup\{\Theta_i\mid i\in I\}\subseteq\sqcup\{\Theta_i\mid i\in I\}$.

\textit{6.} Recall from Definition \ref{dfn:sqcup-Theta} that
$\sqcup\{\Theta_i\mid i\in I\}$ contains the existing least upper
bounds of sets of partial functions containing precisely one partial
function in each $\Theta_i$.  If $I$ and each of the $\Theta_i$ for
each $i\in I$ is finite, then
$|\sqcup\{\Theta_i\mid i\in I\}|\leq\prod_{i\in I}|\Theta_i|$, because
there at most $\prod_{i\in I}|\Theta_i|$ combinations of partial
functions, one in each $\Theta_i$, that admit an upper bound.
Therefore, $\sqcup\{\Theta_i\mid i\in I\}$ is also finite.

\textit{7.} Let $\Theta'\subseteq\cap\{\Theta_i\mid i\in I\}$ be a set
of partial functions admitting an upper bound.  Then
$\Theta'\subseteq\Theta_i$ for each $i\in I$ and, since each
$\Theta_i$ is lub closed, $\sqcup\Theta'\in\Theta_i$ for each
$i\in I$.  Therefore, $\sqcup\Theta'\in\cap\{\Theta_i\mid i\in I\}$.

\textit{8.} Anticipating the definition of and notation for lub
closures (Definition \ref{dfn:lub-closure}), we let
$\overline{\Theta}$ denote the set
$\cap\{\Theta' \mid \Theta'\subseteq\XV
\mbox{ with } \Theta\subseteq\Theta'
\mbox{ and } \Theta' \mbox{ is lub closed}\}$.  It is clear that
$\Theta\subseteq\overline{\Theta}$ and, by \textit{7.}, that
$\overline{\Theta}$ is lub closed.  It is also the \textit{smallest}
lub closed set including $\Theta$, because all such sets $\Theta'$ are
among those whose intersection defines $\overline{\Theta}$.
\end{proof}

\begin{definition}
Given $\theta'\in\XV$ and $\Theta\subseteq\XV$,
let
$$(\theta']_\Theta \stackrel{\rm def}{=}
\{\theta\mid \theta\in\Theta \mbox{ and } \theta \sqsubseteq \theta'\}$$
be the set of partial functions in $\Theta$ that are less informative
than $\theta'$.
\end{definition}

\begin{proposition}
\label{prop:max}
If $\theta,\theta',\theta'',\theta_1,\theta_2\in\XV$ and
if $\Theta\subseteq\XV$ is lub closed, then:
\begin{enumerate}[(1)]
\item $(\theta']_\Theta$ is lub closed and $\max\,(\theta']_\Theta$ exists;
\item If $\theta'\in\{\theta\}\sqcup\Theta$ then 
$\{\theta'' \mid \theta''\in\Theta \mbox{ and } \theta'=\theta\sqcup\theta''\}$
has maximum and that equals $\max\,(\theta']_\Theta$;
\item If $\theta_1,\theta_2\in\{\theta\}\sqcup\Theta$ such that
$\theta_1=\max\,(\theta_2]_\Theta$, then $\theta_1=\theta_2$.
\end{enumerate}
\end{proposition}
\begin{proof}
\textit{1.} First, note that $\theta'$ is an upper bound for
$(\theta']_\Theta$ as well as for any subset $\Theta'$ of it, so any
$\Theta'\subseteq(\theta']_\Theta$ has upper bounds, so by \textit{2.}
in Proposition \ref{prop:basic-ub}, $\sqcup\Theta'$ exists for any
$\Theta'\subseteq(\theta']_\Theta$.  Moreover, if
$\Theta'\subseteq(\theta']_\Theta$ then
$\sqcup\Theta'\sqsubseteq\theta'$, and since $\Theta$ is lub closed it 
follows that $\sqcup\Theta'\in\Theta$, so
$\sqcup\Theta'\in(\theta']_\Theta$.  Therefore, $(\theta']_\Theta$ is
lub closed. \textit{2.} in Proposition \ref{prop:lub-closure} now
implies that $(\theta']_\Theta$ has maximum; to be concrete,
$\max\,(\theta']_\Theta$ is nothing but 
$\sqcup\,(\theta']_\Theta$, which belongs to $(\theta']_\Theta$
(because one can pick $\Theta'=(\theta']_\Theta$ above).

\textit{2.} 
Let $Q$ be the set of partial functions
$\{\theta'' \mid \theta''\in\Theta \mbox{ and }\theta'=\theta\sqcup\theta''\}$.
Note that $Q$ is non-empty (because $\theta'\in\{\theta\}\sqcup\Theta$,
so there is some $\theta''\in\Theta$ such as
$\theta'=\theta\sqcup\theta''$) and has upper-bounds (because
$\theta'$ is an upper bound for it), but that it is not necessarily
lub closed (because, unless $\theta'=\theta$, $Q$ does not
contain $\bot$, contradicting \textit{1.} in Proposition
\ref{prop:lub-closure}).  Hence $Q$ has a lub (by \textit{2.}
in Proposition \ref{prop:basic-ub}), say $q$, and
$q=\sqcup\,Q\sqsubseteq\theta'$; since $\theta\sqsubseteq\theta'$,
it follows that $\theta \sqcup q \sqsubseteq \theta'$.  On the other hand
$\theta'\sqsubseteq\theta\sqcup q$ by \textit{4.} in Proposition
\ref{prop:ub}, because there is some $\theta''\in
Q$ such that $\theta'=\theta\sqcup\theta''$ and $\theta''\sqsubseteq
q$.  Therefore, $\theta'=\theta\sqcup q$.  Since $\Theta$ is lub closed, it
follows that $q\in\Theta$.  Therefore, $q\in Q$, so $q$ is the maximum
element of $Q$.
Let us next show that $q = \max\,(\theta']_\Theta$.  The relation
$q \sqsubseteq \max\,(\theta']_\Theta$ is immediate because
$q\in(\theta']_\Theta$ (we proved above that $q\in\Theta$ and
$q\sqsubseteq\theta'$).  For $\max\,(\theta']_\Theta\sqsubseteq q$ it
suffices to show that $\max\,(\theta']_\Theta \in Q$, that is, that
$\theta\sqcup\max\,(\theta']_\Theta=\theta'$:
$\theta\sqcup\max\,(\theta']_\Theta\sqsubseteq\theta'$ follows because
$\theta\sqsubseteq\theta'$ and 
$\max\,(\theta']_\Theta\sqsubseteq\theta'$, while
$\theta'\sqsubseteq\theta\sqcup\max\,(\theta']_\Theta$ follows because
there is some $\theta''\in\Theta$ such that
$\theta'=\theta\sqcup\theta''$ and, since
$\theta''\sqsubseteq\max\,(\theta']_\Theta$,
$\theta\sqcup\theta''\sqsubseteq\theta\sqcup\max\,(\theta']_\Theta$
(by \textit{4.} in Proposition \ref{prop:ub}).

\textit{3.} admits a direct proof simpler than that of
\textit{2.}; however, since \textit{2.} is needed anyway, we prefer to
use \textit{2.}  
Note that $\theta\sqsubseteq\theta_1\sqsubseteq\theta_2$.
By \textit{2.},
$\theta_1=\max\,\{\theta'' \mid \theta''\in\Theta \mbox{ and }\theta_2=\theta\sqcup\theta''\}$, which implies
$\theta_2=\theta\sqcup\theta_1=\theta_1$.
\end{proof}

\begin{definition}
\label{dfn:compatibility-closure}
Given $\Theta\subseteq\XV$, we let $\overline{\Theta}$, 
the \textbf{least upper bound (lub) closure} of $\Theta$,
be defined as follows:
$$\overline{\Theta}\stackrel{\rm def}{=}
\cap\{\Theta' \mid \Theta'\subseteq\XV
\mbox{ with } \Theta\subseteq\Theta'
\mbox{ and } \Theta' \mbox{ is lub closed}\}.$$
\end{definition}

\begin{proposition}
\label{prop:closure-simple}
The next hold ($\Theta\subseteq\XV,\theta\in\XV$):
\begin{enumerate}[(1)]
\item $\overline{\Theta}$ is the smallest lub closed set including
$\Theta$;
\item $\overline{\emptyset}=\overline{\{\bot\}}=\{\bot\}$;
\item $\overline{\{\theta\}}=\{\bot,\theta\}$.
\end{enumerate}
\end{proposition}
\begin{proof}
\textit{1.} follows by \textit{7.} in Proposition \ref{prop:lub-closure}.
For \textit{2.} and \textit{3.}, first note that $\{\bot\}$ and
$\{\bot,\theta\}$ are lub closed by Proposition \ref{prop:bot-closed};
second, note that they are indeed the smallest lub closed sets
including $\bot$ and resp. $\theta$, as any lub closed
set must include $\bot$ (\textit{1.} in Proposition \ref{prop:lub-closure}).
\end{proof}

\begin{proposition}
\label{prop:closure-operator}
The lub closure map
$\overline{\cdot}:2^{\XV} \rightarrow
2^{\XV}$ is a closure operator, that is, for any
$\Theta,\Theta_1,\Theta_2\subseteq\XV$,
\begin{enumerate}[(1)]
\item (extensivity) $\Theta \subseteq \overline{\Theta}$;
\item (monotonicity) If $\Theta_1 \subseteq \Theta_2$ then
$\overline{\Theta_1} \subseteq \overline{\Theta_2}$;
\item (idempotency) $\overline{\overline{\Theta}}=\overline{\Theta}$.
\end{enumerate}
\end{proposition}
\begin{proof}
Extensivity and idempotency follow immediately from the definitions of
$\overline{\Theta}$ and $\overline{\overline{\Theta}}$ (which are
lub closed by \textit{1.} in Proposition \ref{prop:closure-simple}).
For monotonicity, one should note that $\overline{\Theta_2}$ satisfies
the properties of $\overline{\Theta_1}$ (i.e., $\Theta_1 \subseteq
\overline{\Theta_2}$ and $\Theta_2$ is lub closed); since
$\overline{\Theta_1}$ is the smallest with those properties, it
follows that $\overline{\Theta_1}\subseteq\overline{\Theta_2}$.
\end{proof}

\begin{proposition}
\label{prop:overline-union}
$\overline{\cup\{\Theta_i \mid i \in I\}} =
       \sqcup\{\overline{\Theta_i}\mid i\in I\}$
for any family $\{\Theta_i\}_{i\in I}$ of partial functions in
$\XV$.
\end{proposition}
\begin{proof}
Since $\overline{\Theta_i}$ is lub closed for any $i\in I$,
\textit{5.} in Proposition \ref{prop:lub-closure} implies that
$\sqcup\{\overline{\Theta_i}\mid i\in I\}$ is lub closed and
$\cup\{\overline{\Theta_i} \mid i \in I\} \subseteq
\sqcup\{\overline{\Theta_i}\mid i\in I\}$.  Since \textit{1.}
in Proposition \ref{prop:closure-operator} implies
$\Theta_i\subseteq\overline{\Theta_i}$ for each $i\in I$ and since
$\overline{\cup\{\Theta_i \mid i \in I\}}$ is the smallest lub closed
set including $\cup\{\Theta_i \mid i \in I\}$ (\textit{1.} in
Proposition \ref{prop:closure-simple}), the inclusion
$\overline{\cup\{\Theta_i\mid i \in I\}} \subseteq 
\sqcup\{\overline{\Theta_i}\mid i\in I\}$ holds.  Conversely,
\textit{2.} in Proposition \ref{prop:closure-operator} implies that
$\overline{\Theta_i}\subseteq\overline{\cup\{\Theta_i\mid i \in I\}}$
for any $i\in I$, so  
$\sqcup\{\overline{\Theta_i}\mid i\in I\} \subseteq
\overline{\cup\{\Theta_i\mid i \in I\}}$ holds by \textit{4.} in
Proposition \ref{prop:lub-closure}.
\end{proof}

\begin{corollary}
\label{cor:theta-Theta}
For any $\theta\in\XV$ and
$\Theta\subseteq\XV$, equality
$\overline{\{\theta\}\cup\Theta}=\{\bot,\theta\}\sqcup\overline{\Theta}$ holds.
\end{corollary}
\begin{proof}
$\overline{\{\theta\}\cup\Theta}=\overline{\{\theta\}}\sqcup\overline{\Theta}$
by Proposition \ref{prop:overline-union}, and
$\overline{\{\theta\}}=\{\bot,\theta\}$ by \textit{3.} in Proposition~\ref{prop:closure-simple}.
\end{proof}

\begin{corollary}
\label{cor:finite}
If $\Theta\subseteq\XV$ is finite then 
$\overline{\Theta}$ is also finite.
\end{corollary}
\begin{proof}
Suppose that $\Theta=\{\theta_1,\theta_2,\dots,\theta_n\}$ for some
$n\geq 0$.  Iteratively applying Corollary
\ref{cor:theta-Theta},
$\overline\Theta=\{\bot,\theta_1\}\sqcup\{\bot,\theta_2\}\sqcup\cdots\{\bot,\theta_n\}$;
in obtaining that, we used \textit{2.} in Proposition
\ref{prop:closure-simple} and \textit{1.} in Corollary \ref{cor:lub-family}.
The result follows now by \textit{6.} in Proposition \ref{prop:lub-closure}.
\end{proof}

\section{Parametric Traces and Properties}
\label{sec:parametric-trace-property}

Here we introduce the notions of event, trace and property, first
non-parametric and then parametric.  Traces are sequences of
events.  Parametric events can carry data-values as instances of
parameters.  Parametric traces are traces over parametric events.
Properties are trace classifiers, that is, mappings partitioning the
space of traces into categories (violating traces, validating traces,
don't know traces, etc.).  Parametric properties are parametric trace
classifiers and provide, for each parameter instance, the category to
which the trace slice corresponding to that parameter instance
belongs.  Trace slicing is defined as a reduct operation that forgets
all the events that are unrelated to the given parameter instance.

\subsection{The Non-Parametric Case}
\label{sec:non-param-case}
We next introduce non-parametric events, traces and properties,
which will serve as an infrastructure for their parametric variants in
Section~\ref{sec:param-case}.

\begin{definition}
\label{dfn:nonparametric-events}
Let $\cal E$ be a set of (non-parametric) events, called
\textbf{base events} or simply \textbf{events}.
An \textbf{$\cal E$-trace}, or simply a (non-parametric)
\textbf{trace} when $\cal E$ is understood or not
important, is any finite sequence of events in $\cal E$, that is, an
element in ${\cal E}^*$.  If event $e\in{\cal E}$ appears in trace
$w\in{\cal E}^*$ then we write $e \in w$.
\end{definition}
Our parametric trace slicing and monitoring techniques in Sections
\ref{sec:trace-slicing} and \ref{sec:trace-monitoring} can be
easily adapted to also work with infinite traces.  Since infinite
versus finite traces is not an important aspect of the work reported
here, we keep the presentation free of unnecessary technical
complications and consider only finite traces.

\vspace*{1ex}

\begin{example}
\textit{(part 1 of simple running example)} Like in Section~\ref{sec:acquire-release},
consider a certain resource (e.g., a synchronization object) that can be acquired and
released during the lifetime of a given procedure (between its begin and end).  Then 
the set of events ${\cal E}$ is
$\{\textsf{\footnotesize acquire},\textsf{\footnotesize release},\textsf{\footnotesize begin},\textsf{\footnotesize end}\}$
and execution traces corresponding to this resource are sequences
$\textsf{\footnotesize begin}\ \textsf{\footnotesize acquire}\ \textsf{\footnotesize acquire}\
\textsf{\footnotesize release}\ \textsf{\footnotesize end}\ \textsf{\footnotesize begin}\ \textsf{\footnotesize end}$,
$\textsf{\footnotesize begin}\ \textsf{\footnotesize acquire}\ \textsf{\footnotesize acquire}$,
$\textsf{\footnotesize begin}\ \textsf{\footnotesize acquire}\ \textsf{\footnotesize release}\
\textsf{\footnotesize acquire}\ \textsf{\footnotesize end}$, etc.
For the time being, there are no ``good'' or ``bad'' execution traces.
\end{example}

There is a plethora of formalisms to specify trace requirements.
Many of these result in specifying at least two types of traces:
those {\em validating} the specification (i.e, correct traces), and
those {\em violating} the specification (i.e., incorrect traces).

\vspace*{1ex}

\begin{example}
\textit{(part 2)}
Consider a regular expression specification,
$$(\textsf{\footnotesize begin} (\epsilon\ |\ 
     (\textsf{\footnotesize acquire} 
       (\textsf{\footnotesize acquire} \ |\  \textsf{\footnotesize release})^* 
     \textsf{\footnotesize release})) \textsf{\footnotesize end})^*$$
stating that the procedure can
(non-recursively) take place multiple times and, if the resource
is acquired during the procedure then it is released by the end of the
procedure.  Assume that the resource can be acquired and released
multiple times, with the effect of acquiring and respectively
releasing it precisely once.  The validating (or matching)
traces for this property are those satisfying the pattern, e.g.,
$\textsf{\footnotesize begin}\ \textsf{\footnotesize acquire}\ \textsf{\footnotesize acquire}\
\textsf{\footnotesize release}\ \textsf{\footnotesize end}\ \textsf{\footnotesize begin}\ \textsf{\footnotesize end}$.
At first sight, one may say that all the other traces are violating (or failing)
traces, because they are not in the language of the regular expression.
However, there are two interesting types of such ``failing'' traces: ones
which may still lead to a matching trace provided the right
events will be received in the future, e.g.,
$\textsf{\footnotesize begin}\ \textsf{\footnotesize acquire}\ \textsf{\footnotesize acquire}$, and
ones which have no chance of becoming a matching trace, e.g.,
$\textsf{\footnotesize begin}\ \textsf{\footnotesize acquire}\ \textsf{\footnotesize release}\
\textsf{\footnotesize acquire}\ \textsf{\footnotesize end}$.
\end{example}

In general, traces are not enforced to correspond to terminated
programs (this is particularly useful in monitoring); if one wants to
enforce traces to correspond to terminated programs, then one can have
the system generate a special \textsf{\footnotesize end-of-trace} event and have
the property specification require that event at the end of each trace.

Therefore, a trace property may partition the space of traces into
more than two categories.  For some specification formalisms, for
example ones based on fuzzy logics or multiple truth values, the set
of traces may be split into more than three categories, even into a
continuous space of categories.  Section~\ref{sec:success-ratio}
showed an example of a property where the space of trace categories
was the set of rational numbers between 0 and 1.

\begin{definition}
\label{dfn:nonparametric-property}
An \textbf{$\cal E$-property} $P$, or simply a
(base or non-parametric) \textbf{property}, is a function
$P:{\cal E}^*\rightarrow {\cal C}$ partitioning the set of traces into
(verdict) categories $\cal C$.
\end{definition}

It is common, though not enforced, that the set of property verdict categories
$\cal C$ in Definition~\ref{dfn:nonparametric-property}  includes
\textsf{\footnotesize validating} (or similar), \textsf{\footnotesize violating} (or similar), and
\textsf{\footnotesize don't know}  (or \textsf{\footnotesize ?}) categories.  In general,
$\cal C$ can be any set, finite or infinite.

We believe that the definition of non-parametric trace property above
is general enough that it can easily accommodate any particular
specification formalism, such as ones based on linear temporal logics,
regular expressions, context-free grammars, etc.  All one needs to do
in order to instantiate the general results in this paper for a
particular specification formalism is to decide upon the desired
categories in which traces are intended to be classified, and then
define the property associated to a specification accordingly.

For example, if the specification formalism of choice is that of regular
expressions over $\cal E$ and one is interested in classifying traces
in three categories as in our example above, then one can pick
$\cal C$ to be the set 
$\{\textsf{\footnotesize match},\textsf{\footnotesize fail},\textsf{\footnotesize don't know}\}$
and, for a given regular expression $E$, define its associated
property $P_E : {\cal E}^* \rightarrow {\cal C}$ as follows:
$P_E(w)=\textsf{\footnotesize match}$ iff $w$ is in the language of $E$,
$P_E(w)=\textsf{\footnotesize fail}$ iff there is no $w'\in{\cal E}^*$ such
that $w\,w'$ is in the language of $E$, and
$P_E(w)=\textsf{\footnotesize don't know}$ otherwise; this is the
monitoring semantics of regular expressions in the JavaMOP and RV systems
\cite{meredith-jin-griffith-chen-rosu-2010-jsttt,chen-rosu-2007-oopsla,meredith-rosu-2010-rv}.

Other semantic choices are possible even for the simple case of
regular expressions; for example, one may choose $\cal C$ to
be the set $\{\textsf{\footnotesize match},\textsf{\footnotesize don't care}\}$ and
define $P_E(w)=\textsf{\footnotesize match}$ iff $w$ is in the language of $E$,
and $P_E(w)=\textsf{\footnotesize don't care}$ otherwise;
this is the semantics of regular expressions in
Tracematches \cite{tracematches-oopsla}, where, depending upon how one
writes the regular expression, matching can mean either a
violation or a validation of the desired property.

Similarly, one can have various verdict categories for linear temporal
logic (LTL).  For example, one can report \textsf{\footnotesize violation} when a bad prefix is
encountered, can report \textsf{\footnotesize validation} when a good prefix is encountered,
and can report \textsf{\footnotesize inconclusive} when neither of the above; this is the
current monitoring semantics of monitoring LTL in JavaMOP and RV
\cite{meredith-jin-griffith-chen-rosu-2010-jsttt,meredith-rosu-2010-rv}.  Semantical and algorithmic aspects
regarding LTL monitoring can be found, e.g., in
\cite{kupferman-vardi-2001,rosu-havelund-2005-jase,stolz-2006-rv,bauer-leucker-schallhart-2010-tosem}.

In some
applications, one may not be interested in certain categories of
traces, such as in those classified as \textsf{\footnotesize don't know},
\textsf{\footnotesize don't care}, or \textsf{\footnotesize inconclusive}; if that is the case, then those
applications can
simply ignore these, like Tracematches, JavaMOP and RV do.  It may be
worth making it explicit that in this paper we do not attempt to
propose or promote any particular formalism for specifying properties
about execution traces.  Instead, our approach is to define properties
as generally as possible to capture the various specification formalisms
that we are aware of as special cases, and then to develop our
subsequent techniques to work with such general properties.
We believe that our definition of property in 
Definition~\ref{dfn:nonparametric-property} is general enough to
allow us to claim that our results are specification-formalism-independent.

An additional benefit of defining properties so generally, as mappings
from traces to categories, is that parametric properties, in spite of
their much more general flavor, are also properties (but, obviously,
over different traces and over different categories).

\subsection{The Parametric Case}
\label{sec:param-case}

Events often carry concrete data instantiating parameters.

\vspace*{1ex}

\begin{example}
\textit{(part 3)}
In our running example, events \textsf{\footnotesize acquire} and
\textsf{\footnotesize release} are parametric in the resource being acquired
or released; if $r$ is the name of the generic resource parameter
and $r_1$ and $r_2$ are two concrete resources, then parametric
acquire/release events have the form
$\textsf{\footnotesize acquire}\langle r\mapsto r_1\rangle$,
$\textsf{\footnotesize release}\langle r\mapsto r_2\rangle$, etc.
Not all events need carry instances for all parameters; e.g., the
begin/end parametric events have the form
$\textsf{\footnotesize begin}\langle \bot\rangle$ and
$\textsf{\footnotesize end}\langle \bot\rangle$,
where $\bot$, the partial map undefined everywhere, instantiates no
parameter.
\end{example}

Recall from Definition \ref{dfn:partial-functions} that
$\totalf{A}{B}$ and $\partialf{A}{B}$ denote the sets of total
and, respectively, partial functions from $A$ to $B$.

\begin{definition}
\label{dfn:parametric-trace}
\textbf{(Parametric events and traces).}
Let $X$ be a set of \textbf{parameters} and let $V$ be a set of
corresponding \textbf{parameter values}.  If $\cal E$ is a set of base
events like in Definition~\ref{dfn:nonparametric-events}, then let
$\paramevents$ denote the set of corresponding
\textbf{parametric events} $e\langle\theta\rangle$, where
$e$ is a base event in $\cal E$ and $\theta$ is a partial function in
$\XV$.  A \textbf{parametric trace} is a trace with events
in $\paramevents$, that is, a word in $\paramevents^*$.
\end{definition}

Therefore, a parametric event is an event carrying values for zero,
one, several or even all the parameters, and a parametric trace is a
finite sequence of parametric events.  In practice, the number of
values carried by an event is finite; however, we do not need to
enforce this restriction in our theoretical developments.  Also, in
practice the parameters may be typed, in which case the set of their
corresponding values is given by their type.  To simplify writing, we
occasionally assume the set of parameter values $V$ implicit.

\vspace*{1ex}

\begin{example}
\textit{(part 4)}
A parametric trace for our running example can be the following:
$
\textsf{\footnotesize begin}\langle\!\bot\!\rangle\,
\textsf{\footnotesize acquire}\langle\theta_1\rangle\,
\textsf{\footnotesize acquire}\langle\theta_2\rangle\,
\textsf{\footnotesize acquire}\langle\theta_1\rangle\,
\textsf{\footnotesize release}\langle\theta_1\rangle\,
\textsf{\footnotesize end}\langle\!\bot\!\rangle\, 
\textsf{\footnotesize begin}\langle\!\bot\!\rangle\, 
\textsf{\footnotesize acquire}\langle\theta_2\rangle\, 
\textsf{\footnotesize release}\langle\theta_2\rangle\, 
\textsf{\footnotesize end}\langle\!\bot\!\rangle,
$
where $\theta_1$ maps $r$ to $r_1$ and $\theta_2$ maps $r$ to $r_2$.
To simplify writing, we take the freedom to only list the parameter
instance values when writing parameter instances, that is,
$\langle r_1 \rangle$ instead of 
$\langle r \mapsto r_1 \rangle$, or $\tau\!\!\!\!\upharpoonright_{r_2}$
instead of $\tau\!\!\!\!\upharpoonright_{r \mapsto r_2}$, etc.  
This notation is formally introduced in the next section, as Notation
\ref{notation:compact}.
With
this notation, the above trace becomes
$
\textsf{\footnotesize begin}\langle\rangle\ 
\textsf{\footnotesize acquire}\langle r_1\rangle\
\textsf{\footnotesize acquire}\langle r_2\rangle\ 
\textsf{\footnotesize acquire}\langle r_1\rangle\ 
\textsf{\footnotesize release}\langle r_1\rangle\
\textsf{\footnotesize end}\langle\rangle\ 
\textsf{\footnotesize begin}\langle\rangle\ 
\textsf{\footnotesize acquire}\langle r_2\rangle\
\textsf{\footnotesize release}\langle r_2\rangle\ 
\textsf{\footnotesize end}\langle\rangle.
$
This trace involves two resources, $r_1$ and $r_2$, and it really
consists of \textit{two trace slices}, one for each resource, merged
together.  The begin and end events belong to both trace
slices.  The slice corresponding to $\theta_1$ is
$\textsf{\footnotesize begin}\ 
\textsf{\footnotesize acquire}\ 
\textsf{\footnotesize acquire}\ 
\textsf{\footnotesize release}\ 
\textsf{\footnotesize end}\
\textsf{\footnotesize begin}\ 
\textsf{\footnotesize end}
$, while the one for $\theta_2$ is
$\textsf{\footnotesize begin}\,
\textsf{\footnotesize acquire}\, 
\textsf{\footnotesize end}\, 
\textsf{\footnotesize begin}\, 
\textsf{\footnotesize acquire}\, 
\textsf{\footnotesize release}\, 
\textsf{\footnotesize end}
$.
\end{example}

Recall from Definition \ref{dfn:partial-functions} that two partial
maps of same source and target are compatible when they do not
disagree on any of the elements on which they are both defined, and
that one is less informative than another, written
$\theta\sqsubseteq\theta'$, when they are compatible and the domain of
the former in included in the domain of the latter.  With the notation
in the example above we have, for our running example, that  $\langle
\rangle$ is compatible with $\langle r_1 \rangle$ and with $\langle
r_2 \rangle$, but $\langle r_1 \rangle$ and $\langle r_2 \rangle$ are
not compatible; moreover, $\langle \rangle \sqsubseteq \langle r_1
\rangle$ and $\langle \rangle \sqsubseteq \langle r_2 \rangle$.

\begin{definition}
\label{dfn:trace-slicing}
\textbf{(Trace slicing)}
Given parametric trace $\tau\in\paramevents^*$
and partial function $\theta$ in $\XV$, we let the
\textbf{$\theta$-trace slice}
$\tau\!\!\upharpoonright_\theta\,\in{\cal E}^*$ be the
non-parametric trace in ${\cal E}^*$ defined as:
\begin{iteMize}{$\bullet$}
\item $\epsilon\!\!\upharpoonright_\theta = \epsilon$, where $\epsilon$
  is the empty trace/word, and
\item $(\tau\,e\langle\theta'\rangle)\!\!\upharpoonright_\theta =
\left\{
\begin{array}{ll}
(\tau\!\!\upharpoonright_\theta)\,e
  & \mbox{when } \theta' \sqsubseteq \theta \\
\tau\!\!\upharpoonright_\theta
  & \mbox{when } \theta' \not\sqsubseteq \theta
\end{array}
\right.
$
\end{iteMize}
\end{definition}

\noindent
Therefore, the trace slice $\tau\!\!\upharpoonright_\theta$ first
filters out all the parametric events that are not relevant for the
instance $\theta$, i.e., which contain instances of parameters that
$\theta$ does not care about, and then, for the remaining events
relevant to $\theta$, it forgets the parameters so that the trace
can be checked against base, non-parametric properties.
It is crucial to discard parameter instances that are not relevant to
$\theta$ during the slicing, including those more informative than
$\theta$, in order to achieve a correct slice for $\theta$: in our
running example, the trace slice for $\langle \rangle$ should contain
only \textsf{\footnotesize begin} and \textsf{\footnotesize end} events and no \textsf{\footnotesize acquire} or
\textsf{\footnotesize release}.  Otherwise, the \textsf{\footnotesize acquire} and \textsf{\footnotesize release} of
different resources will interfere with each other in the trace slice
for $\langle \rangle$.

One should not confuse extracting/abstracting traces from executions
with slicing traces.  The former determines the events to include in
the trace, as well as parameter instances carried by events, while the
latter dispatches each event in the given trace to corresponding trace
slices according to the event's parameter instance.  Different
abstractions may result in different parametric traces from the same
execution and thus may lead to different trace slices for the same
parameter instance $\theta$.  For the (map, collection, iterator)
example in Section \ref{sec:intro}, $X = \{m, c, i\}$ and an
execution may generate the following parametric trace:
$\textsf{\footnotesize createColl}\langle m_1, c_1 \rangle$  
$\textsf{\footnotesize createIter}\langle c_1, i_1 \rangle$
$\textsf{\footnotesize next}\langle i_1 \rangle$ $\textsf{\footnotesize updateMap}\langle m_1 \rangle$.
The trace slice for $\langle m_1 \rangle$ is \textsf{\footnotesize updateMap} for this
parametric trace.  Now suppose that we are only interested in
operations on maps.  Then $X = \{ m \}$ and the trace abstracted from
the execution generating the above trace is
$\textsf{\footnotesize createColl}\langle m_1 \rangle$
$\textsf{\footnotesize updateMap}\langle m_1 \rangle$, in which events and parameter
bindings irrelevant to $m$ are removed.
Then the trace slice for $\langle m_1 \rangle$ is
\textsf{\footnotesize createColl}\ \textsf{\footnotesize updateMap}.
In this paper we focus only on trace slicing; more discussion about
trace abstraction can be found in \cite{chen-rosu-2008-tr-a}.

Specifying properties over parametric traces is rather challenging,
because one may want to specify a property for one generic parameter
instance and then say ``and so on for all the other instances''.  In
other words, one may want to specify a sort of a universally quantified
property over base events, but, unfortunately, the underlying specification
formalism may not allow such quantification over data; for
example, none of the conventional formalisms to specify properties on
linear traces listed above (i.e, linear temporal logics, regular
expressions, context-free grammars) or mentioned in the rest of the
paper has data quantification by default.  We say ``by default''
because, in some cases, there are attempts to add data quantification;
for example, \cite{manna-pnueli-1992} investigates the implications and
the necessary restrictions resulting from adding quantifiers to LTL, and
\cite{stolz-2006-rv} investigates a finite-trace variant of parametric LTL
together with a translation to parametric alternating automata.

\begin{definition}
\label{dfn:parametric-property}
Let $X$ be a set of parameters together with their corresponding
parameter values $V$, like in Definition \ref{dfn:parametric-trace},
and let $P:{\cal E}^*\rightarrow{\cal C}$ be a non-parametric property
like in Definition \ref{dfn:nonparametric-property}.  Then we define
the \textbf{parametric property} $\parametric{X}{P}$ as the
property (over traces $\paramevents^*$ and
categories $[\XV \rightarrow {\cal C}]$)
$$
\parametric{X}{P} \ \ : \ \ \paramevents^*
  \rightarrow [\XV \rightarrow {\cal C}]
$$
defined as
$$(\parametric{X}{P})(\tau)(\theta) =
P(\tau\!\!\upharpoonright_\theta)$$
for any $\tau\in\paramevents^*$
and any $\theta\in\XV$.
If $X=\{x_1,...,x_n\}$ we may write $\parametric{x_1,...,x_n}{P}$
instead of $(\parametric{\{x_1,...,x_n\}}{P}$.  Also, if $P_\varphi$
is defined using a pattern or formula $\varphi$ in some particular
trace specification formalism, we take the liberty to write
$\parametric{X}{\varphi}$ instead of $\parametric{X}{P_\varphi}$.
\end{definition}

Parametric properties $\parametric{X}{P}$ over base properties
$P:{\cal E}^*\rightarrow{\cal C}$ are therefore properties taking
traces in $\paramevents^*$ to categories 
$\totalf{\XV}{\cal C}$, i.e., to function domains from parameter
instances to base property categories.  $\parametric{X}{P}$
is defined as if many instances of $P$ are observed at the
same time on the parametric trace, one property instance for each
parameter instance, each property instance concerned with its events
only, dropping the unrelated ones.

\vspace*{1ex}

\begin{example}
\textit{(part 5)}
Let $P$ be the non-parametric property specified by the regular
expression in the second part of our running example above (using the
mapping of regular expressions to properties discussed in the second
part of our running example and after Definition
\ref{dfn:nonparametric-property} -- i.e., the JavaMOP/RV semantic
approach to parametric monitoring \cite{chen-rosu-2007-oopsla}).
Since we want $P$ to hold for any
resource instance, we define the following parametric property:
$$\parametric{r}{(\textsf{\footnotesize begin}\ (\epsilon \ |\  
     (\textsf{\footnotesize acquire}\ 
       (\textsf{\footnotesize acquire} \ |\ \textsf{\footnotesize release})^*\ 
     \textsf{\footnotesize release}))\ \textsf{\footnotesize end})^*}.$$
If $\tau$ is the parametric trace and $\theta_1$ and $\theta_2$ are the
parameter instances in the fourth part of our running example, then the
semantics of the parametric property above on trace $\tau$ is
\textsf{\footnotesize validating} for parameter instance $\theta_1$ and 
\textsf{\footnotesize violating} for parameter instance $\theta_2$.
\end{example}

\section{Slicing With Less}
\label{sec:parametric-trace}

Consider a parametric trace $\tau$ in $\paramevents^*$ and a parameter
instance $\theta$.  Since there is no apriori correlation between the
parameters being instantiated by $\theta$ and those by the various
parametric events in $\tau$, it may very well be the case that
$\theta$ contains parameter instances that never appear in $\tau$.
In this section we show that slicing $\tau$ by $\theta$ is the same as
slicing it by a ``smaller'' parameter instance than $\theta$, namely
one containing only those parameters instantiated by $\theta$ that
also appear as instances of some parameters in some events in $\tau$.
Formally, this smaller parameter instance is the largest partial map
smaller than $\theta$ in the lub closure of all the parameter
instances of events in $\tau$; this partial function is proved to
indeed exist.
We first formalize a notation used informally so far in this paper:

\begin{notation}
\label{notation:compact}
When the domain of $\theta$ is finite, which is always the case in our
examples in this paper and probably will also be the case in most
practical uses of our trace slicing algorithm, and when the
corresponding parameter names are clear from context, we take the
liberty to write partial functions compactly by simply listing their
parameter values;  for example, we write a partial function $\theta$
with $\theta(a)=a_2$, $\theta(b)=b_1$ and $\theta(c)=c_1$ as the
sequence ``$a_2b_1c_1$''.  The function $\bot$ then corresponds to the
empty sequence.
\end{notation}

\begin{example}
Here is a parametric trace with events parametric in $\{a,b,c\}$,
where the parameters take values in the set $\{a_1,a_2,b_1,c_1\}$:
$$
\tau =
e_1\langle a_1 \rangle\ 
e_2\langle a_2 \rangle\ 
e_3\langle b_1 \rangle\
e_4\langle a_2 b_1 \rangle\ 
e_5\langle a_1 \rangle\ 
e_6\langle \rangle\ 
e_7\langle b_1 \rangle\ 
e_8\langle c_1 \rangle\ 
e_9\langle a_2 c_1 \rangle\ 
e_{10}\langle a_1 b_1 c_1 \rangle\ 
e_{11}\langle \rangle.
$$
It may be the case that some of the base events appearing in a trace
are the same; for example, $e_1$ may be equal to $e_2$ and to $e_5$.
It is actually frequently the case in practice (at least in PQL
\cite{pql-oopsla}, Tracematches \cite{tracematches-oopsla}, JavaMOP
\cite{chen-rosu-2007-oopsla} and RV \cite{meredith-rosu-2010-rv}) that parametric events are specified
apriori with a given (sub)set of parameters, so that each event in
$\cal E$ is always instantiated with partial functions over the same
domain, i.e., if $e\langle\theta\rangle$ and
$e\langle\theta'\rangle$ appear in a parametric trace, then
$\Dom(\theta)=\Dom(\theta')$.  While this restriction is reasonable and
sometimes useful, our trace slicing and monitoring algorithms in this
paper do not need it.
\end{example}

Recall from Definition \ref{dfn:trace-slicing} that the trace slice
$\tau\!\!\upharpoonright_\theta$ keeps from $\tau$ only those events
that are relevant for $\theta$ and drops their parameters.

\vspace*{1ex}

\begin{example}
Consider again the sample parametric trace above with events
parametric in $\{a,b,c\}$:
$
\tau =
e_1\langle a_1 \rangle\ 
e_2\langle a_2 \rangle\ 
e_3\langle b_1 \rangle\ 
e_4\langle a_2 b_1 \rangle\ 
e_5\langle a_1 \rangle\ 
e_6\langle \rangle\ 
e_7\langle b_1 \rangle\ 
e_8\langle c_1 \rangle\ 
e_9\langle a_2 c_1 \rangle\ 
e_{10}\langle a_1 b_1 c_1 \rangle\ 
e_{11}\langle \rangle.
$
Several slices of $\tau$ are listed below:
$$
\begin{array}{lcl}
\tau\!\!\upharpoonright_{a_1} & = & e_1 e_5 e_6 e_{11} \\
\tau\!\!\upharpoonright_{a_2} & = & e_2 e_6 e_{11} \\
\tau\!\!\upharpoonright_{a_1b_1} & = & e_1 e_3 e_5 e_6 e_7 e_{11} \\
\tau\!\!\upharpoonright_{a_2b_1} & = & e_2 e_3 e_4 e_6 e_7 e_{11} \\
\tau\!\!\upharpoonright_{\ } & = & e_6 e_{11} \\
\tau\!\!\upharpoonright_{a_1b_1c_1} & = & e_1 e_3 e_5 e_6 e_7 e_8 e_{10} e_{11} \\
\tau\!\!\upharpoonright_{a_2b_1c_1} & = & e_2 e_3 e_4 e_6 e_7 e_8 e_9 e_{11} \\
\tau\!\!\upharpoonright_{a_1b_2c_1} & = & e_1 e_5 e_6 e_8 e_{11} \\
\tau\!\!\upharpoonright_{b_2c_2} & = & e_6 e_{11} \\
\end{array}
$$
In order for the partial functions above to make sense, we assumed
that the set $V$ in which parameters $X=\{a,b,c\}$ take values
includes $\{a_1,a_2,b_1,c_1\}$.
\end{example}

\begin{definition}
\label{dfn:Theta-tau}
Given parametric trace $\tau\in\paramevents^*$,
we let $\Theta_\tau$ denote the lub closure of all the
parameter instances appearing in events in $\tau$, that is,
$\Theta_\tau=\overline{\{\theta\mid \theta\in\XV,\ e\langle\theta\rangle\in\tau\}}$.
\end{definition}

\begin{proposition}
\label{prop:Theta-tau-closed}
For any parametric trace $\tau\in\paramevents^*$, the set
$\Theta_\tau$ is a finite lub closed set.
\end{proposition}
\begin{proof}
$\Theta_\tau$ is already defined as a lub closed set; since $\tau$ is
finite, Corollary \ref{cor:finite} implies that $\Theta_\tau$ is
finite.
\end{proof}

\begin{proposition}
\label{prop:Theta-tau}
Given $\tau\,e\langle\theta\rangle\in\paramevents^*$,
the following equality holds:
$\Theta_{\tau\,e\langle\theta\rangle}=\{\bot,\theta\}\sqcup\Theta_\tau$.
\end{proposition}
\begin{proof}
It follows by the following sequence of equalities:
$$
\begin{array}{lll}
\Theta_{\tau\,e\langle\theta\rangle}
 & = & 
  \overline{\{\theta'\mid \theta'\in\XV,\
    e'\langle\theta'\rangle\in\tau\,e\langle\theta\rangle\}} \\
 & = &
  \overline{\{\theta\} \cup \{\theta'\mid \theta'\in\XV,\
    e'\langle\theta'\rangle\in\tau\}} \\
 & = &
  \overline{\{\theta\} \cup \Theta_\tau} \\
 & = &
\{\bot,\theta\}\sqcup\overline{\Theta_\tau} \\
 & = &
\{\bot,\theta\}\sqcup\Theta_\tau.
\end{array}
$$
The first equality follows by Definition \ref{dfn:Theta-tau},
the second by separating the case
$e'\langle\theta'\rangle=e\langle\theta\rangle$,
the third again by Definition \ref{dfn:Theta-tau},
the fourth by Corollary \ref{cor:theta-Theta}, and the fifth
by Proposition \ref{prop:Theta-tau-closed}.
Therefore, $\Theta_{\tau\,e\langle\theta\rangle}$ is the
smallest lub closed set that contains $\theta$ and includes $\Theta_\tau$. 
\end{proof}

\begin{proposition}
\label{prop:tau-max}
Given $\tau\in\paramevents^*$
and $\theta\in\XV$, 
the equality
$\tau\!\!\upharpoonright_\theta = 
\tau\!\!\upharpoonright_{\max\, (\theta]_{\Theta_\tau}}$
holds.
\end{proposition}
\begin{proof}
We prove the following more general result:
\begin{quote}
``Let $\Theta\subseteq\XV$ be lub closed and let
$\theta\in\XV$;

then $\tau\!\!\upharpoonright_\theta = 
\tau\!\!\upharpoonright_{\max\, (\theta]_{\Theta}}$
for any $\tau\in\paramevents^*$ with
$\Theta_\tau\subseteq\Theta$.''
\end{quote}
First note that the statement above is well-formed because
$\max\,(\theta]_{\Theta}$ exists whenever $\Theta$ is 
lub closed (\textit{1.} in Proposition \ref{prop:max}), and
that it is indeed more general than the stated result: for the
given $\tau\in\paramevents^*$ and
$\theta\in\XV$, we pick $\Theta$ to be $\Theta_\tau$.
We prove the general result by induction on the {\em length} of $\tau$:

- If $|\tau|=0$ then $\tau=\epsilon$ and $\epsilon\!\!\!\upharpoonright_\theta = 
\epsilon\!\!\!\upharpoonright_{\max\, (\theta]_{\Theta}}=\epsilon$.

- Now suppose that $\tau\!\!\upharpoonright_\theta = 
\tau\!\!\upharpoonright_{\max\, (\theta]_{\Theta}}$
for any $\tau\in\paramevents^*$ with
$\Theta_\tau\subseteq\Theta$ and $|\tau|=n\geq 0$, and let us show
that $\tau'\!\!\upharpoonright_\theta = 
\tau'\!\!\upharpoonright_{\max\, (\theta]_{\Theta}}$
for any $\tau'\in\paramevents^*$ with
$\Theta_{\tau'}\subseteq\Theta$ and $|\tau'|=n+1$.  Pick such a
$\tau'$ and let $\tau'= \tau\,e\langle\theta'\rangle$
for a $\tau\in\paramevents^*$ with
$|\tau|=n$ and an
$e\langle\theta'\rangle\in\paramevents$.
Since $\Theta_{\tau'}\subseteq\Theta$, by \textit{6.} in Proposition
\ref{prop:sqcup-simple} and by Proposition \ref{prop:Theta-tau} it
follows that
$\Theta_\tau\subseteq\{\bot,\theta'\}\sqcup\Theta_\tau\subseteq\Theta$,
so the induction hypothesis implies $\tau\!\!\upharpoonright_\theta = 
\tau\!\!\upharpoonright_{\max\, (\theta]_{\Theta}}$.
The rest follows noticing that $\theta'\sqsubseteq\theta$ iff
$\theta'\sqsubseteq \max\, (\theta]_{\Theta}$, which is a consequence
of the definition of $\max\, (\theta]_{\Theta}$ because
$\theta'\in\{\bot,\theta'\}\subseteq 
 \{\bot,\theta'\}\sqcup\Theta_\tau\subseteq\Theta$
(again by \textit{6.} in Proposition
\ref{prop:sqcup-simple} and by Proposition \ref{prop:Theta-tau}).

Alternatively, one could have also done the proof above by induction
on $\tau$, not on its length, but the proof would be more involved, 
because one would need to prove that the domain over which the
property is universally quantified, namely
``any $\tau\in\paramevents^*$ with
$\Theta_\tau\subseteq\Theta$'' is inductively generated.  We therefore
preferred to choose a more elementary induction schema.
\end{proof}

\section{Algorithm for Online Parametric Trace Slicing}
\label{sec:trace-slicing}

Definition \ref{dfn:trace-slicing} illustrates a way to slice a
parametric trace for {\em given} parameter bindings.  However, it is
not suitable for online trace slicing, where the trace is observed
incrementally and no future knowledge is available, because we cannot
know all possible parameter instances $\theta$ apriori.
We next define an algorithm $\A$ that takes a parametric trace
$\tau\in\paramevents^*$ incrementally (i.e.,
event by event), and builds a partial function
$\T \in \partialf{\XV}{{\cal E}^*}$ of finite domain
that serves as a quick lookup table for all slices of $\tau$.  More
precisely, Theorem \ref{thm:trace-slicing} shows that, for any
$\theta\in\XV$, the trace slice
$\tau\!\!\upharpoonright_\theta$ is
$\T(\max\,(\theta]_{\Theta})$
after $\A$ processes $\tau$, where $\Theta=\Theta_\tau$
is the domain of $\T$, a finite lub closed set of
partial functions also calculated by $\A$ incrementally (see Definition~\ref{dfn:Theta-tau}
for $\Theta_\tau$). Therefore, assuming that $\A$ is run on trace $\tau$, all one
has to do in order to calculate a slice $\tau\!\!\upharpoonright_\theta$ for a given
$\theta\in\XV$ is to calculate $\max\,(\theta]_{\Theta}$
followed by a lookup into $\T$.  This way the trace $\tau$,
which can be very long, is processed/traversed only once, as it is being
generated, and appropriate data-structures are maintained by our
algorithm that allow for retrieval of slices for any parameter
instance $\theta$, without having to traverse the trace $\tau$ again,
as an algorithm blindly following the definition of trace slicing 
(Definition \ref{dfn:trace-slicing}) would do.

\newcommand{\mydots}{\parbox{0ex}{\vspace*{-1.5ex}\vdots}}
\newcommand{\hsp}{\mydots\hspace*{1ex}}
\newcommand{\myforall}{\textsf{\footnotesize foreach }}
\newcommand{\mydo}{\textsf{\footnotesize  do }}
\newcommand{\myendfor}{\textsf{\footnotesize endfor }}
\newcommand{\myif}{\textsf{\footnotesize if }}
\newcommand{\mythen}{\textsf{\footnotesize  then }}

\begin{figure}
\begin{center}
$
\begin{array}{l}
\textsf{\footnotesize Algorithm} \ \A \\
\textsf{\footnotesize Input:} \ \mbox{parametric trace }\tau\in
  \paramevents^*\\
\textsf{\footnotesize Output:} \ \mbox{map }\T\in
  \partialf{\XV}{{\cal E}^*}
  \mbox{ and set }\Theta\subseteq\XV\\
\\
1\ \ \T \leftarrow \bot;\ \T(\bot)\leftarrow\epsilon; \ \Theta \leftarrow \{\bot\} \\
2\ \ \myforall \textit{parametric event } e\langle{\theta}\rangle
\textit{ in order (first to last) in }\tau\ \mydo \\
3\ \ \hsp \myforall \ \theta' \in \{\theta\} \sqcup \Theta \ \mydo \\
4\ \ \hsp\hsp \T(\theta') \leftarrow\T(\max\,(\theta']_\Theta)\,e \\
5\ \ \hsp \myendfor \\
6\ \ \hsp \Theta \leftarrow \{\bot,\theta\}\sqcup\Theta \\
7\ \ \myendfor
\end{array}
$
\end{center}
\caption{\label{fig:A}Parametric trace slicing algorithm $\A$.}
\end{figure}

Figure \ref{fig:A} shows our trace slicing algorithm $\A$.  In
spite of $\A$'s small size, its proof of correctness is surprisingly
intricate, making use of almost all the mathematical machinery
developed so far in the paper.
The algorithm $\A$ on input $\tau$, written more succinctly $\A(\tau)$,
traverses $\tau$ from its first event to its last event and, for each
encountered event $e\langle\theta\rangle$, updates both its
data-structures, $\T$ and $\Theta$.  After processing each event, the
relationship between $\T$ and $\Theta$ is that the latter is
the domain of the former.
Line \textsf{\footnotesize 1} initializes the data-structures: $\T$ is
undefined everywhere (i.e., $\bot$) except for the
undefined-everywhere function $\bot$, where
$\T(\bot)=\epsilon$; as expected, $\Theta$ is then initialized
to the set $\{\bot\}$.
The code (lines \textsf{\footnotesize 3} to \textsf{\footnotesize 6}) inside the outer loop
(lines \textsf{\footnotesize 2} to \textsf{\footnotesize 7}) can be triggered when a new event is
received, as in most online runtime verification systems.  When a new
event is received, say $e\langle\theta\rangle$, the mapping
$\T$ is updated as follows: for each
$\theta'\in\XV$ that can be obtained by combining $\theta$
with the compatible partial functions in the domain of the current
$\T$, update $\T(\theta')$ by adding the
non-parametric event $e$ to the end of the slice corresponding to the
largest (i.e., most ``knowledgeable'') entry in the current table
$\T$ that is less informative or as informative as $\theta'$;
the $\Theta$ data-structure is then extended in line \textsf{\footnotesize 6}
(see Proposition~\ref{prop:Theta-tau} for why this way).

\vspace*{1ex}

\begin{example}
\begin{table*}
\begin{center}
\noindent
{\footnotesize
$
\begin{array}{@{}l@{}}
\begin{array}{|l|l|l|l|l|l|l|}
\hline
e_1\langle a_1 \rangle &
e_2\langle a_2 \rangle &
e_3\langle b_1 \rangle &
e_4\langle a_2 b_1 \rangle &
e_5\langle a_1 \rangle &
e_6\langle \rangle &
e_7\langle b_1 \rangle
\\
\hline
\!\!\!\!
\begin{array}[t]{l}
\langle\rangle\!\!:\!\epsilon \\
\langle a_1 \rangle \!\!:\! e_1
\end{array}
&
\!\!\!\!
\begin{array}[t]{l}
\langle\rangle\!\!:\!\epsilon \\
\langle a_1 \rangle \!\!:\! e_1 \\
\langle a_2 \rangle \!\!:\! e_2
\end{array}
&
\!\!\!\!
\begin{array}[t]{l}
\langle\rangle\!\!:\!\epsilon \\
\langle a_1 \rangle \!\!:\! e_1 \\
\langle a_2 \rangle \!\!:\! e_2 \\
\langle b_1 \rangle \!\!:\! e_3 \\
\langle a_1 b_1 \rangle \!\!:\! e_1 e_3 \\
\langle a_2 b_1 \rangle \!\!:\! e_2 e_3
\end{array}
&
\!\!\!\!
\begin{array}[t]{l}
\langle\rangle\!\!:\!\epsilon \\
\langle a_1 \rangle \!\!:\! e_1 \\
\langle a_2 \rangle \!\!:\! e_2 \\
\langle b_1 \rangle \!\!:\! e_3 \\
\langle a_1 b_1 \rangle \!\!:\! e_1 e_3 \\
\langle a_2 b_1 \rangle \!\!:\! e_2 e_3 e_4
\end{array}
&
\!\!\!\!
\begin{array}[t]{l}
\langle\rangle\!\!:\!\epsilon \\
\langle a_1 \rangle \!\!:\! e_1 e_5\\
\langle a_2 \rangle \!\!:\! e_2 \\
\langle b_1 \rangle \!\!:\! e_3 \\
\langle a_1 b_1 \rangle \!\!:\! e_1 e_3 e_5\\
\langle a_2 b_1 \rangle \!\!:\! e_2 e_3 e_4
\end{array}
&
\!\!\!\!
\begin{array}[t]{l}
\langle \rangle \!\!:\! e_6\\
\langle a_1 \rangle \!\!:\! e_1 e_5 e_6\\
\langle a_2 \rangle \!\!:\! e_2 e_6 \\
\langle b_1 \rangle \!\!:\! e_3 e_6 \\
\langle a_1 b_1 \rangle \!\!:\! e_1 e_3 e_5 e_6 \\
\langle a_2 b_1 \rangle \!\!:\! e_2 e_3 e_4 e_6
\end{array}
&
\!\!\!\!
\begin{array}[t]{l}
\langle \rangle \!\!:\! e_6\\
\langle a_1 \rangle \!\!:\! e_1 e_5 e_6\\
\langle a_2 \rangle \!\!:\! e_2 e_6 \\
\langle b_1 \rangle \!\!:\! e_3 e_6 \\
\langle a_1 b_1 \rangle \!\!:\! e_1 e_3 e_5 e_6 \\
\langle a_2 b_1 \rangle \!\!:\! e_2 e_3 e_4 e_6
\end{array}
\\
\hline
\end{array}
\\
\\
\begin{array}{|l|l|l|l|}
\hline
e_8\langle c_1 \rangle &
e_9\langle a_2 c_1 \rangle &
e_{10}\langle a_1 b_1 c_1 \rangle &
e_{11}\langle \rangle
\\
\hline
\!\!\!
\begin{array}[t]{l}
\langle \rangle \!\!:\! e_6 \\
\langle a_1 \rangle \!\!:\! e_1 e_5 e_6\\
\langle a_2 \rangle \!\!:\! e_2 e_6 \\
\langle b_1 \rangle \!\!:\! e_3 e_6 e_7 \\
\langle a_1 b_1 \rangle \!\!:\! e_1 e_3 e_5 e_6 e_7 \\
\langle a_2 b_1 \rangle \!\!:\! e_2 e_3 e_4 e_6 e_7 \\
\langle c_1 \rangle \!\!:\! e_6 e_8\\
\langle a_1 c_1 \rangle \!\!:\! e_1 e_5 e_6 e_8 \\
\langle a_2 c_1 \rangle \!\!:\! e_2 e_6 e_8 \\
\langle b_1 c_1 \rangle \!\!:\! e_3 e_6 e_7 e_8 \\
\langle a_{1} b_{1} c_{1} \rangle \!\!:\! e_1 e_3 e_5 e_6 e_7 e_8 \\
\langle a_{2} b_{1} c_{1} \rangle \!\!:\! e_2 e_3 e_4 e_6 e_7 e_8
\end{array}
\!\!\!
&
\!\!\!\!
\begin{array}[t]{l}
\langle \rangle \!\!:\! e_6 \\
\langle a_1 \rangle \!\!:\! e_1 e_5 e_6\\
\langle a_2 \rangle \!\!:\! e_2 e_6 \\
\langle b_1 \rangle \!\!:\! e_3 e_6 e_7 \\
\langle a_1 b_1 \rangle \!\!:\! e_1 e_3 e_5 e_6 e_7 \\
\langle a_2 b_1 \rangle \!\!:\! e_2 e_3 e_4 e_6 e_7 \\
\langle c_1 \rangle \!\!:\! e_6 e_8\\
\langle a_1 c_1 \rangle \!\!:\! e_1 e_5 e_6 e_8 \\
\langle a_2 c_1 \rangle \!\!:\! e_2 e_6 e_8 e_9 \\
\langle b_1 c_1 \rangle \!\!:\! e_3 e_6 e_7 e_8 \\
\langle a_{1} b_{1} c_{1} \rangle \!\!:\! e_1 e_3 e_5 e_6 e_7 e_8 \\
\langle a_{2} b_{1} c_{1} \rangle \!\!:\! e_2 e_3 e_4 e_6 e_7 e_8 e_9
\end{array}
\!\!\!
&
\!\!\!\!
\begin{array}[t]{l}
\langle \rangle \!\!:\! e_6 \\
\langle a_1 \rangle \!\!:\! e_1 e_5 e_6\\
\langle a_2 \rangle \!\!:\! e_2 e_6 \\
\langle b_1 \rangle \!\!:\! e_3 e_6 e_7 \\
\langle a_1 b_1 \rangle \!\!:\! e_1 e_3 e_5 e_6 e_7 \\
\langle a_2 b_1 \rangle \!\!:\! e_2 e_3 e_4 e_6 e_7 \\
\langle c_1 \rangle \!\!:\! e_6 e_8 \\
\langle a_1 c_1 \rangle \!\!:\! e_1 e_5 e_6 e_8 \\
\langle a_2 c_1 \rangle \!\!:\! e_2 e_6 e_8 e_9 \\
\langle b_1 c_1 \rangle \!\!:\! e_3 e_6 e_7 e_8 \\
\langle a_{1} b_{1} c_{1} \rangle \!\!:\! e_1 e_3 e_5 e_6 e_7 e_8 e_{10} \\
\langle a_{2} b_{1} c_{1} \rangle \!\!:\! e_2 e_3 e_4 e_6 e_7 e_8 e_9
\end{array}
\!\!\!
&
\!\!\!\!
\begin{array}[t]{l}
\langle \rangle \!\!:\! e_6 e_{11} \\
\langle a_1 \rangle \!\!:\! e_1 e_5 e_6 e_{11} \\
\langle a_2 \rangle \!\!:\! e_2 e_6 e_{11} \\
\langle b_1 \rangle \!\!:\! e_3 e_6 e_7 e_{11} \\
\langle a_1 b_1 \rangle \!\!:\! e_1 e_3 e_5 e_6 e_7 e_{11} \\
\langle a_2 b_1 \rangle \!\!:\! e_2 e_3 e_4 e_6 e_7 e_{11} \\
\langle c_1 \rangle \!\!:\! e_6 e_8 e_{11} \\
\langle a_1 c_1 \rangle \!\!:\! e_1 e_5 e_6 e_8 e_{11} \\
\langle a_2 c_1 \rangle \!\!:\! e_2 e_6 e_8 e_9 e_{11} \\
\langle b_1 c_1 \rangle \!\!:\! e_3 e_6 e_7 e_8 e_{11} \\
\langle a_{1} b_{1} c_{1} \rangle \!\!:\! e_1 e_3 e_5 e_6 e_7 e_8 e_{10}
e_{\!11} \\
\langle a_{2} b_{1} c_{1} \rangle \!\!:\! e_2 e_3 e_4 e_6 e_7 e_8 e_9
e_{\!11}
\end{array}
\!\!\!\!
\\
\hline
\end{array}
\end{array}
$
}
\caption{\label{table:run-A}A run of the trace slicing algorithm $\A$ (top-left table
  first, followed by bottom-left table, followed by the right table).}
\end{center}
\end{table*}
Consider again the sample trace in Section
\ref{sec:parametric-trace} with events parametric in $\{a,b,c\}$,
namely
$
\tau =
e_1\langle a_1 \rangle\ 
e_2\langle a_2 \rangle\ 
e_3\langle b_1 \rangle\ 
\
e_4\langle a_2 b_1 \rangle\ 
e_5\langle a_1 \rangle\ 
e_6\langle \rangle\ 
e_7\langle b_1 \rangle\ 
e_8\langle c_1 \rangle\ 
e_9\langle a_2 c_1 \rangle\ 
e_{10}\langle a_1 b_1 c_1 \rangle\ 
e_{11}\langle \rangle
$.
Table \ref{table:run-A} shows how $\A$ works on $\tau$.  An entry of
the form $\langle\theta\rangle\!:\!w$ in a table cell corresponding to
a current parametric event $e\langle\theta\rangle$ means that
$\T(\theta)=w$ after processing all the parametric events up to and
including the current one; $\T$ is undefined on any other partial
function.  Obviously, the $\Theta$ corresponding to a cell is the
union of all the $\theta$'s that appear in pairs
$\langle\theta\rangle\!:\!w$ in that cell.
Note that, as each parametric event $e\langle\theta\rangle$ is
processed, the non-parametric event $e$ is added at most once to each
slice, and that $\Theta$ stays lub closed.
\end{example}

$\A$ computes trace slices for all combinations of parameter instances
observed in parametric trace events.  Its complexity is therefore
$O(n \times m)$ where $n$ is the length of the trace and $m$ is the
number of all possible parameter combinations.  However, $\A$ is not
intended to be implemented directly; it is only used as a correctness
backbone for other trace analysis algorithms, such as the monitoring
algorithms discussed below.  An alternative and apparently more
efficient solution is to only record trace slices for parameter
instances that actually appear in the trace (instead of for all
combinations of them), and then construct the slice for a given
parameter instance by combining such trace slices for compatible
parameter instances.  However, the complexity of constructing all
possible trace slices at the end using such an algorithm is also
$O(n \times m)$, so it would not bring any benefit overall compared to
$\A$.  In addition, $\A$ is more suitable as a backbone for developing
online monitoring algorithms such as those in Section \ref{sec:monitors},
because each event is sent to its slices (that will be consumed by
corresponding monitors) and never touched again.

$\A$ compactly and uniformly captures several special cases and
subcases that are worth discussing.  The discussion below can be
formalized as an inductive (on the length of  $\tau$) proof of
correctness for $\A$, but we prefer to keep this discussion informal
and give a rigorous proof shortly after.  The role of this discussion
is twofold: (1) to better explain the algorithm $\A$, providing the
reader with additional intuition for its difficulty and compactness, and
(2) to give a proof sketch for the correctness of $\A$.

Let us first note that a partial function added to $\Theta$ will never
be removed from $\Theta$; that's because
$\Theta\subseteq\{\bot,\theta\}\sqcup\Theta$.  The same holds true for
the domain of $\T$, because line \textsf{\footnotesize 4} can only add new
elements to $\Dom(\T)$; in fact, the domain of $\T$ is
extended with precisely the set $\{\theta\}\sqcup\Theta$ after each
event parametric in $\theta$ is processed by $\A$.  Moreover,
since $\Dom(\T)=\Theta=\Theta_\epsilon=\{\bot\}$ initially and
since  \textit{5.} and \textit{7.} in Proposition \ref{prop:sqcup-simple}
imply $\Theta\cup(\{\theta\}\sqcup\Theta) =
\{\bot,\theta\}\sqcup\Theta$ while Proposition \ref{prop:Theta-tau}
states that
$\Theta_{\tau\,e\langle\theta\rangle}=\{\bot,\theta\}\sqcup\Theta_\tau$,
we can inductively show that $\Dom(\T)=\Theta=\Theta_\tau$
each time after $\A$ is executed on a parametric trace $\tau$.

Each $\theta'$ considered by the loop at lines \textsf{\footnotesize 3-5} has the
property that $\theta\sqsubseteq\theta'$, and at (precisely) one
iteration of the loop $\theta'$ is $\theta$;
indeed, $\theta\in\{\theta\}\sqcup\Theta$ because $\bot\in\Theta$.
Thanks to Proposition~\ref{prop:tau-max},
Theorem~\ref{thm:trace-slicing} holds essentially iff
$\T(\theta')=\tau\!\!\upharpoonright_{\theta'}$
after $\T(\theta')$ is updated in line~\textsf{\footnotesize 4}.  A tricky
observation which is crucial for this is that \textit{3.}
in Proposition \ref{prop:max} implies that the updates of
$\T(\theta')$ do not interfere with each other for different
$\theta'\in\{\theta\}\sqcup\Theta$; otherwise the non-parametric event
$e$ may wrongly be added multiple times to some trace slices $\T(\theta')$.

Let us next informally argue, inductively, that it is indeed the case
that $\T(\theta')=\tau\!\!\upharpoonright_{\theta'}$
after $\T(\theta')$ is updated in line \textsf{\footnotesize 4} (it
vacuously holds on the empty trace).  Since
$\max\,(\theta']_\Theta\in\Theta$, the inductive hypothesis tells us
that $\T(\max\,(\theta']_\Theta)
= \tau\!\!\upharpoonright_{\max\,(\theta']_\Theta}$; these are further
equal to $\tau\!\!\upharpoonright_{\theta'}$ by Proposition
\ref{prop:tau-max}.  Since $\theta\sqsubseteq\theta'$, the definition
of trace slicing implies that 
$(\tau\,e\langle\theta\rangle)\!\!\upharpoonright_{\theta'} =
\tau\!\!\upharpoonright_{\theta'}\,e$.  Therefore,
$\T(\theta')$ is indeed
$(\tau\,e\langle\theta\rangle)\!\!\upharpoonright_{\theta'}$ after
line \textsf{\footnotesize 4} of $\A$ is executed while processing the
event $e\langle\theta\rangle$ that follows trace $\tau$.  This
concludes our informal proof sketch; let us next give a rigorous proof
of correctness for our trace slicing algorithm $\A$.

\begin{definition}
\label{dfn:A}
Let $\A(\tau).\T$ and 
$\A(\tau).\Theta$ be the two data-structures ($\T$ and
$\Theta$) maintained by the algorithm $\A$ in Figure~\ref{fig:A} after it processes $\tau$.
\end{definition}

\begin{theorem}
\label{thm:trace-slicing}
With the notation in Definition~\ref{dfn:A}, the following hold for any $\tau\in\paramevents^*$:
\begin{enumerate}[(1)]
\item $\Dom(\A(\tau).\T)=\A(\tau).\Theta=\Theta_\tau$;
\item $\A(\tau).\T(\theta)=\tau\!\!\!\upharpoonright_{\theta}$
for any $\theta\in\A(\tau).\Theta$;
\item $\tau\!\!\!\upharpoonright_{\theta}=
\A(\tau).\T(\max\,(\theta]_{\A(\tau).\Theta})$
for any $\theta\in\XV$.
\end{enumerate}
\end{theorem}
\begin{proof}
Since $\A$ processes the events in the input trace in order,
when given the input $\tau\,e\langle\theta\rangle$, the $\Theta$ and
$\T$ structures after $\A$ processes $\tau$ but
before it processes $e\langle\theta\rangle$ (i.e., right before the
last iteration of the loop at lines \textsf{\footnotesize 2-7}) are precisely
$\A(\tau).\Theta$ and $\A(\tau).\T$,
respectively.  Further, the loop at lines \textsf{\footnotesize 3-5}
updates $\T$ on all
$\theta'\in\{\theta\}\sqcup\Theta$; in case $\T$ was not
defined on such a $\theta'$, then it will be defined after
$e\langle\theta\rangle$ is processed.  The definitional domain of
$\T$ is thus continuously growing or potentially remains
stationary as parametric events are processed, but it never decreases.
With these observations, we can prove \textit{1.}\ by induction
on $\tau$.  If $\tau=\epsilon$ then
$\Dom(\A(\epsilon).\T)=\A(\epsilon).\Theta=\Theta_\epsilon=\{\bot\}$.
Suppose now that 
$\Dom(\A(\tau).\T)=\A(\tau).\Theta=\Theta_\tau$
holds for $\tau\in\paramevents^*$, and
let $e\langle\theta\rangle\in\paramevents$ be
any parametric event.  Then the following concludes the proof of
\textit{1.}:
$$
\begin{array}{lcl}
\Dom(\A(\tau\,e\langle\theta\rangle).\T) & = & \Dom(\A(\tau).\T)\cup
       (\{\theta\}\sqcup\A(\tau).\Theta) \\
& = & \A(\tau).\Theta \cup
  (\{\theta\}\sqcup\A(\tau).\Theta) \\
& = & (\{\bot\} \sqcup \A(\tau).\Theta) \cup
 (\{\theta\}\sqcup\A(\tau).\Theta) \\
& = & \{\bot,\theta\}\sqcup\A(\tau).\Theta \\
& = & \A(\tau\, e\langle\theta\rangle).\Theta \\
& = & \{\bot,\theta\}\sqcup\Theta_\tau \\
& = & \Theta_{\tau\, e\langle\theta\rangle}
\end{array}
$$
where the first equality follows from how the loop at lines
\textsf{\footnotesize 3-5} updates $\T$, the second by the induction hypothesis,
the third by \textit{5.} in Proposition \ref{prop:sqcup-simple},
the fourth by \textit{7.} in Proposition \ref{prop:sqcup-simple},
the fifth by how $\Theta$ is updated at line \textsf{\footnotesize 6},
the sixth again by the induction hypothesis,
and, finally, the seventh by Proposition \ref{prop:Theta-tau}.

Before we continue, let us first prove the following property:
\begin{quote}
$\A(\tau\, e\langle\theta\rangle).\T(\theta') =
\A(\tau).\T(\max\,(\theta']_{\A(\tau).\Theta})\,e$\\
for any $e\langle\theta\rangle\in\paramevents$
and any $\theta'\in\{\theta\}\sqcup\A(\tau).\Theta$.
\end{quote}
One should be careful here to {\em not} get tricked thinking that this
property is straightforward, because it says only what line
\textsf{\footnotesize 4} of $\A$ does.  The complexity comes from the fact 
that if there were two different
$\theta_1,\theta_2\in\{\theta\}\sqcup\A(\tau).\Theta$
such that $\theta_1=\max\,(\theta_2]_{\A(\tau).\Theta}$, then
an unfortunate enumeration of the partial functions $\theta'$ in 
$\{\theta\}\sqcup\A(\tau).\Theta$ by the loop at lines
\textsf{\footnotesize 3-5} may lead to the non-parametric event $e$ to be added
twice to a slice: indeed, if $\theta_1$ is processed before 
$\theta_2$, then $e$ is first added to the end of
$\T(\theta_1)$ when $\theta'=\theta_1$, and then
$\T(\theta_1)\,e$ is assigned to $\T(\theta_2)$ when
$\theta'=\theta_2$; this way, $\T(\theta_2)$ ends up
accumulating $e$ twice instead of once, which is obviously wrong.
Fortunately, since $\A(\tau).\Theta$ is lub closed
(by \textit{1.} above and Proposition \ref{prop:Theta-tau-closed}),
\textit{3.} in Proposition \ref{prop:max} implies that there are no
such different
$\theta_1,\theta_2\in\{\theta\}\sqcup\A(\tau).\Theta$.
Therefore, there is no interference between the various assignments at
line~\textsf{\footnotesize 4}, regardless of the order in which the partial
functions $\theta'\in\{\theta\}\sqcup\Theta$ are enumerated, which
means that, indeed, 
$\A(\tau\, e\langle\theta\rangle).\T(\theta') =
\A(\tau).\T(\max\,(\theta']_{\A(\tau).\Theta})\,e$ 
for any $e\langle\theta\rangle\in\paramevents$
and for any $\theta'\in\{\theta\}\sqcup\A(\tau).\Theta$.
This lack of interference between updates of $\T$ also
suggests an important implementation optimization:
\begin{quote}
{\em The loop at lines \textsf{\footnotesize 3-5} can be parallelized
without duplicating the table $\T$!}
\end{quote}
Of course, the loop can be parallelized anyway if
the table is duplicated and then merged within the original table,
in the sense that all the writes to $\T(\theta')$ are done in a copy
of $\T$.  However, experiments show that the table $\T$ can be
literally huge in real applications, in the order of billions of
entries, so duplicating and merging it can be prohibitive.

\textit{2.} can be now proved by induction on the length of $\tau$.  If
$\tau=\epsilon$ then $\A(\epsilon).\Theta=\{\bot\}$, so
$\theta'\in\A(\epsilon).\Theta$ can only be $\bot$; then
$\A(\epsilon).\T(\bot)=\tau\!\!\upharpoonright_{\bot}=\epsilon$.
Suppose now that
$\A(\tau).\T(\theta')=\tau\!\!\upharpoonright_{\theta'}$
for any $\theta'\in\A(\tau).\Theta$ and let us show
that $\A(\tau\,e\langle\theta\rangle).\T(\theta') =
(\tau\,e\langle\theta\rangle)\!\!\upharpoonright_{\theta'}$
for any $\theta'\in\A(\tau\,e\langle\theta\rangle).\Theta$.
As shown in the proof of \textit{1.} above,
$\A(\tau\, e\langle\theta\rangle).\Theta =
 \A(\tau).\Theta \cup
 (\{\theta\}\sqcup\A(\tau).\Theta)$,
so we have two cases to analyze.  First, if
$\theta'\in\{\theta\}\sqcup\A(\tau).\Theta$ then
$\theta\sqsubseteq\theta'$ and so
$(\tau\,e\langle\theta\rangle)\!\!\upharpoonright_{\theta'}=
\tau\!\!\upharpoonright_{\theta'}\,e$; further,
$$
\begin{array}{lll}
\A(\tau\,e\langle\theta\rangle).\T(\theta')
 & = & 
  \A(\tau).\T(\max\,(\theta']_{\A(\tau).\Theta})\,e \\
 & = & \tau\!\!\upharpoonright_{\max\,(\theta']_{\A(\tau).\Theta}}\,e\\
 & = & \tau\!\!\upharpoonright_{\theta'}\,e \\
 & = & (\tau\,e\langle\theta\rangle)\!\!\upharpoonright_{\theta'},
\end{array}
$$
where the first equality follows by the auxiliary property proved
above, the second by the induction hypothesis using the fact that
$\max\,(\theta']_{\A(\tau).\Theta}\in{\A(\tau).\Theta}$, 
and the third by Proposition \ref{prop:tau-max}.
Second, if $\theta'\in\A(\tau).\Theta$ but
$\theta'\not\in\{\theta\}\sqcup\A(\tau).\Theta$ then
$\theta\not\sqsubseteq\theta'$ and so 
$(\tau\,e\langle\theta\rangle)\!\!\upharpoonright_{\theta'}=
\tau\!\!\upharpoonright_{\theta'}$; furthermore,
$$
\begin{array}{lll}
\A(\tau\,e\langle\theta\rangle).\T(\theta')
 & = & 
  \A(\tau).\T(\theta') \\
 & = & \tau\!\!\upharpoonright_{\theta'} \\
 & = & (\tau\,e\langle\theta\rangle)\!\!\upharpoonright_{\theta'},
\end{array}
$$
where the first equality holds because $\theta'$ is not considered by
the loop in lines \textsf{\footnotesize 3-5} in $\A$, that is,
$\theta'\not\in\{\theta\}\sqcup\A(\tau).\Theta$,
and the second equality follows by the induction hypothesis, as
$\theta'\in{\A(\tau).\Theta}$.
Therefore, $\A(\tau\,e\langle\theta\rangle).\T(\theta') =
(\tau\,e\langle\theta\rangle)\!\!\!\upharpoonright_{\theta'}$
for any $\theta'\in\A(\tau\,e\langle\theta\rangle).\Theta$,
which completes the proof of \textit{2.}

\textit{3.} is the main result concerning our trace slicing algorithm
and it follows now easily:
$$
\begin{array}{lll}
\tau\!\!\upharpoonright_{\theta}
 & = & \tau\!\!\upharpoonright_{\max\,(\theta]_{\Theta_\tau}} \\
 & = & \tau\!\!\upharpoonright_{\max\,(\theta]_{\A(\tau).\Theta}} \\
 & = & \A(\tau).\T(\max\,(\theta]_{\A(\tau).\Theta})
\end{array}
$$
The first equality follows by Proposition \ref{prop:tau-max}, the
second equality by \textit{1.} above, and the third equality by \textit{2.} above,
because $\max\,(\theta]_{\A(\tau).\Theta}\in{\A(\tau).\Theta}$.
This concludes the correctness proof of our trace slicing
algorithm $\A$.
\end{proof}

\section{Monitors and Parametric Monitors}
\label{sec:monitors}

In this section we first define monitors $M$ as a variant of Moore
machines with potentially infinitely many states; then we define
parametric monitors $\parametric{X}{M}$ as monitors maintaining one
state of $M$ per parameter instance.  Like for parametric properties,
which turned out to be just properties over parametric traces, we show
that parametric monitors are also just monitors, but for parametric
events and with instance-indexed states and output categories.  We
also show that a parametric monitor $\parametric{X}{M}$ is a
monitor for the parametric property $\parametric{X}{P}$, with $P$
the property monitored by $M$.

\subsection{The Non-Parametric Case}

We start by defining non-parametric monitors as a variant of (deterministic)
Moore machine \cite{moore-56} that allows infinitely many states:

\begin{definition}
\label{dfn:monitor}
A \textbf{monitor} $M$ is a tuple
$(S,{\cal E}, {\cal C}, \i,
  \sigma:S\times{\cal E}\rightarrow S,
  \gamma:S\rightarrow{\cal C})$,
where $S$ is a set of states, $\cal E$ is a set of input events,
$\cal C$ is a set of output categories, $\i\in S$ is the initial
state, $\sigma$ is the transition function, and $\gamma$ is the
output function.  The transition function is extended to
$\sigma:S\times{\cal E}^*\rightarrow S$ as expected:
$\sigma(s,\epsilon)=s$ and $\sigma(s,we)=\sigma(\sigma(s,w),e)$ for
any $s\in S$, $e\in{\cal E}$, and $w\in{\cal E}^*$.
\end{definition}

The notion of a monitor above is rather conceptual.  Actual
implementations of monitors need not generate all
the state space apriori, but on a ``by need'' basis.  Consider, for
example, a monitor for a property specified using an NFA which
performs an NFA-to-DFA construction on the fly, as events are
received.  Such a monitor generates only those states in the DFA that
are needed by the monitored execution trace.  Moreover, the monitor
only needs to store one such state of the DFA, i.e., set of states in the NFA,
namely the current one: once an event is received, the next state is
(deterministically) computed and the old one is discarded.  Therefore,
assuming that one needs constant space to store a state of the original NFA,
then the memory needed by this monitor is linear in the number of states of
the NFA.  An alternative and probably more conventional monitor could be
one which generates the corresponding DFA statically, paying upfront the
exponential price in both time and space.  As empirically suggested by
\cite{tabakov-vardi-2010-rv}, if one is able to statically generate and store
the corresponding DFA then one should most likely take this route,
because in practice it tends to be much faster to jump to a known
next state than to compute it.

Allowing monitors with infinitely many states is a necessity in our
context.  Even though only a finite number of states is reached during
any given (finite) execution trace, there is, in general, no bound on
how many states are reached.  For example, monitors for context-free
grammars like the ones in \cite{meredith-jin-chen-rosu-2008-ase} have
potentially unbounded stacks as part of their state.  Also, as shown
shortly, parametric monitors have domains of functions as state
spaces, which are infinite as well.  Nevertheless, what is common to
all monitors is that they can classify traces into categories.  When a
monitor does not have enough information about a trace to put it in a
category of interest, we can assume that it actually categorizes it as
a ``don't know'' trace, where ``don't know'' can be regarded as a
special category; this is similar to regarding partial functions as total
functions by adding a special ``undefined'' value in their codomain.
The following is therefore natural:

\begin{definition}
\label{dfn:monitor-for-property}
Monitor $M = (S, {\cal E}, {\cal C}, \i, \sigma, \gamma)$ is a
\textbf{monitor for property} $P:{\cal E}^*\rightarrow{\cal C}$ if and only if
$\gamma(\sigma(\i,w))=P(w)$ for each $w\in{\cal E}^*$.
\end{definition}

A property can be associated to each monitor, in a similar style to
which we can associate a language to each automaton:

\begin{definition}
\label{dfn:M-Property}
Monitor $M = (S, {\cal E}, {\cal C}, \i, \sigma, \gamma)$
defines \textbf{the $M$-property}
${\cal P}_M:{\cal E}^*\rightarrow{\cal C}$ as follows:
${\cal P}_M(w)=\gamma(\sigma(\i,w))$ for each $w\in{\cal E}^*$.
\end{definition}

The following result is straightforward, it follows immediately from
Definitions~\ref{dfn:monitor-for-property} and \ref{dfn:M-Property}.
The only reason we frame it as a numbered proposition is because we
need to refer to it in the proof of Corollary~\ref{cor:trace-monitoring}.

\begin{proposition}
\label{prop:M-Property}
With the notation in Definition \ref{dfn:M-Property},
monitor $M$ is indeed a monitor for its corresponding $M$-property
${\cal P}_M$.  Moreover, a monitor can only be a monitor for one
property, that is, if $M$ is a monitor for property $P$ then $P={\cal P}_M$.
\end{proposition}

Since we allow monitors to have infinitely many states, there is a
strong correspondence between properties and monitors:

\begin{definition}
\label{dfn:P-Monitor}
Property $P:{\cal E}^*\rightarrow{\cal C}$ defines 
\textbf{the $P$-monitor}
${\cal M}_P=(S_P,{\cal E},{\cal C},\i_P, \sigma_P,\gamma_P)$ as
follows:
\begin{itemize}
\item[] $S_P={\cal E}^*$,
\item[] $\i_P=\epsilon$,
\item[] $\sigma_P(w,e)=we$ for each $w\in S_P={\cal E}^*$ and $e\in{\cal E}$,
\item[] $\gamma_P(w)=P(w)$ for each $w\in S_P={\cal E}^*$.
\end{itemize}
\end{definition}

\noindent Thus, ${\cal M}_Pw$ holds traces as states, appends events to
them as transition and, as output, it looks up the category of the
corresponding trace using $P$.  The following results are also
straightforward and, again, we frame them as numbered propositions
only because we will refer to them later.

\begin{proposition}
\label{prop:P-Monitor}
With the notation in Definition \ref{dfn:P-Monitor}, the monitor
${\cal M}_P$ is indeed a monitor for property $P$.
\end{proposition}
\begin{proof}
It follows from the sequence of equalities
$\gamma_P(\sigma_P(\i_P,w))=\gamma_P(\sigma_P(\epsilon,w))=\gamma_P(\epsilon
w)=\gamma_P(w)=P(w)$.
\end{proof}

\begin{proposition}
With the notations in Definitions \ref{dfn:M-Property} and
\ref{dfn:P-Monitor},
${\cal P}_{{\cal M}_P}=P$ for any property 
$P:{\cal E}^*\rightarrow{\cal C}$.
\end{proposition}
\begin{proof}
${\cal P}_{{\cal M}_P}(w)=\gamma_P(\sigma_P(\i_P,w))=P(w)$ for any
$w\in{\cal E}^*$.
\end{proof}\medskip

The equality of monitors ${\cal M}_{{\cal P}_M} = M$ does
not hold for any monitor $M$; it does hold when $M={\cal M}_P$ for
some property $P$, though.

\begin{definition}
\label{dfn:mon-equiv}
Monitors $M$ and $M'$ are \textbf{property equivalent}, or just
\textbf{equivalent}, written $M \equiv M'$, iff they are monitors for
the same property (see Definition \ref{dfn:monitor-for-property}).
With the notation in Definition \ref{dfn:M-Property}, we have that
$M\equiv M'$ iff ${\cal P}_M={\cal P}_{M'}$.
\end{definition}

\begin{proposition}
With the notations in Definitions \ref{dfn:M-Property} and
\ref{dfn:P-Monitor},
${\cal M}_{{\cal P}_M} \equiv M$ for any
monitor $M = (S, {\cal E}, {\cal C}, \i, \sigma, \gamma)$.
\end{proposition}
\begin{proof}
By Definition \ref{dfn:mon-equiv}, 
${\cal M}_{{\cal P}_M} \equiv M$ iff
${\cal P}_{{\cal M}_{{\cal P}_M}} = {\cal P}_M$, and the latter
follows by Proposition \ref{prop:P-Monitor} taking $P$ to be
${\cal P}_M$.
\end{proof}

\subsection{The Parametric Case}

We next define parametric monitors in the same style as the other
parametric entities defined in this paper: starting with a base
monitor and a set of parameters, the corresponding parametric monitor
can be thought of as a set of base monitors running in parallel, one
for each parameter instance.

\begin{definition}
\label{dfn:parametric-monitor}
Given parameter set $X$ with corresponding values $V$ and a monitor
$M=(S,\ {\cal E},\ {\cal C},\ \i,\ \sigma:S\times{\cal E}\rightarrow S,\
  \gamma:S\rightarrow{\cal C})$, we define the
\textbf{parametric monitor} $\parametric{X}{M}$ as the monitor
$$
(\totalf{\XV}{S},\ 
\paramevents,\ \totalf{\XV}{\cal C},\ \lambda\theta.\i,\ 
\parametric{X}{\sigma},\ \parametric{X}{\gamma}),
$$
with
$$
\begin{array}{l}
\parametric{X}{\sigma} : \totalf{\XV}{S} \times
\paramevents\rightarrow\totalf{\XV}{S} \\
\parametric{X}{\gamma}:\totalf{\XV}{S}
\rightarrow\totalf{\XV}{\cal C}
\end{array}
$$
defined as
$$
\begin{array}{l}
(\parametric{X}{\sigma})(\delta,e\langle\theta'\rangle)(\theta) =
\left\{
\begin{array}{ll}
\sigma(\delta(\theta),e) & \mbox{ if } \theta'\sqsubseteq\theta \\
\delta(\theta) & \mbox{ if } \theta'\not\sqsubseteq\theta
\end{array}
\right.
\\
\\
(\parametric{X}{\gamma})(\delta)(\theta)=\gamma(\delta(\theta))
\end{array}
$$
for any $\delta\in\totalf{\XV}{S}$ and any
$\theta,\theta'\in\XV$.
\end{definition}

Therefore, a state $\delta$ of parametric monitor $\parametric{X}{M}$
maintains a state $\delta(\theta)$ of $M$ for each parameter instance
$\theta$, takes parametric events as input, and outputs categories
indexed by parameter instances (one output category of $M$ per
parameter instance).

\begin{proposition}
\label{prop:parametric-monitor}
If $M$ is a monitor for $P$ then parametric monitor
$\parametric{X}{M}$ is a monitor for parametric property
$\parametric{X}{P}$, or, with the notation in Definition
\ref{dfn:M-Property},
${\cal P}_{\parametric{X}{M}}=\parametric{X}{{\cal P}_M}$.
\end{proposition}
\begin{proof}
We show that
$(\parametric{X}{\gamma})((\parametric{X}{\sigma})(\lambda\theta.\i,\tau))
= (\parametric{X}{P})(\tau)$ for any $\tau\in\paramevents^*$, i.e.,
after application on $\theta\in\XV$, that
$\gamma((\parametric{X}{\sigma})(\lambda\theta.\i,\tau)(\theta)) =
P(\tau\!\!\upharpoonright_\theta)$ for any $\tau\in\paramevents^*$
and $\theta\in\XV$.  Since $M$ is a monitor for $P$, it
suffices to show that
$(\parametric{X}{\sigma})(\lambda\theta.\i,\tau)(\theta) =
\sigma(\i,\tau\!\!\upharpoonright_\theta)$ for any
$\tau\in\paramevents^*$ and $\theta\in\XV$.
We prove it by induction on $\tau$.  If $\tau=\epsilon$ then
$(\parametric{X}{\sigma})(\lambda\theta.\i,\epsilon)(\theta) =
(\lambda\theta.\i)(\theta)\! =\! \i\! = \!\sigma(\i,\epsilon) = 
\sigma(\i,\epsilon\!\!\upharpoonright_\theta)$.  Suppose that
$(\parametric{X}{\sigma})(\lambda\theta.\i,\tau)(\theta) =
\sigma(\i,\tau\!\!\upharpoonright_\theta)$ for some arbitrary but
fixed $\tau\in\paramevents^*$ and for any $\theta\in\XV$, and
let $e\langle\theta'\rangle$ be any parametric event in $\paramevents$
and let $\theta\in\XV$ be any parameter instance.  The
inductive step is then as follows:
$$
\!\!\!\begin{array}{lll}
(\parametric{X}{\sigma})(\lambda\theta.\i,\tau\,e\langle\theta'\rangle)(\theta)
 \!\!\!\!& =
\!\!\!\!  & (\parametric{X}{\sigma})((\parametric{X}{\sigma})(\lambda\theta.\i,\tau),e\langle\theta'\rangle)(\theta)
\\
 & = 
\!\!\!\!  & (\parametric{X}{\sigma})(\sigma(\i,\tau\!\!\upharpoonright_\theta),e\langle\theta'\rangle)(\theta)
\\
 & =
\!\!\!\!  & \left\{\begin{array}{ll}
    \sigma(\sigma(\i,\tau\!\!\upharpoonright_\theta),e)
     & \mbox{ if } \theta'\sqsubseteq\theta \\
    \sigma(\i,\tau\!\!\upharpoonright_\theta)
     & \mbox{ if } \theta'\not\sqsubseteq\theta
    \end{array}\right.
\\
 & =
\!\!\!\!  & \left\{\begin{array}{ll}
    \sigma(\i,\tau\!\!\upharpoonright_\theta\,e)
     & \mbox{ if } \theta'\sqsubseteq\theta \\
    \sigma(\i,\tau\!\!\upharpoonright_\theta)
     & \mbox{ if } \theta'\not\sqsubseteq\theta
    \end{array}\right.
\\
 & =
\!\!\!\!  & \sigma(\i,(\tau\,e\langle\theta'\rangle)\!\!\upharpoonright_\theta)
\end{array}
$$
The first equality above follows by the second part of Definition
\ref{dfn:monitor}), the second by the induction hypothesis, the third
by Definition \ref{dfn:parametric-monitor}, the fourth again by the
second part of Definition \ref{dfn:monitor}, and the fifth by
Definition \ref{dfn:trace-slicing}.  This concludes our proof.
\end{proof}

\section{Algorithms for Parametric Trace Monitoring}
\label{sec:trace-monitoring}

We next propose two monitoring algorithms for parametric properties.
Our unoptimized but easier to understand algorithm is easily derived
from the parametric trace slicing algorithm in Figure \ref{fig:A}.
Our second algorithm is an online optimization of the first, which
significantly reduces the size of the search space for compatible
parameter instances when a new event is received.

\subsection{Unoptimized but Simpler Algorithm}

Analyzing the definition of a parametric monitor (Definition
\ref{dfn:parametric-monitor}), the first thing we note is that its
state space is not only infinite, but it is not even enumerable.
Therefore, a first challenge in monitoring parametric properties is
how to represent the states of the parametric monitor.  Inspired by
the algorithm for trace slicing in Figure \ref{fig:A}, we encode
functions $\partialf{\XV}{S}$ as tables with entries indexed by
parameter instances in $\XV$ and with contents states in $S$.  Following
similar arguments as in the proof of the trace slicing algorithm, such
tables will have a finite number of entries provided that each event
instantiates only a finite number of parameters.

\begin{figure}
\begin{center}
$
\begin{array}{l}
\textsf{\footnotesize Algorithm} \ \B(M=(S,{\cal E},{\cal C},\i,\sigma,\gamma)) \\
\textsf{\footnotesize Input:} \ \mbox{finite parametric trace }\tau\in
  \paramevents^*\\
\textsf{\footnotesize Output:} \mbox{ mapping } \Gamma:\partialf{\XV}{\cal C} \mbox{
  and set }\Theta\subseteq\XV\\
\\
1\ \ \Delta \leftarrow \bot;\ \Delta(\bot)\leftarrow\i; \ \Theta \leftarrow \{\bot\} \\
2\ \ \myforall \textit{parametric event } e\langle{\theta}\rangle
\textit{ in order in }\tau \ \mydo \\
3\ \ \hsp \myforall \ \theta' \in \{\theta\} \sqcup \Theta \ \mydo \\
4\ \ \hsp\hsp \Delta(\theta') \leftarrow \sigma(\Delta(\max\,(\theta']_\Theta),e) \\
5\ \ \hsp\hsp \Gamma(\theta') \leftarrow\gamma(\Delta(\theta'))
\mbox{ \ \ \ \ \ \ \ // a message may be output here} \\
6\ \ \hsp \myendfor \\
7\ \ \hsp \Theta \leftarrow \{\bot,\theta\}\sqcup\Theta \\
8\ \ \myendfor
\end{array}
$
\end{center}
\caption{\label{fig:B}Parametric monitoring algorithm $\B$}
\end{figure}

Figure \ref{fig:B} shows our monitoring algorithm for parametric
properties.  Given parametric property $\parametric{X}{P}$ and $M$ a
monitor for $P$, $\B(M)$ yields a monitor that is equivalent to
$\parametric{X}{M}$, that is, a monitor for $\parametric{X}{P}$.
Section \ref{sec:implementation} shows one way to use this algorithm:
a monitor $M$ is first synthesized from the base property $P$, then
that monitor $M$ is used to synthesize the monitor $\B(M)$ for the
parametric property $\parametric{X}{P}$.  $\B(M)$ follows very closely
the algorithm for trace slicing in Figure \ref{fig:A}, the main
difference being that trace slices are processed, as generated, by
$M$: instead of calculating the trace slice of $\theta'$ by appending
base event $e$ to the corresponding existing trace slice in
\textsf{\footnotesize line 4} of $\A$, we now calculate and store in table $\Delta$
the state of the {\em monitor instance} corresponding to $\theta'$ by
sending $e$ to the corresponding existing monitor instance
(\textsf{\footnotesize line 4} in $\B(M)$); at the same time we also calculate the
output corresponding to that monitor instance and store it in table
$\Gamma$.  In other words, we replace trace slices in $\A$ by local
monitors processing online those slices.  In our implementation in
Section \ref{sec:implementation}, we also check whether
$\Gamma(\theta')$ at \textsf{\footnotesize line 5} violates the property and, if so,
an error message including $\theta'$ is output to the user.

\begin{definition}
\label{dfn:B-notation}
Given $\tau\in\paramevents^*$, let $\B(M)(\tau).\Theta$ and
$\B(M)(\tau).\Delta$ and $\B(M)(\theta).\Gamma$ be the three
data-structures maintained by the algorithm $\B(M)$ in Figure
\ref{fig:B} after processing $\tau$.
Let $\bot\mapsto\i=\B(M)(\epsilon).\Delta\in\partialf{\XV}{S}$ be the
partial map taking $\bot\in\XV$ to $\i$ and undefined elsewhere.
\end{definition}

\begin{corollary}
\label{cor:from-thm1}
The following hold for any $\tau\in\paramevents^*$:
\begin{enumerate}[(1)]
\item $\Dom(\B(M)(\tau).\Delta)=\B(M)(\tau).\Theta=\Theta_\tau$;
\item $\B(M)(\tau).\Delta(\theta) =
       \sigma(\i,\tau\!\!\upharpoonright_{\theta})$
      and \\
      $\B(M)(\tau).\Gamma(\theta) = 
       \gamma(\sigma(\i,\tau\!\!\upharpoonright_{\theta}))$
      for any $\theta\in\B(M)(\tau).\Theta$;
\item $\sigma(\i,\tau\!\!\upharpoonright_{\theta})=
       \B(M)(\tau).\Delta(\max\,(\theta]_{\B(M)(\tau).\Theta})$
      and \\
      $\gamma(\sigma(\i,\tau\!\!\upharpoonright_{\theta}))=
       \B(M)(\tau).\Gamma(\max\,(\theta]_{\B(M)(\tau).\Theta})$
for any $\theta\in\XV$.
\end{enumerate}
\end{corollary}
\begin{proof}
Follows from Theorem \ref{thm:trace-slicing} and the discussion above.
\end{proof}

\vspace*{2ex}

We next show how to associate a formal monitor to the algorithm $\B(M)$ in
Figure~\ref{fig:B}:

\begin{definition}
\label{dfn:MB}
For the algorithm $\B(M)$ in Figure \ref{fig:B}, let
$${\cal M}_{\B(M)}=(R,\paramevents,\totalf{\XV}{\cal
  C},\bot\mapsto\i,\textit{next},\textit{out})$$ be the monitor
defined as follows:
\begin{iteMize}{$\bullet$}
\item $R\subseteq\partialf{\XV}{S}$ is the set
$$\{\B(M)(\tau).\Delta\mid\tau\in\paramevents^*\}$$
of reachable $\Delta$'s in $\B(M)$, and
\item $\textit{next}: R \times \paramevents \rightarrow R$ and 
$\textit{out}: R \rightarrow \totalf{\XV}{\cal C}$ are functions
defined as follows, where $\tau\in\paramevents^*$, $e\in{\cal E}$,
and $\theta\in\XV$:
$$
\!\!\!\!
\begin{array}{l}
\textit{next}(\B(M)(\tau).\Delta,e\langle\theta\rangle) =
\B(M)(\tau\,e\langle\theta\rangle).\Delta, \mbox{ and} \\
\textit{out}(\B(M)(\tau).\Delta)(\theta) =
\B(M)(\tau).\Gamma(\max\,(\theta]_{\B(M)(\tau).\Theta}).
\end{array}
$$
\end{iteMize}
\end{definition}

\begin{theorem}
\label{thm:trace-monitoring}
${\cal M}_{\B(M)}\equiv \parametric{X}{M}$ for any monitor $M$.
\end{theorem}
\begin{proof}
All we have to do is to show that, for any $\tau\in\paramevents^*$,
$\textit{out}(\textit{next}(\bot\mapsto\i,\tau))$ and
$(\parametric{X}{\gamma})((\parametric{X}{\sigma})(\lambda\theta.\i,\tau))$
are equal as total functions in $\totalf{\XV}{\cal C}$.  Let
$\theta\in\XV$; then:
$$
\!\!\!
\begin{array}{l@{\,}l@{\,}l}
\textit{out}(\textit{next}(\bot\mapsto\i,\tau))(\theta)
 & = & 
  \textit{out}(\B(M)(\tau).\Delta)(\theta) \\
 & = &
  \B(M)(\tau).\Gamma(\max\,(\theta]_{\B(M)(\tau).\Theta}) \\
 & = &
  \gamma(\sigma(\lambda\theta.\i,\tau\!\!\upharpoonright_{\theta})) \\
 & = &
  \gamma((\parametric{X}{\sigma})(\lambda\theta.\i,\tau)(\theta)) \\
 & = &
  (\parametric{X}{\gamma})((\parametric{X}{\sigma})(\lambda\theta.\i,\tau))(\theta).
\end{array}
$$
The first equality above follows inductively by the definition of
$\textit{next}$ (Definition \ref{dfn:MB}), noticing that
$\bot\mapsto\i=\B(M)(\epsilon).\Delta$.  The second equality follows by the
definition of $\textit{out}$ (Definition \ref{dfn:MB}) and the third
by \textit{3.} in Corollary \ref{cor:from-thm1}.  The fourth equality
above follows inductively by the definition of
$\parametric{X}{\sigma}$ (Definition \ref{dfn:parametric-monitor}) and
has already been proved as part of the proof of Proposition
\ref{prop:parametric-monitor}.  Finally, the fifth equality follows by
the definition of $\parametric{X}{\gamma}$ (Definition~\ref{dfn:parametric-monitor}).

Therefore, ${\cal M}_{\B(M)}$ and $\parametric{X}{M}$ define the same
property.
\end{proof}

\begin{corollary}
\label{cor:trace-monitoring}
If $M$ is a monitor for $P$ and $X$ is a set of parameters,
then ${\cal M}_{\B(M)}$ is a monitor for parametric property
$\parametric{X}{P}$.
\end{corollary}
\begin{proof}
With the notation in Definition \ref{dfn:M-Property}, Theorem
\ref{thm:trace-monitoring} implies that
${\cal P}_{{\cal M}_{\B(M)}}={\cal P}_{\parametric{X}{M}}$.  By
Proposition \ref{prop:parametric-monitor} we have that
${\cal P}_{\parametric{X}{M}}=\parametric{X}{{\cal P}_M}$.
Finally, since $P={\cal P}_M$ by Proposition~\ref{prop:M-Property}, we
conclude that ${\cal P}_{{\cal M}_{\B(M)}}=\parametric{X}{P}$.
\end{proof}

\subsection{Optimized Algorithm}
\label{sec:optimized}

\newcommand{\U}{{\cal U}}
\newcommand{\pf}{{\cal P}_f}
\newcommand{\myendif}{\textsf{\footnotesize endif }}
\newcommand{\myelse}{\textsf{\footnotesize else }}
\newcommand{\myelseif}{\textsf{\footnotesize elseif }}
\newcommand{\myreturn}{\textsf{\footnotesize return }}
\newcommand{\mycontinue}{\textsf{\footnotesize continue }}
\newcommand{\mygoto}{\textsf{\footnotesize goto }}

\begin{figure}
\begin{center}
$
\begin{array}{l}
\textsf{\footnotesize Algorithm} \ \C(M=(S,{\cal E},{\cal C},\i,\sigma,\gamma)) \\
\textsf{\footnotesize Globals:}\ \mbox{mapping }\Delta : \partialf{\XV}{S} \mbox{ and }\\
\  \ \ \ \ \ \ \ \ \ \mbox{mapping }\U : \XV \rightarrow \pf(\XV) \mbox{ and }\\
\  \ \ \ \ \ \ \ \ \ \mbox{mapping }\Gamma : \partialf{\XV}{\cal C}\\
\textsf{\footnotesize Initialization:} \ \U(\theta) \leftarrow \emptyset \mbox{ for any } \theta \in \XV, \Delta(\bot) \leftarrow \i\\
\\
\textsf{\footnotesize function main}(e\langle\theta\rangle)\\
\ \ 1\ \ \myif \Delta(\theta) \mbox{undefined} \ \mythen\\
\ \ 2\ \ \hsp \myforall \theta_{max} \sqsubset \theta\ (\mbox{in reversed topological order})\ \mydo \\
\ \ 3\ \ \hsp \hsp \myif \Delta(\theta_{max}) \mbox{ defined} \ \mythen\\
\ \ 4\ \ \hsp \hsp \hsp \mygoto 7\\
\ \ 5\ \ \hsp \hsp \myendif\\
\ \ 6\ \ \hsp \myendfor\\
\ \ 7\ \ \hsp \textsf{\footnotesize defineTo}(\theta, \theta_{max})\\
\ \ 8\ \ \hsp \myforall \theta_{max} \sqsubset \theta\ (\mbox{in reversed topological order})\ \mydo \\
\ \ 9\ \ \hsp \hsp \myforall \theta_{comp} \in \U(\theta_{max}) \mbox{ that is compatible with } \theta\ \mydo\\
10\ \ \hsp \hsp \hsp \myif \Delta(\theta_{comp} \sqcup \theta) \mbox{ undefined } \mythen \\
11\ \ \hsp \hsp \hsp \hsp \textsf{\footnotesize defineTo}(\theta_{comp} \sqcup \theta, \theta_{comp})\\
12\ \ \hsp \hsp \hsp \myendif\\
13\ \ \hsp \hsp \myendfor\\
14\ \ \hsp \myendfor\\
15\ \ \myendif\\
16\ \ \myforall \theta' \in \{\theta\} \cup \U(\theta) \ \mydo \\
17\ \ \hsp \Delta(\theta') \leftarrow \sigma(\Delta(\theta'), e)\\
18\ \ \hsp \Gamma(\theta') \leftarrow \sigma(\Delta(\theta'))\\
19\ \ \myendfor\\
\\
\textsf{\footnotesize function defineTo}(\theta, \theta')\\
1\ \ \Delta(\theta) \leftarrow \Delta(\theta')\\
2\ \ \myforall \theta'' \sqsubset \theta \ \mydo \\
3\ \ \hsp \U(\theta'') \leftarrow \U(\theta'') \cup \{\theta\}\\
4\ \ \myendfor\\
\end{array}
$
\end{center}
\caption{\label{fig:C}Online parametric monitoring algorithm $\C$}
\end{figure}

Algorithm $\C$ in Figure \ref{fig:C} refines Algorithm $\B$ in Figure
\ref{fig:B} for efficient online monitoring.
Since no complete trace is given in online monitoring, $\C$ focuses on
actions to carry out when a parametric event $e\langle\theta\rangle$
arrives; in other words, it essentially expands the body of the outer
loop in $\B$ (lines 3 to 7 in Figure \ref{fig:B}).  The direct use of 
$\B$ would yield prohibitive runtime overhead when monitoring large
traces, because its inner loop requires search for all parameter
instances in $\Theta$ that are compatible with $\theta$; this search
can be very expensive.  $\C$ introduces an auxiliary data structure
and illustrates a mechanical way to accomplish the search, which also
facilitates further optimizations.
While $\B$ did not require that $\theta$ in $e\langle \theta \rangle$
be of finite domain, $\C$ needs that requirement in order to terminate.
Note that in practice $\Dom(\theta)$ is always finite (because the program
state is finite).

$\C$ uses three tables: $\Delta$, $\U$ and $\Gamma$.  $\Delta$ and
$\Gamma$ are the same as $\Delta$ and $\Gamma$ in $\B$, respectively.
$\U$ is an auxiliary data structure used to optimize the search
``for all $\theta' \in \{\theta\} \sqcup \Theta$'' in $\B$ (line 3 in
Figure \ref{fig:B}). 
It maps each parameter instance $\theta$ into the finite set of parameter instances encountered in $\Delta$ so far that are strictly more informative than $\theta$, i.e.,
$\U(\theta) = \{\theta' \mid \theta' \in \Dom(\Delta) \mbox{ and } \theta \sqsubset \theta'\}$.
Another major difference between $\B$ and $\C$ is that $\C$ does {\em not} maintain $\Theta$ during computation;
instead, $\Theta$ is implicitly captured by the domain of $\Delta$ in $\C$.
Intuitively, the $\Theta$ at the beginning/end of the body of the outer loop in $\B$ is the
$\Dom(\Delta)$ at the beginning/end of $\C$, respectively.
However, $\Theta$ is fixed during the loop at lines 3 to 6 in $\B$ and updated atomically in line~7, 
while $\Dom(\Delta)$ can be changed at any time during the execution of $\C$.

$\C$ is composed of two functions, \textsf{\footnotesize main} and \textsf{\footnotesize defineTo}.
The \textsf{\footnotesize defineTo} function takes two parameter instances, $\theta$
and $\theta'$, and adds a new entry corresponding to $\theta$ into
$\Delta$ and $\U$.  Specifically, it sets $\Delta(\theta)$ to
$\Delta(\theta')$ and adds $\theta$ into the set $\U(\theta'')$
for each $\theta'' \sqsubset \theta$.

The \textsf{\footnotesize main} function differentiates two cases when a new event
$e\langle\theta\rangle$ is received and processed.  The simpler case
is that $\Delta$ is already
defined on $\theta$, i.e., $\theta \in \Theta$ at the beginning of the
iteration of the outer loop in $\B$.  In this case,
$\{\theta\} \sqcup\Theta=\{\theta' \mid \theta' \in \Theta \mbox{ and
} \theta \sqsubseteq \theta'\} \subseteq \Theta$, so the lines 
3 to 6 in $\B$ become precisely the lines 16 to 19 in $\C$.
In the other case, when $\Delta$ is not already defined on $\theta$,
\textsf{\footnotesize main} takes two steps to handle $e$.
The first step searches for new parameter instances introduced by
$\{\theta\} \sqcup \Theta$ and adds entries for them into $\Delta$
(lines~2 to 14).
We first add an entry to $\Delta$ for $\theta$ at lines 2 to 7.
Then we search for all parameter instances $\theta_{comp}$ that are
compatible with $\theta$, making use of $\U$ (lines 8 and 9); for each
such $\theta_{comp}$, an appropriate entry is added to $\Delta$ for
its lub with $\theta$, and $\cal U$ updated accordingly (lines 10 to 12).
This way, $\Delta$ will be defined on all the new parameter instances
introduced by $\{\theta\} \sqcup \Theta$ after the first step.
In the second step, the related monitor states and outputs are updated
in a similar way as in the first case (lines 16 to 19).  It is
interesting to note how $\C$ searches at lines 2 and 8 for the
parameter instance $\max\,(\theta]_\Theta$ that $\B$ refers to at line
4 in Figure \ref{fig:B}: it enumerates all the
$\theta_{max}\sqsubset\theta$ in {\em reversed topological order}
(from larger to smaller); \textit{1.} in Proposition \ref{prop:max}
guarantees that the maximum exists and, since it is unique, our search
will find it.

\vspace{1ex}
\noindent
\textbf{\em Correctness of $\C$.}  We prove the correctness of
$\C$ by showing that it is equivalent to the body of the outer loop in
$\B$.  Suppose that parametric trace $\tau$ has already been
processed by both $\C$ and $\B$, and a new event
$e\langle\theta\rangle$ is to be processed next.

Let us first note that $\C$ terminates if $\Dom(\theta)$ is finite.
Indeed, if $\Dom(\theta)$ is finite then there is only a finite number of partial
maps less informative than $\theta$, that is, only a finite number of iterations
for the loops at lines 2 and 8 in \textsf{\footnotesize main}; since $\U$ is only
updated at line 3 in \textsf{\footnotesize defineTo}, $\U(\theta)$ is finite for any
$\theta \in \XV$ and thus the loop at line 9 in \textsf{\footnotesize main} also
terminates.  Assuming that running the base monitor $M$ takes constant
time, the worst case complexity of $\C(M)$ is $O(k \times l)$ to
process $e \langle \theta \rangle$, where $k$ is
$2^{\mid \Dom(\theta) \mid}$ and $l$ is the number of incompatible
parameter instances in $\tau$.  Parametric properties often have a
fixed and small number of parameters, in which case $k$ is not
significant.  Depending on the trace, $l$ can unavoidably grow
arbitrarily large; in the worst case, each event may carry an instance
incompatible with the previous ones.

\newtheorem{lemma}{Lemma}

\begin{lemma}
\label{lemma:u}
In the algorithm $\C$ in Figure \ref{fig:C},
$\U(\theta) = \{\theta' \mid \theta' \in \Dom(\Delta) \mbox{ and }
\theta \sqsubset \theta'\}$ before and after each execution of
\textsf{\footnotesize defineTo}, for all $\theta\in\XV$.
\end{lemma}

\begin{proof}
By how $\C$ is initialized,
for any $\theta \in \XV$ we have
$\emptyset = \U(\theta) = \{\theta' \mid \theta' \in \Dom(\Delta)
\mbox{ and } \theta \sqsubset \theta'\}$ 
before the first execution of \textsf{\footnotesize defineTo}.
Now suppose that 
$\U(\theta) = \{\theta' \mid \theta' \in \Dom(\Delta) \mbox{ and }
\theta \sqsubset \theta'\}$ for any $\theta\in\XV$ before an execution
of \textsf{\footnotesize defineTo} and show that it also holds after the execution of
\textsf{\footnotesize defineTo}.
Since $\textsf{\footnotesize defineTo}(\theta,\theta')$ adds a new parameter instance
$\theta$ into $\Dom(\Delta)$ and also adds $\theta$ into the set
$\U(\theta'')$ for any $\theta''\in\XV$ with $\theta'' \sqsubset \theta$,
we still have $\U(\theta) = \{\theta' \mid \theta' \in \Dom(\Delta)
\mbox{ and } \theta \sqsubset \theta'\}$ for any $\theta\in\XV$ after
the execution of \textsf{\footnotesize defineTo}.  Also, the only way $\C$ can add a
new parameter instance $\theta$ into $\Dom(\Delta)$ is by using
\textsf{\footnotesize defineTo}.  Therefore the lemma holds.
\end{proof}

\newcommand{\DB}{\Delta_{\mathbb{B}}(\tau)}
\newcommand{\GB}{\Gamma_{\mathbb{B}}(\tau)}
\newcommand{\DBE}{\Delta_{\mathbb{B}}(\tau e)}
\newcommand{\GBE}{\Gamma_{\mathbb{B}}(\tau e)}

\newcommand{\DC}{\Delta_{\mathbb{C}}}
\newcommand{\UC}{\U_{\mathbb{C}}}
\newcommand{\GC}{\Gamma_{\mathbb{C}}}
\newcommand{\TC}{\Theta_{\mathbb{C}}}

\newcommand{\DCB}{\DC^{b}}
\newcommand{\DCM}{\DC^{m}}
\newcommand{\UCM}{\UC^{m}}
\newcommand{\GCB}{\GC^{b}}
\newcommand{\DCE}{\DC^{e}}
\newcommand{\GCE}{\GC^{e}}

The next theorem proves the correctness of $\C$.
Before we state and prove it, let us recall some previously introduced
notation and also introduce some new useful notation.  First, recall from
Definition~\ref{dfn:B-notation} that $\B(M)(\tau).\Delta$ and
$\B(M)(\tau).\Gamma$ are the $\Delta$ and $\Gamma$ data-structures
of $\B(M)$ after it processes trace
$\tau$.  Also, recall that we fixed parametric trace $\tau$ and event
$e\langle\theta\rangle$.
For clarity, let $\UC$, $\DC$, and $\GC$ be the three data-structures
maintained by $\C(M)$ (in other words, we index the data-structures
with the symbol $\mathbb{C}$).
Let $\DCB$ and $\GCB$ be the $\DC$ and $\GC$ when
$\textsf{\footnotesize main}(e\langle\theta\rangle)$ begins (``$b$'' stays for ``at
the \underline{b}eginning'');  
let $\DCE$ and $\GCE$ be the $\DC$ and $\GC$ when
$\textsf{\footnotesize main}(e\langle\theta\rangle)$ ends (``$e$'' stays for ``at
the \underline{e}nd''; and let $\DCM$ and $\UCM$ be the $\DC$ and
$\UC$ when $\textsf{\footnotesize main}(e\langle\theta\rangle)$ reaches line 16
(``$m$ stays for ``in the \underline{m}iddle'').

\begin{theorem}
\label{thm:implementation}
The following hold:
\begin{enumerate}[(1)]
\item
$\Dom(\DCM) = \{\bot, \theta\} \sqcup \Dom(\DCB)$;
\item
$\DCM(\theta')\! =\! \DCM(\max(\theta']_{\Dom(\DCB)}\!)$, for all $\theta'\! \in \!\Dom(\DCM)$;
\item
If $\DCB = \B(M)(\tau).\Delta$ and $\GCB = \B(M)(\tau).\Gamma$,
then $\DCE = \B(M)(\tau\,e\langle\theta\rangle).\Delta$ and
$\GCB = \B(M)(\tau\,e\langle\theta\rangle).\Gamma$.
\end{enumerate}
\end{theorem}

\begin{proof}
Let $\TC = \Dom(\DCB) = \Dom(\DB)$ and $\DB = \B(M)(\tau\langle\theta\rangle).\Delta$ for simplicity.

\vspace{1ex}
\noindent
\textit{1.}
There are two cases to analyze, depending upon whether $\theta$ is in 
$\TC$ or not.
If $\theta \in \TC$ then the lines 2 to 14 are skipped
and $\Dom(\DC)$ remains unchanged, that is,
$\{\bot, \theta\} \sqcup \TC = \TC =
\Dom(\DCB) = \Dom(\DCM)$ when $\textsf{\footnotesize main}(e\langle\theta\rangle)$ reaches line 16.
If $\theta \notin \TC$ then lines 2 to 14 are executed
to add new parameter instances into $\Dom(\DC)$. 
First, an entry for $\theta$ will be added to $\DC$ at line 7.
Second, an entry for $\theta_{comp}\sqcup\theta$ will be added to
$\DC$ at line 11 (if $\DC$ not already defined on
$\theta_{comp}\sqcup\theta$) eventually for any
$\theta_{comp}\in\TC$ compatible with $\theta$: that
is because $\theta_{max}$ can also be $\bot$ at line 8, in which case
Lemma \ref{lemma:u} implies that $\U(\theta_{max})=\TC$.
Therefore, when line 16 is reached, $\Dom(\DCM)$ is defined on all
the parameter instances in 
$\{\theta\}\cup(\{\theta\}\sqcup\TC)$.  Since
$\bot\in\TC$, the latter equals 
$\{\theta\}\sqcup\TC$, and since $\DCM$ remains
defined on $\TC$, we conclude that $\DCM$ is defined
on all instances in 
$(\{\theta\}\sqcup\TC)\cup\TC$, which
by \textit{5.} and \textit{7.} in Proposition \ref{prop:sqcup-simple}
equals $\{\bot, \theta\} \sqcup \TC$.

\vspace{1ex}
\noindent
\textit{2.} 
We analyze the same two cases as above.
If $\theta \in \TC$ then lines 2 to 14 are skipped
and $\Dom(\DC)$ remains unchanged.  Then
$\max\,(\theta']_{\TC}=\theta'$ for each
$\theta'\in\Dom(\DCM)$, so the result follows.
Suppose now that $\theta \notin \TC$.  By \textit{1.}
and its proof, each $\theta'\in\Dom(\DCM)$
is either in $\TC$ or otherwise in 
$(\{\theta\} \sqcup \TC) - \TC$.
The result immediately holds when
$\theta' \in \TC$ as
$\max\,(\theta']_{\TC} = \theta'$
and $\Delta(\theta')$ stays unchanged until line 16.
If $\theta' \in (\{\theta\} \sqcup \TC) -
\TC$ then $\Delta(\theta')$ is set at either line 7
($\theta'=\theta$) or at line 11 ($\theta'\neq\theta$):

(a) For line 7, the loop at lines 2 to 6 checks all the parameter
instances that are less informative than $\theta$ to find the first
one in $\TC$ in reversed topological order (i.e., if
$\theta_1 \sqsubset \theta_2$ then $\theta_2$ will be checked before
$\theta_1$).
Since by \textit{1.} in Proposition \ref{prop:max} we know that 
$\max\,(\theta]_{\TC}\in\TC$ exists
(and it is unique), the loop at lines 2 to 6 will break precisely
when $\theta_{max}=\max\,(\theta]_{\TC}$, so the
result holds when $\theta'=\theta$ because of the entry introduced for
$\theta$ in $\DC$ at line 7 and because the remaining lines 8 to 14
do not change $\DC(\theta)$.

(b) When $\DC(\theta')$ is set at line 11, note that
the loop at lines 8 to 14 also iterates over all
$\theta_{max} \sqsubset \theta$ in reversed topological order, so
$\theta'=\theta_{comp}\sqcup\theta$ for some
$\theta_{comp}\in\TC$ compatible with $\theta$ such
that $\theta_{max}\sqsubset\theta_{comp}$, where
$\theta_{max}\sqsubset\theta$ is such that there is no
other $\theta'_{max}$ with
$\theta_{max}\sqsubset\theta'_{max}\sqsubset\theta$ and
$\theta'=\theta'_{comp}\sqcup\theta$ for some
$\theta'_{comp}\in\TC$ compatible with $\theta$ such
that $\theta'_{max}\sqsubset\theta'_{comp}$.
%
We claim that there is only one such $\theta_{comp}$, which is
precisely $\max\,(\theta']_{\TC}$:
Let $\theta'_{comp}$ be the parameter instance
$\max\,(\theta']_{\TC}$.
The above implies that
$\theta_{comp}\sqsubseteq\theta'_{comp}\sqsubseteq\theta'$.
Also, $\theta'_{comp}\sqcup\theta=\theta'$ because
$\theta'=\theta_{comp}\sqcup\theta\sqsubseteq\theta'_{comp}\sqcup\theta\sqsubseteq\theta'$.
Let $\theta'_{max}$ be $\theta'_{comp}\sqcap\theta$, that is, the
largest with $\theta'_{max}\sqsubseteq\theta'_{comp}$ and 
$\theta'_{max}\sqsubseteq\theta$ (we let its existence as exercise).
It is relatively easy to see now that
$\theta_{comp}\sqsubset\theta'_{comp}$ implies
$\theta_{max}\sqsubset\theta'_{max}$ (we let it as an exercise, too),
which contradicts the assumption of this case that $\DC$ was not
defined on $\theta'$.  Therefore,
$\theta_{comp}=\max\,(\theta']_{\TC}$ before line 11
is executed, which means that, after line 11 is executed,
$\DC(\theta') = \DC(\max\,(\theta']_{\TC})$;
moreover, none of these will be changed anymore until line 16 is
reached, which proves our result.

\vspace{1ex}
\noindent
\textit{3.}
Since $\Gamma$ is updated according to $\Delta$ in both $\C$ and $\B$,
it is enough to prove that $\DCE = \DBE$.
For $\B$, we have \\
1) $\Dom(\DBE) = \{\bot, \theta\} \sqcup \TC = (\{\theta\} \sqcup \TC) \cup \TC$;\\
2) $\forall$ $\theta' \in \{\theta\} \sqcup \TC$, $\DBE(\theta') = \sigma(\DB(\max,(\theta']_{\TC}), e)$;\\
3) $\forall$ $\theta' \in \TC - \{\theta\} \sqcup \TC$, $\DBE(\theta') = \DB(\theta')$.\\
So we only need to prove that\\
1) $\Dom(\DCE) = \{\bot, \theta\} \sqcup \TC$;\\
2) $\forall$ $\theta' \in \{\theta\} \sqcup \TC$, $\DCE(\theta') = \sigma(\DCB(\max,(\theta']_{\TC}), e)$;\\
3) $\forall$ $\theta' \in \TC - \{\theta\} \sqcup \TC$, $\DCE(\theta') = \DCB(\theta')$.\\
By \textit{1.}, we have $\Dom(\DCM) = \{\bot,
\theta\} \sqcup \TC$.  Since lines 16 to 19 do not change
$\Dom(\DC)$, $\Dom(\DCE) = \Dom(\DCM) = \{\bot, \theta\} \sqcup
\TC$. 1) holds.  

By \textit{2.} and Lemma \ref{lemma:u}, $\DCM(\theta') =
\DCB(\max,(\theta']_{\TC})$ for any $\theta' \in \Dom(\DCM)$.
Also, notice that line 17 sets $\DC(\theta')$ to $\sigma(\DC(\theta'), e)$, which is $\sigma(\DCB(\max,(\theta']_{\TC}), e)$, for the $\theta'$ in the loop.
So, to show 2) and 3), we only need to prove that the loop at line 16 to 19 iterates over $\{\theta\} \sqcup \TC$.
Since lines 16 to 19 do not change $\UC$, we need to show $\{\theta\} \cup \UCM(\theta) = \{\theta\} \sqcup \TC$.
Since $\Dom(\DCM) = \{\bot, \theta\} \sqcup \TC$, we have $\{\theta\} \sqcup \Dom(\DCM) = \{\theta\} \sqcup (\{\bot, \theta\} \sqcup \TC) = \{\theta\} \sqcup ((\{\theta\} \sqcup \TC) \cup \TC)$.
By Proposition \ref{prop:sqcup-simple}, $\{\theta\} \sqcup \Dom(\DCM) = (\{\theta\} \sqcup (\{\theta\} \sqcup \TC)) \cup (\{\theta\} \sqcup \TC) = (\{\theta\} \sqcup \TC) \cup (\{\theta\} \sqcup \TC) = \{\theta\} \sqcup \TC$.
Also, as $\theta \in \Dom(\DCM)$, we have $\{\theta\} \sqcup \Dom(\DCM) = \{\theta' \mid \theta' \in \Dom(\DCM) \mbox{ and } \theta \sqsubseteq \theta'\} = \{\theta\} \sqcup \UCM(\theta)$ by Lemma \ref{lemma:u}.
So $\{\theta\} \cup \UCM(\theta) = \{\theta\} \sqcup \TC$.
\end{proof}

We conclude this section with a discussion on the complexity of the parametric monitoring algorithms 
$\A$ and $\C$ above.  Note that, in the worst case, to process a newly received parametric event
$e\langle\theta\rangle$ after $\A$ or $\C$ has already processed a parametric trace $\tau$, each of $\A$
or $\C$ takes at least linear time/space in the number of $\theta$-compatible parameter instances occurring
in events in $\tau$.  Indeed, $\A$ iterates explicitly through all such parameter instances
(line 3 in Figure~\ref{fig:B}), while $\C$ optimizes this traversal by only enumerating through maximal
parameter instances; in the worst case, we can assume that $\tau$ is such that each event comes with a new
parameter instance which is maximal, so in the worst case $\A$ and $\C$ can take linear time/space in $\tau$
to process $e\langle\theta\rangle$, which is, nevertheless, bad.  Indeed, it means that monitoring some traces
is incrementally slower (with no upper bound) as events are received, until the monitor eventually
runs out of resources.

Unfortunately, there is nothing to be fundamentally done to avoid
the problem above.  It is an inherent problem of parametric monitoring.
Consider, for example, the ``authenticate before use'' parametric property
specified in Section~\ref{sec:authenticate} using parametric LTL as
$\parametric{k}{\always(\textsf{\footnotesize use}\langle k \rangle \rightarrow \eventuallyPast
\textsf{\footnotesize authenticate}\langle k \rangle)}$; to make it clear that events depend on the parameter
key $k$, we tagged them with the key.  Without any knowledge about the semantics of the program to
monitor, any monitor for this property {\em must} store all the authenticated keys, i.e., all the instances of
the parameter $k$.  Indeed, without that, there is no way to know whether a key instance has been
authenticated or not when a \textsf{\footnotesize use} event is observed on that key.
The number of such key instances is theoretically unbounded, so, in the worst case, any
monitor for this property can be incrementally slower and eventually run out of resources.

As seen in Section~\ref{sec:implementation}, the runtime overhead due to
opaque monitoring of parametric properties tends to be manageable in practice.  By ``opaque''
we mean that no semantic information about the source code of the monitored program is used.
If the lack of an efficiency guarantee is a problem in some applications, then the alternative
is to statically analyze the monitored program and to use the obtained semantic information
to eliminate the need for monitoring.  For example, static analyses like those in
\cite{pql-oopsla,bodden-chen-rosu-2009-aosd,clara,dwyer-purandare-person-2010-rv} may
significantly reduce the need for instrumentation, even eliminate it completely.
Moreover, model-checking techniques for parametric properties could also be used for actually proving
that the properties hold and thus they need not be monitored; however, we are not aware of model checking
approaches to verifying parametric properties as presented in this paper.

\section{Implementation in JavaMOP and RV}
\label{sec:implementation}

The discussed parametric monitoring technique is now
fully implemented in two runtime verification systems, namely in
JavaMOP (see \texttt{http://javamop.org}) and in RV \cite{meredith-rosu-2010-rv}
(developed by a startup company, Runtime Verification, Inc.;
the RV system is currently publicly unavailable -- contact the first
author for an NDA-protected version of RV).
Here we first informally discuss several optimizations implemented
in the two runtime verification systems, and then we discuss our experiments
and the evaluation of the two systems.

\subsection{Implementation Optimizations}
Both JavaMOP and RV apply several optimizations
to the algorithm $\C$ in Section \ref{sec:optimized}, to reduce its runtime overhead.
These  are not discussed in depth here, because they are
orthogonal to the main objective of this paper.

\subsubsection{Optimizations in JavaMOP}

Note that $\C$ iterates through all the possible parameter
instances that are less informative than $\theta$ in three different
loops: at lines 2 and 8 in the \textsf{\footnotesize main} function, and at line 2
in the \textsf{\footnotesize defineTo} function.  Hence, it is important to reduce
the number of such instances in each loop.  Even though our semantics
and theoretical algorithms for parametric monitoring in this paper
work with infinite sets of parameters, our current implementation in
JavaMOP assumes that the set of parameters $X$ is bounded and fixed
apriori (declared as part of the specification to monitor).
A simple analysis of the events appearing in the specification
allows to quickly detect parameter instances that can never appear as
lubs of instances of parameters carried by events; maintaining any
space for those in $\Delta$, or $\Gamma$, or iterating over them in
the above mentioned loops, is a waste.  For example, if a
specification contains only two event definitions,
$e_1\langle p_1 \rangle$ and $e_2\langle p_1, p_2 \rangle$,
parameter instances defining only parameter $p_2$ can never appear as
lubs of observed parameter instances.
A static analysis of the specification, discussed in
\cite{chen-meredith-jin-rosu-2009-ase,bodden-chen-rosu-2009-aosd},
exhaustively explores all possible event combinations that can lead to
situations of interest to the property, such as to violation,
validation, etc.  Such information is useful to reduce the number of
loop iterations by skipping iterations over parameter instances that
cannot affect the result of monitoring.  These static analyses are
currently used at compile time in our new JavaMOP implementation to
unroll the loops in $\C$ and reduce the size of $\Delta$ and $\U$.

Another optimization is based on the observation that 
it is convenient to start the monitoring process only when certain events
are received.  Such events are called monitor creation events in
\cite{chen-rosu-2007-oopsla}.  The parameter instances carried by such
creation events may also be used to reduce the number of parameter
instances that need to be considered.  An extreme, yet surprisingly
common case is when creation events instantiate {\em all} the property
parameters.  In this case, the monitoring process does not need to
search for compatible parameter instances even when an event with an
incomplete parameter instance is observed.  The old JavaMOP
\cite{chen-rosu-2007-oopsla} supported only traces whose monitoring
started with a fully instantiated monitor creation event; 
this was perceived (and admitted) as a performance-tradeoff limitation
of JavaMOP \cite{oopsla07abc} (and \cite{chen-rosu-2007-oopsla}).
Interestingly, it now becomes just a common-case optimization of our
novel, general and unrestricted technique presented here.

\subsubsection{Optimizations in RV}
The RV system implements all the optimizations in JavaMOP and adds
two other important optimizations that significantly reduce the overhead.

The first additional optimization of RV is a non-trivial garbage-collector \cite{jin-meredith-griffith-rosu-2011-pldi}.
Note that JavaMOP also has a garbage collector, but it only does the obvious:
it garbage collects a monitor instance only when all its corresponding parameter
instances are collected.  RV performs a static analysis of the property
to monitor and, based on that, it garbage collects a monitor instance as soon
as it realizes that it can never trigger in the future.  This
can happen when any triggering behavior needs at least one event that
can only be generated in the presence of a parameter instance that is already
dead.  Consider, for example, the safe iterator example in Section~\ref{ex:safe-iterators},
and consider that iterator $i_7$ is created for collection $c_3$.  Then a monitor instance
corresponding to the parameter instance $\langle c_3\ i_7 \rangle$ is created and manged.
Suppose that, at some moment, the iterator $i_7$ is garbage collected by the JVM.
Can the monitor instance corresponding to $\langle c_3\ i_7 \rangle$ be garbage collected?
Not in JavaMOP, because, for safety, JavaMOP collects a monitor only when all its parameter
instances are collected, and in this case $c_3$ is still alive.  However, this monitor is flagged for
garbage collection in RV.  The rationale for doing so is that the only way for the monitor to
trigger is to eventually encounter a \textsf{\footnotesize next} event with $i_7$ as parameter,
but that event can never be generated because $i_7$ is dead.  Note, on the other hand,
that the monitor $\langle c_3\ i_7 \rangle$ cannot be garbage collected if $c_3$ is collected but
$i_7$ is still alive, because the iterator alone can still violate the safe-iterator property, even if
its corresponding collection is already dead.

Runtime verification systems like JavaMOP and Tracematches use off-the-shelf
weak reference libraries to implement their garbage collectors.  However, it turns
out that these libraries, in order to be general and thus serve their purpose,
perform many checks that are unnecessary in the context of monitoring.
The second optimization of RV in addition to those of JavaMOP consists of a collection of
data structures based on weak references, which was carefully engineered to
take full advantage of the particularities of monitoring parametric properties.
These data structures allow for effective indexing and lazy collection of monitors,
to minimize the number of expensive traversals of the entire pool of monitors.
Moreover, RV caches monitor instances to save time when the same monitor instances are
accessed frequently.  For example, there is a high chance that the same iterator is accessed
several times consecutively by a program, in which case saving and then retrieving the same
corresponding monitor instances from the data-structures at each iterator access can take
considerable unnecessary overhead.

\subsection{Experiments and Evaluation}

\label{sec:imp_eval}

We next discuss our experience with using the two runtime verification systems that implement
optimized variants of the parametric property monitoring techniques describe in this paper.
Also, we compare their performance with that of {Tracematches}, which is, at our knowledge,
the most efficient runtime verification system besides JavaMOP and RV.  Recall that Tracematches
achieves virtually the same semantics of parametric monitoring like ours, but using a considerably
different approach.

\subsubsection{Experimental Settings}
We used a Pentium 4 2.66GHz / 2GB RAM / Ubuntu 9.10
machine and version 9.12 of the DaCapo (DaCapo 9.12) benchmark
suite~\cite{DaCapo:paper}, the most up-to-date version.  We also present some of the
results of our experiments using the previous version of DaCapo,
2006-10 MR2 (DaCapo 2006-10), namely those for the {bloat} and {jython} benchmarks.
DaCapo 9.12 does not provide the {bloat} benchmark from the DaCapo 2006-10, which
we favor because it generates large overheads when monitoring
{iterator}-based properties.  The {bloat} benchmark with the
\unsafeiter{} specification causes 19194\% runtime overhead (i.e., 192~times
slower) and uses 7.7MB of heap memory in {Tracematches}, and causes
569\% runtime overhead and uses 147MB in {JavaMOP}, while the original
program uses only 4.9MB.  Also, although the DaCapo 9.12 provides
{jython}, {Tracematches} cannot instrument {jython}
due to an error that we were not able to understand or fix.  Thus, we present the result of
{jython} from the DaCapo 2006-10.  The default data input for DaCapo was used and the
{-converge} option to obtain the numbers after convergence within
$\pm3\%$. 
Instrumentation introduces a
different garbage collection behavior in the monitored program, sometimes
causing the program to slightly outperform the original program; this accounts
for the negative overheads seen in both runtime and memory. 

We used the Sun JVM 1.6.0 for the entire evaluation.  The AspectJ
compiler ({ajc}) version 1.6.4 is used for weaving the aspects generated by {JavaMOP} and
{RV} into the target benchmarks.  Another AspectJ compiler, 
{abc}~\cite{abc-05} 1.3.0, is used for weaving 
{Tracematches} properties because {Tracematches}
is part of {abc} and does not work with {ajc}. 
For {JavaMOP}, we used the most recent
release version, 2.1.2. For
{Tracematches}, we used the most recent release version, 1.3.0,
from~\cite{tm-bench-march-2008}, which is included in the {abc}
compiler as an extension.  To figure out the reason that some examples do not
terminate when using {Tracematches}, we also used the {abc} compiler
for weaving aspects generated by JavaMOP and {RV}.  Note that JavaMOP and {RV} are
AspectJ compiler-independent. They show similar overheads and
terminate on all examples when using the {abc} compiler for
weaving as when {ajc} is used.  Because the overheads are similar,
we do not present the results of using {abc} to weave JavaMOP and 
{RV} generated aspects in this paper.  However, using {abc} to weave
JavaMOP and {RV} properties confirms that the high overhead and
non-termination come from {Tracematches} itself, not from
the~{abc}~compiler.

The following properties are used in our experiments.  Some of them were already discussed
in Sections \ref{sec:motivating-examples} and \ref{sec:examples}, others are borrowed from
\cite{bodden-chen-rosu-2009-aosd, tm-static-ecoop,
  meredith-jin-chen-rosu-2008-ase,
  chen-meredith-jin-rosu-2009-ase}.\medskip

\begin{iteMize}{$\bullet$}
\item \hasnext: Do not use the next element in an iterator without
  checking for the existence of it;\medskip

\item \unsafeiter: Do not update a collection when using the iterator
  interface to iterate its elements;\medskip

\item \unsafemap: Do not update a map when using the iterator
  interface to iterate its values or its keys;\medskip

\item \unsafesynccoll: If a collection is synchronized, then its
  iterator also should be accessed synchronously;\medskip

\item \unsafesyncmap: If a collection is synchronized, then its iterators on values and keys also should be accessed synchronized.\medskip
\end{iteMize}

\noindent All of them are tested on {Tracematches}, {JavaMOP}, and {RV}
for comparison. We also monitored all five properties at the same time in {RV},
which was not possible in other monitoring systems for performance reasons or structural limitations.

\begin{figure*}[t!]
\begin{center}
\scalebox{0.65}{
\begin{tabular}{|c@{\hspace{2pt}}
|@{\hspace{2pt}}r@{\hspace{2pt}}
|@{\hspace{2pt}}rrr@{\hspace{2pt}}
|@{\hspace{2pt}}rrr@{\hspace{2pt}}
|@{\hspace{2pt}}rrr@{\hspace{2pt}}
|@{\hspace{2pt}}rrr@{\hspace{2pt}}
|@{\hspace{2pt}}rrr@{\hspace{2pt}}
||r|}
\hline

 & \multicolumn{1}{@{\hspace{2pt}}c|@{\hspace{2pt}}}{\property{Orig} (sec)}
 & \multicolumn{3}{@{\hspace{2pt}}c|@{\hspace{2pt}}}{\hasnext}
 & \multicolumn{3}{@{\hspace{2pt}}c|@{\hspace{2pt}}}{\unsafeiter}
 & \multicolumn{3}{@{\hspace{2pt}}c|@{\hspace{2pt}}}{\property{UnsafeMapIter}} 
 & \multicolumn{3}{@{\hspace{2pt}}c|@{\hspace{2pt}}}{\property{UnsafeSyncColl}} 
 & \multicolumn{3}{@{\hspace{2pt}}c@{\hspace{2pt}}||}{\property{UnsafeSyncMap}} 
 & \multicolumn{1}{@{\hspace{2pt}}c|}{\property{All}} \\

\hline

(A)
 &
 & {TM}
 & {MOP}
 & {RV}
 & {TM}
 & {MOP}
 & {RV}
 & {TM}
 & {MOP}
 & {RV}
 & {TM}
 & {MOP}
 & {RV}
 & {TM}
 & {MOP}
 & {RV}
 & {RV}
\\

\hline

 bloat
& 3.6
& 2119
& 448
& 116
& 19194
& 569
& 251
& $\infty$
& 1203
& 178
& 1359
& 746
& 212
& 1942
& 716
& 130
& 982
\\

jython
& 8.9
& 13
& 0
& 0
& 11
& 0
& 1
& 150
& 18
& 3
& 11
& 1
& 1
& 10
& 0
& 0
& 4
\\

%
%
%
%
%
%
%
%

\hline

 avrora
& 13.6
& 45
& 54
& 55

& 637
& 311
& 118

& $\infty$
& 113
& 42

& 75
& 144
& 80

& 54
& 74
& 16
& 275
\\

 batik
& 3.5
& 3
& 2
& 3

& 355
& 9
& 8

& $\infty$
& 8
& 5

& 208
& 9
& 9

& 5
& 3
& 0
& 28
\\

 eclipse
& 79.0

& -2
& 4
& -1

& 0
& -1
& -1

& 5
& -3
& 0

& -4
& 2
& 1

& $\infty$
& -1
& -1

& 0
\\

 fop
& 2.0
& 200
& 49
& 48

& 350
& 21
& 13

& $\infty$
& 58
& 14

& $\infty$
& 78
& 25

& $\infty$
& 71
& 19

& 133
\\

 h2
& 18.7
& 89
& 17
& 13

&128
& 9
& 4

& 1350
& 21
& 6

& 868
& 21
& 4

& 83
& 20
& 5

& 23
\\

luindex
& 2.9
& 0
& 0
& 1

& 0
& 0
& 1

& 1
& 4
& 1

& 1
& 1
& 1

& 2
& 0
& 0

& 1
\\

lusearch
& 25.3
& -1
& 1
& 0

& 1
& 2
& 2

& 2
& 2
& 0

& 4
& 0
& 1

& 3
& 1
& 1

& 3
\\

pmd
& 8.3
& 176
& 84
& 59

& 1423
& 162
& 123

& $\infty$
& 571
& 188

& 1818
& 192
& 76

& $\infty$
& 144
& 26

& 620
\\

sunflow
& 32.7
& 47
& 5
& 3

& 7
& 2
& 0

& 9
& 4
& 1

& 13
& 6
& 5

& 17
& 6
& 6

& 6
\\

tomcat
& 13.8
& 8
& 1
& 1

& 37
& 1
& 1

& 3
& 1
& 1

& 2
& 0
& 1

& 2
& 1
& 3

& 1
\\

tradebeans
& 45.5
& 0
& -1
& 1

& 1
& 1
& 2

& 5
& 3
& -1

& -1
& 1
& 2

& 3
& 1
& 5

& 2
\\

tradesoap
& 94.4
& 1
& 3
& 0

& 2
& 1
& 1

& 2
& 0
& 1

& 0
& 0
& 1

& 2
& 2
& 5

& 1
\\

xalan
& 20.3
& 4
& 2
& 2

& 27
& 7
& 2

& 10
& 5
& 2

& 3
& 2
& 3

& 4
& 4
& 3

& 4
\\

%

\hline
\hline

 (B)
 &
 & {TM}
 & {MOP}
 & {RV}
 & {TM}
 & {MOP}
 & {RV}
 & {TM}
 & {MOP}
 & {RV}
 & {TM}
 & {MOP}
 & {RV}
 & {TM}
 & {MOP}
 & {RV}
 & {RV}
\\

\hline


bloat
& 4.9
& 56.8
& 19.3
& 13.2
& 7.7
& 146.8
& 79.0
& $\infty$
& 173.4
& 56.1
& 6.8
& 127.9
& 48.3
& 6.9
& 55.4
& 12.7
& 340.9
\\





 jython
& 5.3
& 5.7
& 4.6
& 4.8
& 4.9
& 4.6
& 4.8
& 6.0
& 19.5
& 4.7
& 5.3
& 4.5
& 4.4
& 5.9
& 4.8
& 5.1
& 4.7
\\





\hline

 avrora
& 4.7
& 4.6
& 12.4
& 9.1
& 4.4
& 136.2
& 15.8
& $\infty$
& 14.7
& 8.5
& 4.3
& 28.0
& 12.6
& 4.4
& 13.0
& 4.9
& 22.3
\\

 batik
& 79.1
& 79.2
& 78.7
& 79.3
& 75.2
& 93.6
& 86.6
& $\infty$
& 91.2
& 79.6
& 78.2
& 93.2
& 85.1
& 79.9
& 86.9
& 76.7
& 104.3
\\

 eclipse
& 95.9
& 100.8
& 107.6
& 97.1
& 98.3
& 100.0
& 110.3
& 106.9
& 93.8
& 101.1
& 100.4
& 109.2
& 90.1
& $\infty$
& 98.6
& 98.7
& 98.9
\\

 fop
& 20.7
& 97.4
& 47.1
& 52.5
& 24.3
& 24.2
& 29.4
& $\infty$
& 69.2
& 28.1
& $\infty$
& 54.8
& 24.8
& $\infty$
& 55.9
& 25.2
& 47.5
\\

 h2
& 265.0
& 267.8
& 598.5
& 565.2
& 267.2
& 266.2
& 262.4
& 312.4
& 688.3
& 268.2
& 271.4
& 690.3
& 265.5
& 271.0
& 718.3
& 270.0
& 283.7
\\

 luindex
& 6.8
& 5.6
& 5.5
& 5.6
& 6.3
& 6.9
& 6.8
& 7.4
& 8.2
& 6.9
& 7.4
& 7.4
& 7.5
& 7.1
& 7.4
& 11.0
& 11.8
\\

 lusearch
& 4.6
& 4.7
& 4.4
& 4.8
& 4.6
& 4.8
& 4.2
& 4.0
& 4.3
& 4.8
& 4.5
& 4.5
& 4.6
& 4.6
& 4.8
& 4.7
& 4.7
\\

 pmd
& 18.0
& 56.9
& 59.8
& 48.5
& 17.2
& 146.3
& 86.4
& $\infty$
& 212.7
& 93.6
& 20.3
& 238.4
& 84.6
& $\infty$
& 117.1
& 32.9
& 420.0
\\

 sunflow
& 4.4
& 4.5
& 4.8
& 4.9
& 4.8
& 4.3
& 4.7
& 4.7
& 4.4
& 4.4
& 5.1
& 4.3
& 4.9
& 4.5
& 4.7
& 4.5 
& 4.6
\\

 tomcat
& 11.6
& 11.4
& 12.3
& 11.4
& 12.5
& 11.0
& 11.5
& 11.9
& 11.4
& 11.0
& 11.3
& 11.3
& 11.3
& 11.4
& 11.4
& 11.8
& 11.8
\\

 tradebeans
& 63.2
& 62.9
& 62.7
& 62.1
& 63.7
& 63.9
& 64.1
& 63.3
& 62.5
& 62.7
& 63.2
& 62.8
& 62.0
& 64.0
& 62.8
& 64.0
& 62.5
\\

 tradesoap
& 64.1
& 61.8
& 62.3
& 63.3
& 63.4
& 63.1
& 64.4
& 64.1
& 63.5
& 62.0
& 60.7
& 65.0
& 65.9
& 65.5
& 64.5
& 65.6
& 64.5
\\

 xalan
& 4.9
& 4.9
& 5.0
& 5.1
& 4.9
& 4.9
& 4.9
& 4.9
& 4.5
& 4.9
& 5.0
& 4.8
& 5.0
& 5.1
& 4.9
& 4.9
& 5.0
\\

%
%
%
%
%
%
%

\hline
\end{tabular}
}
\caption{Comparison of {Tracematches} (TM), {JavaMOP} (MOP), and {RV}: 
(A) average {\em percent} runtime overhead;
(B) total peak memory usage in MB. 
(convergence within 3\%, $\infty$: not terminated after 1 hour)
}
\label{table:performance_evaluation}
\end{center}
\end{figure*}

\begin{figure*}[t!]
\hspace{-2pt}
\begin{center}
\scalebox{0.6}{
\begin{tabular}{|l@{}
|@{\hspace{2pt}}rrrr@{\hspace{2pt}}
|@{\hspace{2pt}}rrrr@{\hspace{2pt}}
|@{\hspace{2pt}}rrrr@{\hspace{2pt}}
|@{\hspace{2pt}}rrrr@{\hspace{2pt}}
|@{\hspace{2pt}}rrrr@{\hspace{2pt}}
|}
\hline

 & \multicolumn{4}{@{\hspace{2pt}}c|@{\hspace{2pt}}}{\hasnext}
 & \multicolumn{4}{@{\hspace{2pt}}c|@{\hspace{2pt}}}{\unsafeiter}
 & \multicolumn{4}{@{\hspace{2pt}}c|@{\hspace{2pt}}}{\property{UnsafeMapIter}} 
 & \multicolumn{4}{@{\hspace{2pt}}c|@{\hspace{2pt}}}{\property{UnsafeSyncColl}} 
 & \multicolumn{4}{@{\hspace{2pt}}c@{\hspace{2pt}}|}{\property{UnsafeSyncMap}} \\

\hline

& E
& M
& FM
& CM
& E
& M
& FM
& CM
& E
& M
& FM
& CM
& E
& M
& FM
& CM
& E
& M
& FM
& CM
\\

\hline

bloat
& 155M 
& 1.9M 
& 1.9M 
& 1.8M 
& 81M 
& 1.9M 
& 1.8M 
& 1.6M 
& 74M 
& 3.6M 
& 43K 
& 3.4M 
& 143M 
& 4.1M 
& 0
& 3.7M 
& 161M 
& 3.4M 
& 0
& 3.4M 
\\

jython
& 106
& 50
& 47
& 26
& 179K 
& 50
& 38
& 38
& 179K 
& 101K 
& 94
& 101K 
& 156
& 100
& 0
& 83
& 256
& 150
& 0
& 122
\\

\hline

avrora
& 1.5M 
& 909K 
& 850K 
& 765K 
& 1.4M 
& 909K 
& 860K 
& 808K 
& 1.3M 
& 1.2M 
& 18 
& 1.2M 
& 2.4M 
& 1.8M 
& 0 
& 1.7M 
& 1.5M 
& 909K 
& 0 
& 904K 
\\

batik
& 49K
& 24K
& 21K
& 21K
& 125K
& 24K
& 21K
& 10K
& 55K
& 33K
& 140
& 27K
& 73K
& 50K
& 0
& 34K
& 50K
& 26K
& 0
& 26K
\\

eclipse
& 226K
& 7.6K
& 5.3K
& 2.9K
& 119K
& 6.6K
& 5.1K
& 2.6K
& 113K
& 22K
& 2.2K
& 7.8K
& 233K
& 15K
& 0
& 7.5K
& 241K
& 18K
& 0
& 9.2K
\\

fop
& 1.0M
& 184K
& 74K
& 151K
& 709K
& 7.7K
& 7.2K
& 1.8K
& 499K
& 177K
& 67
& 160K
& 1.2M
& 239K
& 0
& 217K
& 1.2M
& 231K
& 0
& 213K
\\

h2
& 27M
& 6.5M
& 6.0M
& 5.6M
& 12M
& 3.7K
& 3.3K
& 1.3K
& 12M
& 6.6M
& 9
& 6.5M
& 27M
& 6.5M
& 0
& 6.5M
& 27M
& 6.5M
& 0
& 6.5M
\\

luindex
& 371
& 66
& 40
& 2
& 4.4K
& 65
& 39
& 0
& 378
& 183
& 2
& 59
& 436
& 132
& 0
& 30
& 472
& 125
& 0
& 25
\\

lusearch
& 1.4K
& 131
& 196
& 114
& 748K
& 130
& 210
& 18
& 20K
& 944
& 338
& 1.4K
& 1.7K
& 262
& 0
& 402
& 1.8K
& 263
& 0
& 158
\\

pmd
& 8.3M
& 789K
& 694K
& 571K
& 6.4M
& 551K
& 473K
& 382K
& 4.3M
& 1.3M
& 110K
& 1.1M
& 8.8M
& 1.5M
& 0
& 1.3M
& 8.6M
& 1.1M
& 0
& 999K
\\

sunflow
& 2.7M
& 101K
& 101K
& 100K
& 1.3M
& 2
& 0
& 0
& 1.3M
& 83K
& 0
& 83K
& 2.7M
& 101K
& 0
& 101K
& 2.7M
& 101K
& 0
& 101K
\\

tomcat
& 25
& 6
& 0
& 0
& 132
& 4
& 0
& 0
& 68
& 26
& 0
& 0
& 29
& 10
& 0
& 0
& 33
& 12
& 0
& 0
\\

tradebeans
& 11
& 3
& 0
& 0
& 31
& 2
& 0
& 0
& 29
& 13
& 0
& 0
& 13
& 5
& 0
& 0
& 15
& 6
& 0
& 0
\\

tradesoap
& 11
& 3
& 0
& 0
& 31
& 2
& 0
& 0
& 29
& 13
& 0
& 0
& 13
& 5
& 0
& 0
& 15
& 6
& 0
& 0
\\

xalan
& 11
& 3
& 0
& 0
& 8.9K
& 2
& 0
& 0
& 119K
& 20K
& 0
& 20K
& 13
& 5
& 0
& 0
& 15
& 6
& 0
& 0
\\

%

\hline
\end{tabular}
}
\caption{Monitoring statistics: number of events (E), number of created monitors (M), 
number of flagged monitors (FM), number of collected monitors (CM).}
\label{table:stat}
\end{center}
\end{figure*}

\subsubsection{Results and Discussions}
\label{sec:results}

Figures~\ref{table:performance_evaluation} and~\ref{table:stat} summarize the
results of the evaluation.  Note that the structure of the DaCapo 9.12 allows
us to instrument all of the benchmarks plus all supplementary libraries that
the benchmarks use, which was not possible for DaCapo 2006-10.  Therefore,
{fop} and {pmd} show higher overheads than the benchmarks
using DaCapo 2006-10 from~\cite{chen-meredith-jin-rosu-2009-ase}.  While other
benchmarks show overheads less than 80\% in {JavaMOP},
{bloat}, {avrora}, and {pmd} show prohibitive
overhead in both runtime and memory performance. 
This is because they generate many iterators and all properties in
this evaluation are intended to monitor iterators. 
For example, {bloat} creates 1,625,770 collections and 941,466 iterators in total 
while 19,605 iterators coexist at the same time at peak, in an execution. 
{avrora} and {pmd} also create many collections and iterators.
Also, they call \textsf{\footnotesize hasNext()} 78,451,585 times, 1,158,152 times and 4,670,555 times
and \textsf{\footnotesize next()} 77,666,243 times, 352,697 times and 3,607,164 times, respectively.
Therefore, we mainly discuss those three examples in this section, 
although {RV} shows improvements for other examples as well.

Figure~\ref{table:performance_evaluation}~(A) shows the percent runtime
overhead of {Tracematches}, {JavaMOP}, and {RV}.
Overall, {RV} averages two times less runtime overhead than
{JavaMOP} and orders of magnitude less runtime overhead than
{Tracematches} (recall that these are the most optimized runtime
verification systems).
With {bloat}, {RV} shows less than 260\% runtime overhead for
each property, while {JavaMOP} always shows over 440\%
runtime overhead and {Tracematches} always shows over 1350\% for completed runs and
\emph{crashed} for \unsafemap.  With {avrora}, on average, {RV}
shows 62\% runtime overhead, while {JavaMOP} shows 139\% runtime
overhead and {Tracematches} shows 203\% and hangs for \unsafemap{}.
With {pmd}, on average, {RV} shows 94\% runtime overhead,
while {JavaMOP} shows 231\% runtime overhead and
{Tracematches} shows 1139\% and hangs for \unsafemap{} and
\unsafesyncmap{}.

Also, {RV} was tested with all five properties together and showed
982\%, 275\%, and 620\% overhead, respectively, which are still faster or
comparable to monitoring one of many properties alone in {JavaMOP} or
{Tracematches}.  The overhead for monitoring all the properties
simultaneously can be slightly larger than the sum of their individual
overheads since the additional memory pressure makes the JVM's garbage
collection behave differently.

Figure~\ref{table:performance_evaluation} (B) shows the peak memory usage of
the three systems.  {RV} has lower peak memory usage than
{JavaMOP} in most cases.  The cases where {RV} does not show
lower peak memory usage are within the limits of expected memory jitter.
However, memory usage of {RV} is still higher than the memory usage of
{Tracematches} in some cases.  {Tracematches} has several
finite automata specific memory optimizations \cite{oopsla07abc}, which cannot
be implemented in formalism-independent systems like {RV} and JavaMOP.
Although {Tracematches} is sometimes more memory efficient, it shows
prohibitive runtime overhead monitoring {bloat} and {pmd}.
There is a trade-off between memory usage and runtime overhead.  If
{RV} more actively removes terminated monitors, memory usage will be
lower, at the cost of runtime performance.  Overall, the monitor garbage collection
optimization in RV achieves the most efficient parametric monitoring system with
reasonable memory performance. 

Figure~\ref{table:stat} shows the number of triggered events, of created
monitors, of monitors flagged as unnecessary by RV's optimization, and
of monitors collected by the JVM.  Among the DaCapo examples, {bloat},
{avrora}, {h2}, {pmd} and {sunflow}
generated a very large number of events (millions) in all
properties, resulting in millions of monitors created in most cases.
{h2} does not exhibit large overhead because monitor instances in
{h2} have shorter lifetimes, therefore the created monitor instances
are not used heavily like in {bloat}.  {sunflow} has millions
of events but does not create as many monitor instances as as other benchmarks.
When monitoring the \hasnext{} and \unsafeiter{} properties, RV's garbage collector
effectively flagged monitors as unnecessary and most were collected by the JVM.

The experimental evaluation in this section shows that the approach to
parametric trace slicing and monitoring discussed in this paper is indeed feasible,
provided that it is not implemented naively.  Indeed, as seen in the tables in this section,
implementation optimizations make a huge difference in the runtime and memory overhead.
This paper was not dedicated to optimizations and implementations; its objective was to
only introduce the mathematical notions, notations, proofs and abstract algorithms underlying
the semantical foundation of parametric properties and their monitoring.  Current and future
implementations are and will build on this foundation, applying specific optimizations and
heuristics to reduce the runtime or the memory overhead caused by monitoring.

\section{Concluding Remarks, Future Work and Acknowledgments}
\label{sec:conclusion}

A semantic foundation for parametric traces, properties and
monitoring was proposed.  A parametric trace slicing
technique, which was discussed and proved correct, allows the
extraction of all the non-parametric trace slices from a parametric
slice by traversing the original trace only once and dispatching each
parametric event to its corresponding slices.  
It thus enables the leveraging of any non-parametric,
i.e., conventional, trace analysis techniques to the parametric case.
A parametric 
monitoring technique, also discussed and proved correct, makes use of it
to monitor arbitrary parametric properties against parametric
execution traces using and indexing ordinary monitors for the base,
non-parametric property.  Optimized implementations of the discussed
techniques in JavaMOP and RV reveal that their generality, compared to
the existing similar but ad hoc and limited techniques in current use,
does not come at a performance expense.  Moreover, further static analysis
optimizations like those in
\cite{pql-oopsla,bodden-chen-rosu-2009-aosd,clara,dwyer-purandare-person-2010-rv} may
significantly reduce the runtime and memory overheads of monitoring parametric properties
based on the techniques and algorithms discussed in this paper.

The parametric trace slicing technique in Section
\ref{sec:trace-slicing} enables the leveraging of any non-parametric,
i.e., conventional, trace analysis techniques to the parametric case.
We have only considered monitoring in this paper.  Another interesting
and potentially rewarding use of our technique could be in the context
of property mining. For example, one could run the trace slicing
algorithm on large benchmarks making intensive use of library classes,
and then, on the obtained trace slices corresponding to particular
classes or groups of classes of interest, run property mining
algorithms.  The mined properties, or the lack thereof, may provide
insightful formal documentation for libraries, or even detect errors.
Preliminary steps in this direction are reported in \cite{lee-chen-rosu-2011-icse}.

\subsection*{Acknowledgments}

We would like to warmly thank the other members of the MOP team who
contributed to the implementation of the new and old JavaMOP system,
as well as of extensions of it, namely to Dennis Griffith, Dongyun
Jin, Choonghwan Lee and Patrick Meredith.  We are also grateful to
Klaus Havelund, who found several errors in our new JavaMOP
implementation while using it in teaching a course at Caltech, and who
recommended us several simplifications in its user interface.
We are also grateful to Matt Dwyer and to Tewfik Bultan for using
JavaMOP in their classes at the Universities of Nebraska and of California,
respectively.  We express our thanks also to Eric Bodden for his lead of the
static analysis optimization efforts in \cite{bodden-chen-rosu-2009-aosd},
and to the Tracematches \cite{tracematches-oopsla,oopsla07abc},
PQL \cite{pql-oopsla}, Eagle \cite{DBLP:conf/vmcai/BarringerGHS04} and
RuleR \cite{DBLP:conf/rv/BarringerRH07} teams for inspiring debates
and discussions.
The research presented in this paper was generously funded by the 
NSF grants NSF CCF-0916893, NSF CNS-0720512, and NSF CCF-0448501,
by NASA grant NASA-NNL08AA23C, by a Samsung SAIT grant and by several
Microsoft gifts and UIUC research board awards.

Sadly, the second author, Feng Chen, passed away on August 8, 2009,
in the middle of this project, due to an undetected blood clot.
Feng was the main developer of JavaMOP and a co-inventor of most of
its underlying techniques and algorithms, including those in this paper.
His results and legacy will outlive him.  May his soul rest in peace.

\bibliographystyle{abbrv}

\bibliography{citations}

\begin{thebibliography}{10}

\bibitem{tracematches-oopsla}
C.~Allan, P.~Avgustinov, A.~S. Christensen, L.~Hendren, S.~Kuzins, O.~Lhotak,
  O.~de~Moor, D.~Sereni, G.~Sittampalam, and J.~Tibble.
\newblock Adding trace matching with free variables to {A}spect{J}.
\newblock In {\em OOPSLA'05}, 2005.

\bibitem{aspectj}
Aspect{J}.
\newblock http://eclipse.org/aspectj/.

\bibitem{abc-05}
P.~Avgustinov, A.~S. Christensen, L.~Hendren, S.~Kuzins, J.~Lhotak, O.~Lhotak,
  O.~de~Moor, D.~Sereni, G.~Sittampalam, and J.~Tibble.
\newblock {A}{B}{C}: an extensible {A}spect{J} compiler.
\newblock In {\em Aspect-Oriented Software Development (AOSD'05)}, pages
  87--98. ACM, 2005.

\bibitem{oopsla07abc}
P.~Avgustinov, J.~Tibble, and O.~de~Moor.
\newblock Making trace monitoring feasible.
\newblock In R.~P. Gabriel, editor, {\em OOPSLA'07}. ACM, 2007.

\bibitem{DBLP:conf/vmcai/BarringerGHS04}
H.~Barringer, A.~Goldberg, K.~Havelund, and K.~Sen.
\newblock Rule-based runtime verification.
\newblock In {\em VMCAI}, volume 2937 of {\em LNCS}, pages 44--57, 2004.

\bibitem{DBLP:conf/rv/BarringerRH07}
H.~Barringer, D.~E. Rydeheard, and K.~Havelund.
\newblock Rule systems for run-time monitoring: From {Eagle} to {RuleR}.
\newblock In {\em Runtime Verification (RV'07)}, volume 4839 of {\em LNCS},
  pages 111--125, 2007.

\bibitem{bauer-leucker-schallhart-2010-tosem}
A.~Bauer, M.~Leucker, and C.~Schallhart.
\newblock Runtime verification for {LTL} and {TLTL}.
\newblock {\em ACM Transactions on Software Engineering and Methodology}, 20,
  2011.

\bibitem{DaCapo:paper}
S.~M. Blackburn, R.~Garner, C.~Hoffman, A.~M. Khan, K.~S. McKinley, R.~Bentzur,
  A.~Diwan, D.~Feinberg, D.~Frampton, S.~Z. Guyer, M.~Hirzel, A.~Hosking,
  M.~Jump, H.~Lee, J.~E.~B. Moss, A.~Phansalkar, D.~Stefanovi\'{c},
  T.~{VanDrunen}, D.~von Dincklage, and B.~Wiedermann.
\newblock The {DaCapo} benchmarks: {J}ava benchmarking development and
  analysis.
\newblock In {\em OOPSLA'06}, pages 169--190. ACM, 2006.

\bibitem{jlo}
E.~Bodden.
\newblock J-lo, a tool for runtime-checking temporal assertions.
\newblock Master's thesis, RWTH Aachen University, 2005.

\bibitem{bodden-chen-rosu-2009-aosd}
E.~Bodden, F.~Chen, and G.~Ro\c{s}u.
\newblock Dependent advice: A general approach to optimizing history-based
  aspects.
\newblock In {\em AOSD'09}, pages 3--14. ACM, 2009.

\bibitem{tm-static-ecoop}
E.~Bodden, L.~Hendren, and O.~Lhot{\'{a}}k.
\newblock A staged static program analysis to improve the performance of
  runtime monitoring.
\newblock In {\em European Conference on Object Oriented Programming
  (ECOOP'07)}, volume 4609 of {\em LNCS}, pages 525--549. Springer, 2007.

\bibitem{clara}
E.~Bodden, P.~Lam, and L.~Hendren.
\newblock {C}lara: A framework for partially evaluating finite-state runtime
  monitors ahead of time.
\newblock In {\em Runtime Verification (RV'10)}, volume 6418 of {\em LNCS},
  pages 183--197. Springer, 2010.

\bibitem{chen-meredith-jin-rosu-2009-ase}
F.~Chen, P.~Meredith, D.~Jin, and G.~Rosu.
\newblock Efficient formalism-independent monitoring of parametric properties.
\newblock In {\em ASE'09}. IEEE/ACM, 2009.

\bibitem{chen-rosu-2007-oopsla}
F.~Chen and G.~Ro\c{s}u.
\newblock {MOP: An Efficient and Generic Runtime Verification Framework}.
\newblock In {\em OOPSLA'07}, pages 569--588. ACM, 2007.

\bibitem{chen-rosu-2008-tr-a}
F.~Chen and G.~Ro\c{s}u.
\newblock {Mining Parametric State-Based Specifications from Executions}.
\newblock Technical Report UIUCDCS-R-2008-3000, Dept. of Computer Science at
  UIUC, 2008.

\bibitem{chen-rosu-2009-tacas}
F.~Chen and G.~Ro\c{s}u.
\newblock Parametric trace slicing and monitoring.
\newblock In {\em TACAS'09}, volume 5505 of {\em LNCS}, pages 246--261, 2009.

\bibitem{damm01lscs}
W.~Damm and D.~Harel.
\newblock {LSCs}: Breathing life into message sequence charts.
\newblock {\em Formal Methods in System Design}, 19(1):45--80, 2001.

\bibitem{dwyer-purandare-person-2010-rv}
M.~Dwyer, R.~Purandare, and S.~Person.
\newblock Runtime verification in context: Can optimizing error detection
  improve fault diagnosis.
\newblock In {\em Runtime Verification (RV'10)}, volume 6418 of {\em LNCS},
  pages 36--50. Springer, 2010.

\bibitem{ptql-oopsla}
S.~Goldsmith, R.~O'Callahan, and A.~Aiken.
\newblock Relational queries over program traces.
\newblock In {\em OOPSLA'05}, 2005.

\bibitem{jin-meredith-griffith-rosu-2011-pldi}
D.~Jin, P.~O. Meredith, D.~Griffith, and G.~Ro\c{s}u.
\newblock Garbage collection for monitoring parametric properties.
\newblock In {\em Programming Language Design and Implementation (PLDI'11)},
  pages 415--424. ACM, 2011.

\bibitem{kupferman-vardi-2001}
O.~Kupferman and M.~Y.~Vardi.
\newblock Model checking of safety properties.
\newblock {\em Form. Methods Syst. Des.}, 19:291--314, October 2001.

\bibitem{lee-chen-rosu-2011-icse}
C.~Lee, F.~Chen, and G.~Ro\c{s}u.
\newblock Mining parametric specifications.
\newblock In {\em Proceeding of the 33rd International Conference on Software
  Engineering (ICSE'11)}, pages 591--600. ACM, 2011.

\bibitem{manna-pnueli-1992}
Z.~Manna and A.~Pnueli.
\newblock {\em The temporal logic of reactive and concurrent systems}.
\newblock Springer-Verlag New York, Inc., New York, NY, USA, 1992.

\bibitem{lsc-monitor}
S.~Maoz and D.~Harel.
\newblock From multi-modal scenarios to code: compiling lscs into aspectj.
\newblock In {\em FSE'06}, pages 219--230, 2006.

\bibitem{pql-oopsla}
M.~Martin, V.~B. Livshits, and M.~S. Lam.
\newblock Finding application errors and security flaws using {PQL}: a program
  query language.
\newblock In {\em OOPSLA'05}, 2005.

\bibitem{meredith-jin-chen-rosu-2008-ase}
P.~Meredith, D.~Jin, F.~Chen, and G.~Ro\c{s}u.
\newblock Efficient monitoring of parametric context-free patterns.
\newblock In {\em ASE'08}. IEEE/ACM, 2008.

\bibitem{meredith-rosu-2010-rv}
P.~Meredith and G.~Ro\c{s}u.
\newblock Runtime verification with the {RV} system.
\newblock In {\em First International Conference on Runtime Verification
  (RV'10)}, volume 6418 of {\em Lecture Notes in Computer Science}, pages
  136--152. Springer, 2010.

\bibitem{meredith-jin-griffith-chen-rosu-2010-jsttt}
P.~O. Meredith, D.~Jin, D.~Griffith, F.~Chen, and G.~Ro\c{s}u.
\newblock An overview of the mop runtime verification framework.
\newblock {\em Journal on Software Tools for Technology Transfer (J. of STTT)},
  2010.
\newblock To appear.

\bibitem{moore-56}
E.~F. Moore.
\newblock Gedanken-experiments on sequential machines.
\newblock {\em Automata Studies, Annals of Mathematical Studies}, 34:129--153,
  1956.

\bibitem{rosu-chen-2008-tr-a}
G.~Ro\c{s}u and F.~Chen.
\newblock {Parametric Trace Slicing and Monitoring}.
\newblock Technical Report UIUCDCS-R-2008-2977, University of Illinois at
  Urbana-Champaign, 2008.

\bibitem{rosu-chen-ball-2008-rv}
G.~Ro{\c s}u, F.~Chen, and T.~Ball.
\newblock Synthesizing monitors for safety properties -- this time with calls
  and returns --.
\newblock In {\em Runtime Verification (RV'08)}, volume 5289 of {\em LNCS},
  pages 51--68, 2008.

\bibitem{rosu-havelund-2005-jase}
G.~Rosu and K.~Havelund.
\newblock Rewriting-based techniques for runtime verification.
\newblock {\em Automated Software Engineering}, 12(2):151--197, 2005.

\bibitem{stolz-2006-rv}
V.~Stolz.
\newblock Temporal assertions with parameterized propositions.
\newblock In O.~Sokolsky and S.~Tasiran, editors, {\em Runtime Verification},
  volume 4839 of {\em Lecture Notes in Computer Science}, pages 176--187.
  Springer Berlin / Heidelberg, 2007.

\bibitem{strom-yemeni-1986-tse}
R.~E. Strom and S.~Yemeni.
\newblock Typestate: A programming language concept for enhancing software
  reliability.
\newblock {\em IEEE Transactions on Software Engineering}, 12:157--171, January
  1986.

\bibitem{tabakov-vardi-2010-rv}
D.~Tabakov and M.~Vardi.
\newblock Optimized temporal monitors for systemc.
\newblock In {\em First International Conference on Runtime Verification
  (RV'10)}, volume 6418 of {\em Lecture Notes in Computer Science}. Springer,
  2010.

\bibitem{tm-bench-march-2008}
{Tracematches Benchmarks}.
\newblock http://abc.comlab.ox.ac.uk/tmahead.

\end{thebibliography}

\end{document}